%% file: main.tex
\documentclass[thm-restate,a4paper,USenglish,cleveref, autoref, thm-restate]{lipics-v2021}

\newcommand{\papertitle}{Deciding Termination of Simple Randomized Loops}
\title{\papertitle}

\author{Éléanore Meyer}{RWTH Aachen University, Germany \and \url{https://verify.rwth-aachen.de/emeyer/}}{eleanore.meyer@cs.rwth-aachen.de}{https://orcid.org/0000-0003-1038-4944}{}%

\author{Jürgen Giesl}{RWTH Aachen University, Germany \and \url{https://verify.rwth-aachen.de/giesl/}}{giesl@informatik.rwth-aachen.de}{https://orcid.org/0000-0003-0283-8520}{}

\authorrunning{E.\ Meyer and J.\ Giesl} %

\Copyright{Éléanore Meyer and Jürgen Giesl} %



\category{} %

\relatedversion{\emph{Conference version}: \cite{mfcs25paper}}
\supplement{}
\supplementdetails[subcategory={Source Code}, cite={}, swhid={}]{Software}{https://github.com/aprove-developers/SiRop}

\funding{funded by the Deutsche Forschungsgemeinschaft (DFG, German Research Foundation) - DFG Research Training Group 2236 UnRAVeL}

\hideLIPIcs{}
\nolinenumbers %

\newcommand{\paper}[1]{}
\newcommand{\report}[1]{#1}

\usepackage{mathtools}
\usepackage{nicefrac}
\usepackage{esvect}

\usepackage[dvipsnames]{xcolor}
\usepackage{tikz}
\usetikzlibrary{calc,decorations.pathreplacing,patterns}

\usepackage{xifthen}

\hypersetup{
  pdftitle={\papertitle}, citecolor=olive, %
}

\input{./commands.tex}

\crefname{enumi}{Item}{Items}
\Crefname{enumi}{Item}{Items}
\crefname{equation}{}{}
\Crefname{equation}{}{}
\crefname{figure}{Fig.}{Fig.}
\Crefname{figure}{Fig.}{Fig.}
\crefname{section}{Sect.}{Sect.}
\Crefname{section}{Sect.}{Sect.}
\crefname{appendix}{App.}{App.}
\Crefname{appendix}{App.}{App.}
\crefname{theorem}{Thm.}{Thm.}
\Crefname{theorem}{Thm.}{Thm.}
\crefname{definition}{Def.}{Def.}
\Crefname{definition}{Def.}{Def.}
\crefname{example}{Ex.}{Ex.}
\Crefname{example}{Ex.}{Ex.}
\crefname{corollary}{Cor.}{Cor.}
\Crefname{corollary}{Cor.}{Cor.}

\begin{document}
\maketitle

\begin{abstract}
  We show that universal positive almost sure termination (\textsf{UPAST}) is decidable for a class of simple randomized programs, i.e., it is decidable whether the expected runtime of such a program is finite for all inputs.
  Our class contains all programs that consist of a single loop, with a linear loop guard and a loop body composed of two linear commuting and diagonalizable updates.
  In each iteration of the loop, the update to be carried out is picked at random, according to a fixed probability.
  We show the decidability of \textsf{UPAST}
  for this class of programs, where the program's variables and inputs may range over various sub-semirings of the real numbers.
  In this way, we extend a line of research initiated by Tiwari in 2004 into the realm of randomized programs.
\end{abstract}

\input{./introduction.tex}

\input{./programs_and_termination.tex}

\input{./on_constraint_terms.tex}

\input{./positive_eigenvalues.tex}

\input{./towards_ent_witnesses.tex}

\input{./deciding_past.tex}

\input{./conclusion.tex}

\bibliography{references}

\makeatletter \newcommand{\proofsForSection}[1]{
  \section[Proofs for \cref@section@name~\ref{#1}]{Proofs for \cref{#1}}
}
\makeatother \report{
  \appendix{}
  \label{sec:appendix}
  \input{./probability_theory_app.tex}
  \input{./linear_algebra_app.tex}

  \input{./programs_and_termination_app.tex}
  \input{./on_constraint_terms_app.tex}
  \input{./positive_eigenvalues_app.tex}
  \input{./towards_ent_witnesses_app.tex}
  \input{./deciding_past_app.tex}}

\end{document}

%% file: commands.tex
\makeatletter
\pgfdeclarepatternformonly[\LineSpace]{my north east lines}{\pgfqpoint{-1pt}{-1pt}}{\pgfqpoint{\LineSpace}{\LineSpace}}{\pgfqpoint{\LineSpace}{\LineSpace}}%
{
  \pgfsetcolor{\tikz@pattern@color}
  \pgfsetlinewidth{0.4pt}
  \pgfpathmoveto{\pgfqpoint{0pt}{0pt}}
  \pgfpathlineto{\pgfqpoint{\LineSpace + 0.1pt}{\LineSpace + 0.1pt}}
  \pgfusepath{stroke}
}
\makeatother

\newdimen\LineSpace{}
\tikzset{
  line space/.code={\LineSpace=#1}, line space=3pt
}

\renewcommand{\AA}{\mathbb{A}}
\newcommand{\CC}{\mathbb{C}}
\newcommand{\NN}{\mathbb{N}}
\newcommand{\QQ}{\mathbb{Q}}
\newcommand{\QQbar}{\mathbb{\overline{Q}}}
\newcommand{\RR}{\mathbb{R}}
\newcommand{\ZZ}{\mathbb{Z}}
\newcommand{\im}{\mathrm{i}}
\newcommand{\programExpr}[5]{\normalfont{\texttt{while }}#1#5 > \veczero\colon #5 \leftarrow #2#5 \,\oplus_{#3}\, #4#5}

\newcommand{\program}{\mathcal{P}}
\newcommand*{\rom}[1]{\expandafter\@slowromancap\romannumeral #1@}

\newcommand{\run}{\mathfrak{r}}
\newcommand{\Runs}{\normalfont{\textsf{Runs}}}
\newcommand{\Path}{\normalfont{\textsf{Path}}}

\newcommand{\eps}{\varepsilon}

\DeclareMathOperator{\T}{T}
\DeclareMathOperator{\diag}{diag}
\DeclareMathOperator{\ind}{\mathbf{1}}

\newcommand{\scaleinside}{\mathfrak{i}}
\newcommand{\scaleoutside}{\mathfrak{o}}
\newcommand{\constrainttermgroup}{\mathfrak{D}}
\newcommand{\nmax}{\mathfrak{n}}
\newcommand{\pmax}{\mathfrak{p}}

\newcommand{\LRVar}{\mathcal{L}}
\newcommand{\Pre}{\mathrm{Pre}}
\newcommand{\F}{\mathcal{F}}
\newcommand{\URVar}{\mathcal{U}}

\renewcommand{\P}{\operatorname{\mathbb P}}
\DeclareMathOperator{\E}{\mathbb E}
\newcommand{\Val}{\mathcal{V}al}

\newcommand{\mat}[1]{{\mathbf{#1}}}
\newcommand{\matA}{{\mat{A}}}
\newcommand{\matAD}{{\mat{A_{D}}}}
\newcommand{\matB}{{\mat{B}}}
\newcommand{\matBD}{{\mat{B_{D}}}}
\newcommand{\matC}{{\mat{C}}}
\newcommand{\matM}{{\mat{M}}}
\newcommand{\matS}{{\mat{S}}}
\renewcommand{\vec}[1]{\vv{#1}}
\newcommand{\vecx}{\vec{x}}
\newcommand{\vecy}{\vec{y}}
\newcommand{\vecz}{\vec{z}}
\newcommand{\veczero}{\vec{0}}

\newcommand{\NT}{\mathrm{NT}}
\newcommand{\ENT}{\mathrm{ENT}}
\newcommand{\semiring}{\mathcal{S}}

\newcommand{\symMat}[1]{#1}

\newcommand{\toolname}{SiRop} %

\newcommand{\conj}[1]{\overline{#1}}

\newcommand{\abs}[1]{\left|{#1}\right|}

\renewcommand{\emptyset}{\varnothing}
\renewcommand{\Re}{\operatorname{Re}}

\newcommand{\CEqZero}{\mathcal{C}^{=0}}
\newcommand{\CGTZero}{\mathcal{C}^{>0}}

\DeclareMathOperator{\Vandermonde}{V}

\DeclareMathOperator{\sign}{sign}

\newcommand{\oldcom}[1]{}

%% file: introduction.tex
\section{Introduction}
\label{sect:Introduction}
We consider the problem of universal positive almost sure termination (\textsf{UPAST}), i.e., deciding whether a given randomized program has finite expected runtime on all inputs \cite{bournez2005proving,DBLP:conf/mfcs/Saheb-Djahromi78}.
This is a stronger property than universal almost sure termination (\textsf{UAST}) which requires that the probability of termination is $1$.
Our programs are simple randomized loops of the form
\begin{equation}
  \label{loop-intro}
  \programExpr{\matC}{\matA}{p}{\matB}{\vecx}
\end{equation}
Here, $\vecx = (x_{1},\ldots,x_{n})$ denotes the vector of program variables that range over a semiring $\semiring \subseteq \RR$, and $\matC \in \RR^{m \times n}$ is a matrix representing the loop guard with $m$ linear constraints over the program variables.
In each execution of the loop body, a matrix is chosen among $\matA,\matB \in \semiring^{n \times n}$ according to the probability $p \in [0,1]$ and the value $\vecx$ is updated accordingly.

\medskip

\emph{Our Contribution}: We show that \textsf{UPAST}
is decidable for all $\semiring \in \{ \ZZ, \QQ, \AA \}$ when limited to loops with diagonalizable commuting matrices $\matA$ and $\matB$, where $\AA$ is the set of algebraic real numbers.\footnote{\label{fn:whyAlgebraic}Our approach only considers \emph{algebraic}
  real $p$, $\matA$, $\matB$, and $\matC$, as it is not possible to represent arbitrary real numbers on computers.
  However, in \cref{sect:Deciding PAST} we will see that such a loop terminates for all algebraic real inputs iff it terminates for all real inputs.
}
Thus, we extend previous results on the termination of linear and affine\footnote{In an affine (or non-homogeneous) loop, the guard may have the form $\matC \vecx > \vec{c}$ and the update may have the form $\vecx \leftarrow \matA \vecx + \vec{a}$ for arbitrary vectors $\vec{c}$ and $\vec{a}$.}
non-randomized loops to the randomized setting.
In addition to deciding universal termination, our approach can compute a non-termination witness $\vecx \in \semiring^{n}$, i.e., if the loop is non-terminating, then $\vecx$ is an input leading to an infinite expected runtime.

Our programs go beyond single path loops as we might have $\matA \neq \matB$.
Thus, for every $k \in \NN$,\linebreak
there is not just a single execution of length $k$ but one has a ``range'' of possible executions of length $k$ where each execution occurs according to a known probability.
To ensure tractability of the resulting problem we require commutativity of both updates, so that we can focus on how often each update has been selected in an execution, but we do not have to take the $2^{k}$ different orders into account in which the two updates might have been chosen.
Moreover, we require diagonalizability to obtain closed forms of a certain shape, which allows us to analyze the behavior of a ``range'' of different executions at once.
To demonstrate the practical applicability of our decision procedure and the computation of non-termination witnesses, we provide a prototype implementation for the case $\semiring = \AA$ with our tool \textsf{\toolname}.

\medskip

\emph{Related Work}: We continue a line of research started in 2004 by Tiwari~\cite{tiwari04} who showed decidability of universal termination for loops with an affine guard and an affine update as its body, where the guard, updates, and inputs range over the real numbers.
In his proof, Tiwari reduced the affine to the linear case.
In 2006, Braverman~\cite{braverman06} proved that the problem remains decidable for loops and inputs ranging over the rational numbers $\QQ$, and if the guard and update are linear, then he also showed decidability over the integers $\ZZ$.
Similar to Tiwari, Braverman also reduced the affine case for $\QQ$ to the linear case.
In 2015, Ouaknine et al.\ proved~\cite{ouaknine15} that the affine case is decidable over the integers $\ZZ$ whenever the update is of the form $\vecx \leftarrow \matA \vecx + \vec{a}$, provided that $\matA \in \ZZ^{n \times n}$ is diagonalizable.
This restriction was removed by Hosseini et al.\ in 2019~\cite{hosseini19}.
In a related line of work, we proved decidability of universal termination over the integers $\ZZ$ for triangular affine loops, i.e., where the matrix $\matA \in \ZZ^{n \times n}$ is triangular~\cite{floriantriangularinitial}.
Later, we extended these results to triangular weakly non-linear loops which extend triangular loops by allowing certain non-linear updates~\cite{hark20,hark23}.

The only decidability results for termination of randomized programs that we are aware of consider probabilistic vector addition systems~\cite{brazdilpVASS}
or constant probability programs~\cite{constantprobability}, i.e., loops whose guards consist of only one affine inequation and whose bodies consist of several probabilistic branches (with fixed probabilities) that may increase or decrease the program variables by fixed constants.
The programs in~\cite{brazdilpVASS,constantprobability} are orthogonal to the ones considered in our approach as they only allow to modify the program variables by adding constants, but do not allow for multiplication.
Another related area of research~\cite{prob-solv1,prob-solv3,prob-solv2} deals with prob-solvable loops and moment invariants.
Given such a loop, these techniques can compute closed forms for all moments of program variables for a given iteration of the loop\linebreak
and, by taking a limit, also upon the loop's termination.
Thus, if restricted to almost surely terminating programs, they can decide \textsf{UPAST}.
However, in contrast to our method, these techniques require that all variables in the loop guard may only take finitely many values.
Moreover, there are many automated approaches for tackling \textsf{UPAST} using so-called \underline{r}anking \underline{s}uper\underline{m}artingales (RSM), e.g.,~\cite{lexrsm,ecoimp,cegispro2,program-analysis-with-martingales,termination-analysis-positivstellensatz,probkoat,amber,absynth,hoffmanncostanalysistypes}.
To generate a suitable RSM, one often uses techniques based on affine or polynomial templates, which renders the approach incomplete.
In~\cite{hardness-probabilistic-termination}, the authors showed that deciding \textsf{UPAST} is harder than deciding universal termination for non-randomized programs in terms of the arithmetic hierarchy.

\medskip

\emph{Outline of our Approach}:\report{ While we assume familiarity with basics from probability theory and linear algebra, we recapitulate some main concepts in \cref{sect:Preliminaries from Probability Theory,sec:preliminaries_from_linear_algebra}.}\paper{ We assume familiarity with basics from probability theory and linear algebra (we recapitulate some main concepts in \cite{arxiv}).}
\Cref{sect:programsTermination} formally introduces simple randomized loops, gives their semantics in terms of a probability space, and presents most of the definitions used throughout the paper.
For a loop as in \Cref{loop-intro}, we consider (finite) executions $f$ corresponding to words over the alphabet $\{\symMat{A},\symMat{B}\}$, where the $i$-th symbol in $f$ indicates which update matrix was used in the $i$-th application of the assignment $\vecx \leftarrow \matA\vecx \,\oplus_{p}\, \matB\vecx$.
For such executions $f$, $|f|_{\symMat{A}}$ and $|f|_{\symMat{B}}$ denote the number of $\symMat{A}$- and $\symMat{B}$-symbols in $f$, respectively.
Moreover, we introduce the function $\Val_{\vecx}$ that maps finite executions $f$ to the values $\Val_{\vecx}(f) \in \RR^{m}$ of the constraints in the loop guard after executing $f$ on a concrete input $\vecx \in \RR^{n}$, i.e., $\Val_{\vecx}(f) = \matC \cdot \matA^{|f|_{\symMat{A}}} \cdot \matB^{|f|_{\symMat{B}}} \cdot \vecx$, since $\matA$ and $\matB$ commute.
Our decision procedure does not search for a non-terminating input directly, but for an \emph{eventually} non-terminating input $\vecx$.
An input $\vecx$ is eventually non-terminating if by repeated execution of the loop body on $\vecx$ (while ignoring the guard), a non-terminating input can eventually be reached.
In \cref{lem:nt_vs_ent}, we show that a loop has an eventually non-terminating input iff it also has a non-terminating input.
Later (in \cref{sect:Deciding PAST}) we will show how to lift such an eventually non-terminating input to an actual non-terminating input.

In \cref{sect:On Constraint Terms}, we introduce a mapping $\URVar$ that maps executions $f$ to the difference between the relative number $\tfrac{|f|_{\symMat{A}}}{|f|}$ of times that the update matrix $\matA$ has been chosen in $f$ and the probability $p$ of choosing $\matA$, i.e., $\URVar(f) =\tfrac{|f|_{\symMat{A}}}{|f|} - p$.
Moreover, we essentially partition the set of indices $\{1, \ldots, n\}$ of all program variables into suitable sets $\constrainttermgroup_{(\scaleinside,\scaleoutside)}$ with $(\scaleinside,\scaleoutside) \in \mathcal{I} $ for some finite set $\mathcal{I} \subsetneq \RR_{>0}^{2}$.
We will then show that for all $c \in \{1,\ldots,m\}$, and all executions $f$ with $|f|_{\symMat{A}},|f|_{\symMat{B}} \geq 1$, the value ${(\Val_{\vecx}(f))}_{c}$ of the $c$-th constraint after executing $f$ on $\vecx$ is
\begin{equation}
  \textstyle
  {(\Val_{\vecx}(f))}_{c} = \sum_{(\scaleinside,\scaleoutside) \in
    \mathcal{I}} \; {\left(\scaleinside \cdot
  \scaleoutside^{\URVar(f)}\right)}^{|f|} \; \sum_{i \in
    \constrainttermgroup_{(\scaleinside,\scaleoutside)}} \;
  \zeta_{i,\matA}^{|f|_{\symMat{A}}} \, \zeta_{i,\matB}^{|f|_{\symMat{B}}} \,
  \gamma_{c,i} (\vecx).
  \label{eq:introduction_constraint_term_groups}
\end{equation}
Here, all $\zeta_{i,\matA}, \zeta_{i,\matB}$ are complex numbers of modulus $1$, i.e., $|\zeta_{i,\matA}| = |\zeta_{i,\matB}| = 1$ for all $i \in \{1,\ldots,\linebreak
  n\}$, and the functions $\gamma_{c,i}$ are linear maps $\RR^{n} \to \CC$. The maps $\gamma_{c,i}$ and the values $\zeta_{i,\matA}, \zeta_{i,\matB} \in \CC$ only depend on the matrices $\matC$, $\matA$, and $\matB$, but not on the specific input $\vecx$. While weaker requirements would suffice to ensure that ${(\Val_{\vecx}(f))}_{c}$ has some closed form, diagonalizability of $\mat{A}$ and $\mat{B}$ guarantees that it has the form \eqref{eq:introduction_constraint_term_groups}, which is crucial for our procedure. \Cref{lem:comparison_of_constraint_term_groups} shows that by a lexicographic comparison of those $(\scaleinside,\scaleoutside) \in \mathcal{I}$ for which the inner sum in \cref{eq:introduction_constraint_term_groups}
is not $0$ for all executions $f$, one can compute which of the pairs $(\scaleinside,\scaleoutside)$ is the ``dominant'' one.
Here, a pair $(\scaleinside,\scaleoutside) \in \mathcal{I}$ is considered dominant whenever the value of the first factor ${\left(\scaleinside \cdot \scaleoutside^{\URVar(f)}\right)}^{|f|}$ of \cref{eq:introduction_constraint_term_groups} grows the fastest if the execution of $f$ is continued (i.e., if $|f| \rightarrow \infty$) and the corresponding inner sum is not $0$ for all executions $f$.
The dominant pair depends on the specific input $\vecx$ and on whether $\URVar(f)$ is $\mathfrak{p}$ositive or $\mathfrak{n}$egative, and correspondingly, one has to use different lexicographic comparisons to determine the dominant pair.
For $d\in\{\nmax,\pmax\}$, let $\constrainttermgroup_{d,c,\vecx}$ denote the set $\constrainttermgroup_{(\scaleinside,\scaleoutside)}$ where the pair $(\scaleinside,\scaleoutside)$ is dominant for input $\vecx$ and $c \in \{1,\ldots,m\}$ (and positive $\URVar(f)$ if $d = \pmax$ or negative $\URVar(f)$ for $d = \nmax$).
Then, the sign of the ``coefficient'' $v(f) = \sum_{i \in \constrainttermgroup_{d,c,\vecx}}
  \!\!
  \zeta_{i,\matA}^{|f|_{\symMat{A}}} \zeta_{i,\matB}^{|f|_{\symMat{B}}} \gamma_{c,i}
  (\vecx)$ of the dominant pair eventually determines the sign of ${(\Val_{\vecx}(f))}_{c}$, provided that $|v(f)|$ is large enough (\Cref{lem:domination_of_eventually_dominating_constraint_term_groups}).

In \cref{sect:Positive Eigenvalues}, we consider the rearrangement
\begin{equation}
  \textstyle
  v(f) \; = \; \underbrace{\textstyle \sum_{i \in \mathfrak{R}_{d,c,\vecx}} \gamma_{c,i}
    (\vecx)}_{= R} \; + \;
  \sum_{i \in \mathfrak{C}_{d,c,\vecx}} \zeta_{i,\matA}^{|f|_{\symMat{A}}} \,
  \zeta_{i,\matB}^{|f|_{\symMat{B}}} \, \gamma_{c,i}
  (\vecx) \label{eq:introduction_constraint_term_groups_inner_split}
\end{equation}
where $\mathfrak{R}_{d,c,\vecx} = \{ i \in \constrainttermgroup_{d,c,\vecx} \mid \zeta_{i,\matA} = \zeta_{i,\matB} = 1 \}$ and $\mathfrak{C}_{d,c,\vecx} = \{ i \in \constrainttermgroup_{d,c,\vecx} \mid \{\zeta_{i,\matA},\zeta_{i,\matB}\} \neq \{1\} \}$.
\cref{lem:dual_positive_eigenvalues_for_eventually_dominating_constraints} then gives a necessary condition for non-termination (and hence a sufficient condition for universal termination): If $\vecx$ is an eventually non-terminating input, then there must be a $d \in \{\nmax,\pmax\}$ such that we have $R > 0$ for all constraints $c$, with $R$ as in \cref{eq:introduction_constraint_term_groups_inner_split}.

\Cref{sec:towards_witnesses_for_non-termination} turns this necessary condition for non-termination into a sufficient condition.
To that end, we define the set $W$ of witnesses for eventual non-termination containing all inputs $\vecx \in \semiring^{n}$ for which there is some $d \in \{\nmax,\pmax\}$ such that $R > \sum_{i \in \mathfrak{C}_{d,c,\vecx}} |\gamma_{c,i} (\vecx)|$ holds for all constraints $c$.
\Cref{lem:soundness_of_witnesses} shows that all $\vecx \in W$ are eventually non-terminating.
While this condition is only sufficient for non-termination, we show in \Cref{lem:boosting} that if the program is non-terminating, then there is also an input in $W$.
So the considered program is non-terminating iff $W \neq \emptyset$, i.e., the program is universally terminating iff $W = \emptyset$ (\Cref{CharacterizingTermination}).

Finally, our decision procedure for \textsf{UPAST} is presented in \cref{sect:Deciding PAST}.
\Cref{lem:witness_set_is_semialgebraic} shows that $W$ is semialgebraic and thus the emptiness problem is decidable over the real algebraic numbers, i.e., if $\semiring = \AA$.
Moreover, \cref{lem:witness_set_is_union_of_convex_sets} shows that $W$ can be represented as a finite union of convex semialgebraic sets.
Hence, emptiness of $W$ can also be decided over the rationals and integers~\cite{khachiyan97}.
\Cref{cor:computing_witnesses_for_nontermination} shows how witnesses for non-termination can be obtained from eventually non-terminating inputs $\vecx \in W$.
We discuss our implementation in the tool \textsf{\toolname} in \Cref{Implementation and Conclusion}.\paper{ For all proofs, we refer to~\cite{arxiv}.
  Moreover, \cite{arxiv} also contains proof sketches for our main results to help understanding the essential proof ideas.}\report{
  We mainly present proof sketches in the main part of the paper.
  For the full proofs, we refer to \Cref{app:programsTermination}-\ref{app:Deciding PAST}.}

%% file: programs_and_termination.tex
\section{Programs \& Termination}
\label{sect:programsTermination}
As usual, let $\QQbar \subseteq \CC$ denote the set of algebraic numbers, i.e., the set of all roots of (univariate) polynomials from $\QQ[x]$.
As mentioned, $\AA = \QQbar \cap \RR$ denotes the set of algebraic reals, and $[n]$ denotes the set of positive natural numbers below and including $n$ for every $n\in\NN$ with $\NN = \NN_{>0} \cup \{0\}$, i.e., $[n] = \{1, \ldots, n-1, n\}$ for $n \geq 1$ and $[0] = \emptyset$.

We now define our class of simple randomized loops.
The program variables range over a semiring $\semiring \subseteq \AA$ with a guard consisting of a conjunction of $m$ strict linear inequations over the program variables, represented by a matrix $\matC \in \AA^{m \times n}$, and two commuting linear updates $\matA,\matB \in \semiring^{n \times n}$ that are diagonalizable (over $\CC$).
In each loop iteration, the applied update is chosen among $\matA,\matB$ according to the outcome of a (possibly biased) coin toss.

It is well known (e.g., \cite[Thm.\ 1.3.12]{hornmatrixanalysis}) that two matrices $\matA, \matB \in \CC^{n \times n}$ for $n\in\NN$ are commuting and diago\-nalizable iff they are simultaneously diagonalizable, i.e., there is a regular matrix $\matS \in \CC^{n \times n}$ such that $\matA = \matS \matAD \matS^{-1}$ and $\matB = \matS \matBD \matS^{-1}$, where $\matAD = \diag(a_{1},\ldots,a_{n}) \in \CC^{n \times n}$ and $\matBD = \diag(b_{1},\ldots,b_{n}) \in \CC^{n \times n}$ are complex diagonal matri\-ces, and $a_1,\ldots,a_n$ and $b_1,\ldots,b_n$ are the eigenvalues of $\matA$ and $\matB$, respectively.
Moreover, if $\matA, \matB \in \QQbar^{n \times n}$ are algebraic, then $\matS, \matAD, \matBD \in \QQbar^{n \times n}$ can also be chosen to be algebraic.

\begin{definition}[Simple Randomized Loops]
  \label{def:programs}
  Let $m,n \in \NN_{>0}$, $\matC \in \AA^{m \times n}$, $p \in [0,1] \cap \AA$, and $\matA, \matB \in \semiring^{n \times n}$ such that $\matA$ and $\matB$ are simultaneously diagonalizable.\footnote{\label{trivialSemiring}In the following, we assume $\NN \subseteq \semiring$ and exclude the trivial semiring $\semiring = \{0\}$ as every simple randomized loop terminates for the input $\veczero$.}
  Then,
  \[
    \programExpr{\matC}{\matA}{p}{\matB}{\vecx}
  \]
  is called a \emph{simple randomized loop} (over $\semiring$) of dimension $n$ with $m$ constraints.
  In the remainder, we will omit ``simple randomized'' and just refer to these programs as ``loops''.
\end{definition}

The meaning of ``$\vecx \leftarrow \matA\vecx \, \oplus_p \, \matB\vecx$'' is that $\vecx$ is updated to $\matA\vecx$ with probability $p$ and to $\matB\vecx$ with probability $1-p$.
The comparison $\matC\vecx > \veczero$ is understood componentwise, i.e., $\matC\vecx > \veczero$ iff ${(\matC\vecx)}_{c} > 0$ for all $c\in[m]$, where ${(\matC\vecx)}_{c}$ is the $c$-th entry of the vector ${\matC\vecx}$.
To simplify the notation, from now on we will consider a fixed loop $\program$ of dimension $n$ with $m$ constraints.
Moreover, w.l.o.g.\ we assume $p\in(0,1)$, i.e., $p\not\in\{0,1\}$, as otherwise one can set $\matA$ to $\matB$ if $p=0$ and $\matB$ to $\matA$ if $p=1$, and then replace $p$ by an arbitrary number from $(0,1)$.

A \emph{run} $\run$ of the loop $\program$ is an infinite word $\run = \run_{1}\run_{2}\ldots \in {\{\symMat{A}, \symMat{B}\}}^{\omega}$, where $\Runs$ denotes the set of all runs.
For a run $\run = \run_{1}\ldots$, the value $\run_{i} \in \{\symMat{A},\symMat{B}\}$ indicates which of the updates $\matA$ or $\matB$ was used in the $i$-th iteration of the loop.
\pagebreak[2]
Here, $\symMat{A}$ and $\symMat{B}$ are distinct markings even when $\matA = \matB$.
To simplify the notation, we introduce random variables $\run_{i}\colon\Runs \to \{\symMat{A},\symMat{B}\}$ for $i\in\NN_{>0}$ that map a run to its $i$-th element.
Note that all such\linebreak
random variables $\run_{i}$ are independent and identically distributed.
A \emph{(finite) execution} $f \in \bigcup_{k\in\NN} {\{\symMat{A}, \symMat{B}\}}^{k}$ is a (possibly empty) prefix of a run.
Let $\Path$ denote the (countable) set of all such finite executions.
Given a finite execution $f = \run_{1} \ldots \run_{k}$ with $\run_i \in \{\symMat{A}, \symMat{B}\}$, let $|f| = k$ denote its length.
Furthermore, for $s \in \{\symMat{A},\symMat{B}\}$, $|f|_{s} = |\{i\in[k] \mid \run_{i} = s\}|$ denotes the number of performed updates with update matrix $\matA$ or $\matB$, respectively, during the execution of $f$.

Since the definition of runs is independent from the specific input of the loop, the semantics of the loop $\program$ depend only on the value of $p\in(0,1)$.
To obtain a probability measure $\P$ for $\program$, one first considers cylinder sets $\Pre_{f} = \{f\run_{1}\run_{2}\ldots \in\Runs \mid \run_i \in \{\symMat{A}, \symMat{B} \} \text{ for } i \geq 1 \}$ for all $f\in\Path$, i.e., $\Pre_{f}$ contains all runs with prefix $f$.
By requiring $\P (\Pre_{f}) = p^{|f|_{\symMat{A}}} \cdot {(1-p)}^{|f|_{\symMat{B}}}$ for all $f \in \Path$, one obtains a (unique) probability measure $\P\colon \mathcal{F} \rightarrow (0,1)$ on the $\sigma$-field $\F$ generated by all cylinder sets $\Pre_{f}$, see, e.g.,~\cite[Thm.\ 2.7.2]{ashprobabilityandmeasuretheory}.

\begin{definition}[Semantics of Loops]
  The \emph{semantics} of $\program$ is given as a probability space $(\Runs,\F,\P)$ where $\F = \sigma(\{ \Pre_{f} \mid f\in\Path \})$ and $\P (\Pre_{f}) = p^{|f|_{\symMat{A}}} \cdot {(1-p)}^{|f|_{\symMat{B}}}$.
\end{definition}

To capture the behavior of $\program$ on some specific input $\vecx\in\semiring^{n}$, we introduce a function $\Val_{\vecx} \colon \Path \to \AA^{m}$ that associates finite executions $f\in\Path$ with the values of $\program$'s guard $\matC\vecx$ \emph{after} the execution of $f$.
Recall that $\matA$ and $\matB$ commute.
Hence, we define
\begin{equation*}
  \textstyle
  \Val_{\vecx} (\run_{1}\ldots \run_{k})
  \; = \; \matC \left(\prod_{i=1}^{k}
  \begin{cases}
    \matA & \text{if } f_{i} = \symMat{A} \\
    \matB & \text{otherwise}
  \end{cases} \right) \vecx
  \; = \; \matC \cdot \matA^{|\run_{1}\ldots \run_{k}|_{\symMat{A}}} \cdot \matB^{|\run_{1}\ldots \run_{k}|_{\symMat{B}}}\cdot \vecx
\end{equation*}
where for every matrix $\matM \in \CC^{n \times n}$, $\matM^{0}$ is the $n$-dimensional identity matrix.

\begin{restatable}[Values of Constraints]{lemma}{valuesofconstraints}
  \label{lem:constraint_terms_in_diag_form}
  For any $f\in\Path$ and $\vecx = (x_{1},\ldots,x_{n}) \in \AA^{n}$ we have
  \[ \textstyle
    \Val_{\vecx} (f)\!=\!\matC \cdot \matS \cdot \diag\left(a_{1}^{|f|_{\symMat A}} \cdot b_{1}^{|f|_{\symMat B}},\ldots,a_{n}^{|f|_{\symMat A}} \cdot b_{n}^{|f|_{\symMat B}}\right) \cdot \matS^{-1} \cdot \vecx .
  \]
\end{restatable}

\begin{example}
  \label{ex:initial_program}
  Consider the loop ``$\programExpr{\matC}{\matA}{p}{\matB}{\vecx}$'' with
  \[ \textstyle
    \mbox{\small $\begin{array}{rcl@{\qquad}rcl@{\qquad}rcl}
          \matC & =          &
          \begin{pmatrix}
            1 & 1 & 1
          \end{pmatrix}
                & \matA      & = &
          \begin{pmatrix}
            11 & 5  & -8 \\
            9  & 15 & 8  \\
            7  & -1 & 6  \\
          \end{pmatrix}
                & \matB      & = &
          \begin{pmatrix}
            -7 & 5 & 16  \\
            17 & 5 & -16 \\
            -9 & 7 & -12 \\
          \end{pmatrix}
          \\
                &            &   & \matS & = &
          \begin{pmatrix}
            -\im & \im  & 3 \\
            \im  & -\im & 7 \\
            1    & 1    & 1
          \end{pmatrix}
                & \matS^{-1} & = &
          \begin{pmatrix}
            \tfrac{-1+7\im}{20} & \tfrac{-1-3\im}{20} & \tfrac{1}{2} \\
            \tfrac{-1-7\im}{20} & \tfrac{-1+3\im}{20} & \tfrac{1}{2} \\
            \tfrac{1}{10}       & \tfrac{1}{10}       & 0
          \end{pmatrix}
        \end{array}$}
  \]
  Then, for $\vecx \in \AA^{3}$ and all $f\in\Path$, $\Val_{\vecx}(f)$ equals
  \[ \textstyle
    \matC \cdot \matS \cdot \diag({(6-8\im)}^{|f|_{\symMat{A}}}
    \cdot {(-12+16\im)}^{|f|_{\symMat{B}}}, \;\; {(6+8\im)}^{|f|_{\symMat{A}}} \cdot
    {(-12-16\im)}^{|f|_{\symMat{B}}}, \;\;
    20^{|f|_{\symMat{A}}} \cdot 10^{|f|_{\symMat{B}}})\; \cdot \matS^{-1} \cdot
    \vecx .
  \]
\end{example}

The value of the $c$-th constraint after the execution of $f \in \Path$ on the initial value $\vecx \in \AA^{n}$ is given by the expression
\begin{align}
  {(\Val_{\vecx} (f))}_{c} & = \textstyle {(\matC \cdot \matA^{|f|_{\symMat{A}}} \cdot \matB^{|f|_{\symMat{B}}} \cdot \vecx)}_{c} \nonumber \\
                           & = \textstyle {( \matC \cdot \matS \cdot \matAD^{|f|_{\symMat{A}}} \cdot \matBD^{|f|_{\symMat{B}}} \cdot \matS^{-1}
  \cdot \vecx )}_{c} \tag{\Cref{lem:constraint_terms_in_diag_form}} \\
                           & = \textstyle \sum_{i\in[n]} {(\matC \cdot \matS)}_{c,i} \cdot a_{i}^{|f|_{\symMat{A}}} \cdot b_{i}^{|f|_{\symMat{B}}} \cdot {(\matS^{-1} \cdot \vecx)}_{i} \notag \\
                           & = \textstyle \sum_{i\in[n]} a_{i}^{|f|_{\symMat{A}}} \cdot b_{i}^{|f|_{\symMat{B}}} \cdot \gamma_{c,i} (\vecx) \label{eq:constraint_terms_in_closed_forms}
\end{align}
for linear maps $\gamma_{c,i} \colon \AA^{n} \to \CC$ with
\begin{equation}
  \label{def:gammaci}
  \gamma_{c,i}(\vecx) = {(\matC \cdot \matS)}_{c,i} \cdot {(\matS^{-1} \cdot \vecx)}_{i}.
\end{equation}
Here,
\pagebreak[2]
$(\matC \cdot \matS)_{c,i}$ denotes the entry of $\matC \cdot \matS$ at row $c$ and column $i$.
Moreover, since $\matC \in \AA^{m \times n}$, $\matA,\matB \in \AA^{n \times n}$, and $\vecx\in\AA^{n}$, we have $\gamma_{c,i} (\vecx) \in \QQbar$ for all $(c,i)\in[m]\times[n]$.\footnote{\label{algebraicMatrices}This observation will be needed in the final SMT encoding for our decision procedure (see \Cref{lem:witness_set_is_semialgebraic}), as we have to encode the coefficients $\gamma_{c,i}(\vecx)$ for a given $\vecx \in \AA^{n}$.}
In the following, we refer to the addends $a_{i}^{|f|_{\symMat{A}}} \cdot b_{i}^{|f|_{\symMat{B}}} \cdot \gamma_{c,i} (\vecx)$ of the sum in \cref{eq:constraint_terms_in_closed_forms} as \emph{constraint terms}.

\begin{example}
  \label{ex:gammas}
  Reconsider \Cref{ex:initial_program}.
  Then, $(\Val_{\vecx} (f))_1 = \sum_{i\in[3]} a_{i}^{|f|_{\symMat{A}}} \cdot b_{i}^{|f|_{\symMat{B}}} \cdot \gamma_{1,i} (\vecx)$ with eigenvalues $a_{1} = 6-8\im, \, a_{2} = 6+8\im,\, a_{3} = 20, b_{1} = -12+16\im, \, b_{2} = - 12 - 16\im, \, b_{3} = 10$, and
  \begin{align*}
    \textstyle
    \gamma_{1,1} (\vecx) & = \left(- \tfrac{1}{20} + \tfrac{7\im}{20}\right)x_{1} - \left(\tfrac{1}{20} + \tfrac{3\im}{20}\right)x_{2} + \tfrac{1}{2} x_{3} \\
    \gamma_{1,2} (\vecx) & = \left(- \tfrac{1}{20} - \tfrac{7\im}{20}\right)x_{1} - \left(\tfrac{1}{20} - \tfrac{3\im}{20}\right)x_{2} + \tfrac{1}{2} x_{3} \\
    \gamma_{1,3} (\vecx) & = \tfrac{11}{10}x_{1} + \tfrac{11}{10}x_{2} .
  \end{align*}
\end{example}

\begin{corollary}
  \label{GammaAB}
  For all $(c,i)\in[m]\times[n]$ and all $\vecx,\vecy \in \AA^n$, due to the definition of $\gamma_{c,i}$ we have $\gamma_{c,i}(\matA \cdot \vecx + \vecy) = a_i \cdot \gamma_{c,i}(\vecx) + \gamma_{c,i} (\vecy)$ and $\gamma_{c,i}(\matB \cdot \vecx + \vecy) = b_i \cdot \gamma_{c,i}(\vecx) + \gamma_{c,i} (\vecy)$.
\end{corollary}

The following lemma shows that if one has two pairs of eigenvalues $(a,b)$ and $(\conj{a},\conj{b})$ where $\conj{a}$ and $\conj{b}$ are the complex conjugates of $a$ and $b$, then the sum of all linear maps $\gamma_{c,i} (\vecx)$ where $(a_i,b_i) = (a,b)$ is the complex conjugate of the sum of all linear maps $\gamma_{c,i} (\vecx)$ where $(a_i,b_i) = (\conj{a},\conj{b})$.
For instance in \Cref{ex:gammas}, $(a_1,b_1)$ are the complex conjugates of $(a_2,b_2)$ and indeed, we have $\gamma_{1,2}(\vecx) = \conj{\gamma_{1,1}(\vecx)}$.
This lemma will later be needed to show that when representing ${(\Val_{\vecx} (f))}_{c}$ and summing up the coefficients of its addends in a suitable way, all resulting coefficients are real numbers (see \Cref{rem:sums_of_constraint_term_groups_are_real_valued}).

\begin{restatable}[Sums of Conjugated Constraint Terms are Real-Valued]{lemma}{sumsofconjugatedaddendsarerealvalued}
  \label{lem:sums_of_conjugated_addends_are_real_valued}
  Let $c\in[m]$, let $a,b \in \CC$, and let $\gamma_{c,i}$ be the linear map from \cref{def:gammaci}.
  Then, for all inputs $\vecx \in \AA^{n}$ we have $\gamma_{2} = \conj{\gamma_{1}}$ where
  \[
    \textstyle
    \gamma_{1} = \sum_{i \,\in\, [n], \;
    (a_{i},b_{i}) \,=\, (a,b)} \;
    \gamma_{c,i} (\vecx) \quad\text{and}\quad \gamma_{2} = \sum_{i \,\in\, [n], \;
    (a_{i},b_{i}) \,=\, (\conj{a}, \conj{b})} \;
    \gamma_{c,i} (\vecx).
  \]
\end{restatable}

In order to define the notion of termination for $\program$, we first introduce the concept of a run's length by counting the number of iterations until the guard is violated for the first time.
Throughout the paper, we use the convention $\min \emptyset = \infty$.

\begin{definition}[Length Of Runs]
  For any $\vecx \in \AA^n$, we define the random variable $\LRVar_{\vecx} \colon \Runs \to \NN \cup \{\infty\}$ as $\LRVar_{\vecx}(\run_{1}\run_2\ldots ) = \min \{k \in \NN \mid \Val_{\vecx} (\run_{1}\ldots\run_{k}) \not > \veczero\}$.
\end{definition}

We now define the \emph{expected runtime} of $\program$ for the input $\vecx$ as the expectation $\E(\LRVar_{\vecx})$.
If $\E (\LRVar_{\vecx}) = \infty$, we call the corresponding input $\vecx$ \emph{non-terminating}.
So we consider \emph{positive almost sure termination} \cite{bournez2005proving,DBLP:conf/mfcs/Saheb-Djahromi78}, where termination corresponds to a finite expected runtime.

\begin{definition}[Non-Terminating Inputs]
  The set of \emph{non-terminating} inputs is $\NT = \{\vecx \in \AA^{n} \mid \E (\LRVar_{\vecx}) = \infty \}$.
\end{definition}
Consequently, we call $\program$ \emph{terminating} whenever $\NT = \emptyset$ and non-terminating otherwise.

As in~\cite{braverman06,hark20,hark23,hosseini19,ouaknine15,tiwari04}, we focus on \emph{eventual non-termination} instead of actual non-termination as this allows us to ignore a finite number of initial updates of the loop.
In our setting, an input $\vecx$ is \emph{eventually non-terminating}
if a non-terminating input $\vecy$ can be reached by repeated application of the updates in the loop body to $\vecx$.

\begin{definition}[Eventual Non-Termination]
  \label{def:ent}
  We define the set of \emph{eventually non-terminating} inputs as $\ENT = \bigcup_{j,k \in \NN} \{ \vecx \in \AA^{n} \mid \matA^{j}\matB^{k}\vecx \in \NT \}$.
\end{definition}

The motivation behind considering $\ENT$ instead of $\NT$ is that it allows us to ``jump'' over the first iterations of the loop (where the loop guard might be violated).
In this way, we can focus only on the longterm behavior of the loop on a given input.

\begin{example}[Difference Between $\NT$ \& $\ENT$]
  Consider the loop ``$\programExpr{\matC}{\matA}{0.5}{\matB}{\vecx}$'' with $\matC = \mbox{\small $
        \begin{pmatrix}
          1 & 1
        \end{pmatrix}
      $}$ and $\matA = \matB = \mbox{\small $
        \begin{pmatrix}
          2 & 0 \\
          0 & 1
        \end{pmatrix}
      $}$.
  Then, $\vecx = \mbox{\small $
        \begin{pmatrix}
          1 \\
          -2
        \end{pmatrix}
      $}
    \not\in \NT$ as $\matC \vecx = -1\linebreak
    \not> 0$, i.e., $\vecx$ violates the loop guard. While $\matA \vecx$ also violates the loop guard (since $\matC \matA \vecx = 0\linebreak
    \not> 0$), we have $\matC \matA^j \vecx > 0$ for all $j \geq 2$. Thus, $\matA^2 \vecx \in \NT$ and therefore, $\vecx \in \ENT$.
\end{example}

Considering $\ENT$ instead of $\NT$ is justified by the fact that $\NT = \emptyset$ iff $\ENT=\emptyset$, as shown by \Cref{lem:nt_vs_ent}.

\begin{restatable}[Correspondence of $\ENT$ \& $\NT$]{lemma}{correspondencebetweenentandnt}
  \label{lem:nt_vs_ent}
  For any semiring $\semiring \subseteq \AA$, we have $\NT \cap \semiring^{n} = \emptyset$ iff $\ENT \cap \semiring^{n} = \emptyset$.
\end{restatable}

%% file: on_constraint_terms.tex
\section{On Constraint Terms}
\label{sect:On Constraint Terms}
In this section we consider the value $\Val_{\vecx}(f)$ for $|f| \rightarrow \infty$, motivated by our interest in eventual non-termination.
In the first part of this section, we represent ${(\Val_{\vecx}(f))}_{c}$, for $c \in [m]$, as a sum over so-called ``constraint term groups'', expressed using a quantity $\URVar(f)$ corresponding to ``how much $f$ has deviated from the expected execution''.
The section's second part then shows that for specific $f \in \Path$ it suffices to only consider certain addends of this sum in order to decide whether ${(\Val_{\vecx}(f))}_{c} > 0$ as $|f| \rightarrow \infty$.
This observation will lead to a necessary condition for eventual non-termination in \cref{sect:Positive Eigenvalues}, which will subsequently be turned into a sufficient criterion in \cref{sec:towards_witnesses_for_non-termination}.

When executing the loop $\program$, the relative number of times that update $\matA$ is selected over $\matB$ will intuitively tend towards $p$ with increasing number of iterations.
We now consider this relative quantity and additionally subtract $p$ to center its distribution around $0$.

\begin{definition}[Deviation From Equilibrium]
  \label{def:u}
  The mapping $\URVar\colon \Path \to [-p,1-p]$ is defined as $\URVar (f) = \tfrac{|f|_{\symMat{A}}}{|f|} - p$ for every non-empty $f\in\Path$ and $\URVar(f) = 0$ otherwise.
\end{definition}

We will investigate which addends determine the sign of ${(\Val_{\vecx} (f))}_{c}$ for $|f| \to \infty$.
To this end, we want to express the value of the constraint terms after some finite execution $f$ in terms of $\URVar(f)$ and $|f|$.
The advantage is that for sufficiently long paths $f$, we know how $|f|$ and $\URVar(f)$ ``behave'' (i.e., $|f|$ ``tends towards'' $\infty$ and $\URVar(f)$ is ``expected to tend towards'' 0).
\begin{restatable}[Normal Form of Constraint Terms]{lemma}{normalformofconstraintterms}
  \label{lem:normal_form_of_constraint_terms}
  Let $a^{|f|_{\symMat{A}}}b^{|f|_{\symMat{B}}}\gamma$ be a constraint term with $a,b \neq 0$.
  We write $a,b \in \CC$ in polar form as $|a|\zeta_{\matA}$ and $|b|\zeta_{\matB}$, respectively, where $\zeta_{\matA} = \frac{a}{|a|}, \zeta_{\matB} = \frac{b}{|b|}\in \CC$ are complex units, i.e., $|\zeta_{\matA}| = |\zeta_{\matB}| = 1$.
  Then, for any (non-empty) finite execution $f\in\Path$ with $\URVar(f) \in (-p,1-p)$ we have
  \[ \textstyle
    a^{|f|_{\symMat{A}}}b^{|f|_{\symMat{B}}}\gamma = \zeta_{\matA}^{|f|_{\symMat{A}}}\zeta_{\matB}^{|f|_{\symMat{B}}} {\left( \tfrac{|a|^{p}}{|b|^{p-1}} {\left(\tfrac{|a|}{|b|}\right)}^{\URVar(f)}\right)}^{|f|} \gamma .
  \]
\end{restatable}

Note that in \Cref{lem:normal_form_of_constraint_terms} we excluded all constraint terms with $a = 0$ or $b = 0$ in order to avoid a division by zero.
However, we additionally required\footnote{The set of runs $\run$ where $\URVar (f) \in \{-p,1-p\}$ for every prefix $f$ of $\run$ has probability 0, as $\URVar(f) \in \{p,1-p\}$ means that only $\matA$ or $\matB$ has been selected in $f$.
  However, we had required $0 < p < 1$.}
$\URVar (f) \not \in \{-p,1-p\}$ implying $0 < |f|_{\symMat{A}}$ and $0 < |f|_{\symMat{B}}$, since $|f| > 0$.
Hence, for all constraint terms $a^{|f|_{\symMat{A}}}b^{|f|_{\symMat{B}}}\gamma$ with $a=0$ or\linebreak
$b=0$ and all considered $f\in\Path$, we have $a^{|f|_{\symMat{A}}}b^{|f|_{\symMat{B}}} = 0$ and thus the value of such constraint terms can safely be ignored when computing ${\left(\Val_{\vecx} (f)\right)}_{c}$.
This leads to the
\pagebreak[2]
equation
\begin{align}
  {\left(\Val_{\vecx} (f)\right)}_{c}
   & = \textstyle \sum_{i\in[n]} a_{i}^{|f|_{\symMat{A}}}b_{i}^{|f|_{\symMat{B}}}\gamma_{c,i}(\vecx) \nonumber \\
   & = \textstyle \sum_{i\in[n], \;
    a_{i}, b_{i} \neq 0} \; \zeta_{i,\matA}^{|f|_{\symMat{A}}}\zeta_{i,\matB}^{|f|_{\symMat{B}}} {\left(\tfrac{|a_{i}|^{p}}{|b_{i}|^{p-1}} {\left(\tfrac{|a_{i}|}{|b_{i}|}\right)}^{\URVar(f)}\right)}^{|f|} \gamma_{c,i}(\vecx),
  \label{eq:complete_constraint_in_normal_form}
\end{align}
where $\zeta_{i,\matA} = \frac{a_i}{|a_i|}$ and $\zeta_{i,\matB} = \frac{b_i}{|b_i|}$.

\begin{example}
  \label{ex:normal_form}
  Transforming the sum from \Cref{ex:gammas} into the form \cref{eq:complete_constraint_in_normal_form} and setting $p = \frac{1}{2}$ yields
  \begin{align*}
    \hspace*{-.5cm} {\left(\Val_{\vecx} (f)\right)}_{1} ={} & \textstyle
    \zeta_{1,\matA}^{|f|_{\symMat{A}}} \zeta_{1,\matB}^{|f|_{\symMat{B}}} {\left(
    10\sqrt{2} \cdot {(\tfrac{1}{2})}^{\URVar(f)}\right)}^{|f|} \gamma_{1,1}
    (\vecx) \; + \;
    \textstyle \zeta_{2,\matA}^{|f|_{\symMat{A}}} \zeta_{2,\matB}^{|f|_{\symMat{B}}}
    {\left( 10 \sqrt{2} \cdot {(\tfrac{1}{2})}^{\URVar(f)}\right)}^{|f|} \gamma_{1,2} (\vecx) \\
                                                            & {}\textstyle + \; \zeta_{3,\matA}^{|f|_{\symMat{A}}} \zeta_{3,\matB}^{|f|_{\symMat{B}}} {\left( 10 \sqrt{2} \cdot {2}^{\URVar(f)}\right)}^{|f|} \gamma_{1,3} (\vecx)
  \end{align*}
  with $\zeta_{1,\matA} = e^{-\im \arctan (\nicefrac{4}{3})},\, \zeta_{2,\matA} = e^{\im \arctan{\nicefrac{4}{3}}},\, \zeta_{3,\matA} = 1$ and $\zeta_{1,\matB} = e^{\im (\pi - \arctan(\nicefrac{4}{3}))},\, \zeta_{2,\matB} = e^{\im (\arctan(\nicefrac{4}{3}) - \pi)},\, \zeta_{3,\matB} = 1$ for all $f\in\Path$.
\end{example}

By inspecting the right-hand side of \cref{eq:complete_constraint_in_normal_form} it becomes clear that the subexpressions $\frac{|a_{i}|^{p}}{|b_{i}|^{p-1}}$ and $\frac{|a_{i}|}{|b_{i}|}$ govern the overall asymptotic growth of ${(\Val_{\vecx}(f))}_{c}$ as $|f|$ increases.
In the following, we group all constraint terms into so-called \emph{constraint term groups}
which are sets of indices $i$ corresponding to constraint terms where $\frac{|a_{i}|^{p}}{|b_{i}|^{p-1}}$ and $\frac{|a_{i}|}{|b_{i}|}$ have common values $\scaleinside$ and $\scaleoutside$.
Here, $\scaleinside = \frac{|a_{i}|^{p}}{|b_{i}|^{p-1}}$ is the expression that is important in \cref{eq:complete_constraint_in_normal_form} if $\URVar (f)$ is \underline{\textbf{i}}n a region close to 0.\footnote{Note that $a_{i}^{|f|_{\symMat{A}}}b_{i}^{|f|_{\symMat{B}}}$ is the $|f|$-th power of the weighted geometric mean $a_i^{p}b_i^{1-p}$ whenever $\URVar (f) = 0$.}
If one is \underline{\textbf{o}}utside such a region, then $\scaleoutside = \frac{|a_{i}|}{|b_{i}|}$ is important as well.

\begin{definition}[Constraint Term Groups]
  \label{def:constraint_term_groups}
  For any $(\scaleinside,\scaleoutside)\in\RR^{2}_{>0}$, let
  \[ \textstyle
    \constrainttermgroup_{(\scaleinside,\scaleoutside)} = \{ i \in [n] \mid 0 \not \in \{a_{i},b_{i}\},\, \tfrac{|a_{i}|^{p}}{|b_{i}|^{p-1}} = \scaleinside, \tfrac{|a_{i}|}{|b_{i}|} = \scaleoutside\} .
  \]
  Moreover, we define the \emph{finite} set $\mathcal{I}$ of all pairs $(\scaleinside,\scaleoutside) \in \RR^{2}_{>0}$ with $\constrainttermgroup_{(\scaleinside,\scaleoutside)} \neq \emptyset$.
  For $c\in[m]$ and $\vec{x} \in \AA^n$, let $\constrainttermgroup_{(\scaleinside,\scaleoutside),c,\vecx} = \emptyset$ whenever $\sum_{{i\in\constrainttermgroup_{(\scaleinside,\scaleoutside)}}}
    \zeta_{i,\matA}^{|f|_{\symMat{A}}}\zeta_{i,\matB}^{|f|_{\symMat{B}}}\gamma_{c,i}(\vecx) = 0$ holds for all $f \in \Path$.\footnote{\label{footnote:comment_on_constrainttermgroups}%
    We will show how to check this in \Cref{lem:braverman_generalisation}.
    Note that this is not implied by $\sum_{{i\in\constrainttermgroup_{(\scaleinside,\scaleoutside)}}}
      \gamma_{c,i}(\vecx) = 0$. As a counterexample, consider $\gamma_{1} = -1$, $\gamma_{2} = 1$ (and thus, $\gamma_1 + \gamma_2 = 0$), and $\zeta_{1,\matA} = -1$, $\zeta_{1,\matB} = \zeta_{2,\matA} = \zeta_{2,\matB} = 1$ (and hence, $\zeta_{1,\matA} \zeta_{1,\matB} \gamma_1 + \zeta_{2,\matA} \zeta_{2,\matB}
      \gamma_2 = 2$).}
  Otherwise let $\constrainttermgroup_{(\scaleinside,\scaleoutside),c,\vecx} = \constrainttermgroup_{(\scaleinside,\scaleoutside)}$.
  We refer to the sets $\constrainttermgroup_{(\scaleinside,\scaleoutside),c,\vecx}$ as \emph{constraint term groups}.
\end{definition}

For $c\in[m]$ and non-empty $f\in\Path$ with $\URVar(f) \in (-p,1-p)$, \Cref{eq:complete_constraint_in_normal_form} can be rearranged to
\begin{equation} \textstyle
  {\left( \Val_{\vecx}(f) \right)}_{c} = \sum_{(\scaleinside,\scaleoutside) \in
    \mathcal{I}} \; {\left(\scaleinside \cdot
  \scaleoutside^{\URVar(f)}\right)}^{|f|} \; \sum_{i \in
    \constrainttermgroup_{(\scaleinside,\scaleoutside),c,\vecx}} \;
  \zeta_{i,\matA}^{|f|_{\symMat{A}}} \, \zeta_{i,\matB}^{|f|_{\symMat{B}}} \, \gamma_{c,i}
  (\vecx) \label{eq:complete_constraint_in_grouped_normal_form}
\end{equation}
with \Cref{def:constraint_term_groups}.

\Cref{lem:shared_modulus_of_eigenvalues_in_constraint_term_groups}
shows that for all $a_{1},b_{1},a_{2},b_{2} \neq 0$ we have $\tfrac{|a_{1}|}{|b_{1}|} = \tfrac{|a_{2}|}{|b_{2}|}$ and $\tfrac{|a_{1}|^{p}}{|b_{1}|^{p-1}} = \tfrac{|a_{2}|^{p}}{|b_{2}|^{p-1}}$ iff $|a_{1}| = |a_{2}|$ and $|b_{1}| = |b_{2}|$.
Thus, if $i, i' \in \constrainttermgroup_{(\scaleinside,\scaleoutside)}$, then $|a_i| = |a_{i'}|$ and $|b_i| = |b_{i'}|$.
We will first return to this result in \Cref{sect:Positive Eigenvalues}, where we will use that for all $i,i' \in \constrainttermgroup_{(\scaleinside,\scaleoutside)}$ with $a_{i},b_{i},a_{i'},b_{i'} \in \RR_{>0}$ we have $(a_{i},b_{i}) = (a'_{i},b'_{i})$.
Later on, we will revisit it in \Cref{sec:towards_witnesses_for_non-termination}.

\begin{restatable}[Equality of Eigenvalues]{lemma}{equalityofeigenvalues}
  \label{lem:shared_modulus_of_eigenvalues_in_constraint_term_groups}
  Let $a_{1},b_{1},a_{2},b_{2} \in \RR_{>0}$ be positive reals such that $\frac{a_{1}}{b_{1}} = \frac{a_{2}}{b_{2}}$ and $\frac{a_{1}^{p}}{b_{1}^{p-1}} = \frac{a_{2}^{p}}{b_{2}^{p-1}}$.
  Then we have $(a_{1},b_{1}) = (a_{2},b_{2})$.
\end{restatable}

Due to \cref{eq:complete_constraint_in_grouped_normal_form}, we have to consider sums $\sum_{i \in \constrainttermgroup_{(\scaleinside,\scaleoutside),c,\vecx}} \zeta_{i,\matA}^{|f|_{\symMat{A}}} \zeta_{i,\matB}^{|f|_{\symMat{B}}} \gamma_{c,i} (\vecx)$ for $\vecx\in\AA^{n}$ and $c\in[m]$.
While $\zeta_{i,\matA}$, $\zeta_{i,\matB}$, and $\gamma_{c,i}(\vecx)$ are complex numbers in general, such sums are always real-valued.
The corresponding \Cref{rem:sums_of_constraint_term_groups_are_real_valued} is an immediate consequence of \Cref{lem:sums_of_conjugated_addends_are_real_valued}.
Later, this remark will allow us to make statements about the signs of such sums, see \Cref{lem:domination_of_eventually_dominating_constraint_term_groups}.
\begin{restatable}[Coefficients of Constraint Term Groups are Real-Valued]{remark}{coefficientsofconstrainttermgroupsarerealvalued}
  \label{rem:sums_of_constraint_term_groups_are_real_valued}
  Let $c \in [m]$ and $(\scaleinside,\scaleoutside) \in \mathcal{I}$.
  Then, for all $f\in\Path$ we have $\sum_{i \in \constrainttermgroup_{(\scaleinside,\scaleoutside),c,\vecx}} \zeta_{i,\matA}^{|f|_{\symMat{A}}} \zeta_{i,\matB}^{|f|_{\symMat{B}}} \gamma_{c,i} (\vecx) \;\in\; \AA$.
\end{restatable}

\begin{example}
  \label{ex:constraint_term_groups}
  We continue \Cref{ex:normal_form}.
  There are two different non-empty sets $\constrainttermgroup_{(\scaleinside,\scaleoutside)}$, i.e., $\constrainttermgroup_{(10\sqrt{2},\nicefrac{1}{2})} = \{1,2\}$ and $\constrainttermgroup_{(10\sqrt{2},2)} = \{3\}$.
  As indicated by \Cref{lem:shared_modulus_of_eigenvalues_in_constraint_term_groups}, this implies $|a_1| = |a_2|$ and $|b_1| = |b_2|$.
  Hence, $\zeta_{2,\matA} = \frac{a_2}{|a_2|} = \frac{a_2}{|a_1|} = \conj{\frac{a_1}{|a_1|}} = \conj{\zeta_{1,\matA}}$ and similarly, $\zeta_{2,\matB}= \conj{\zeta_{1,\matB}}$.
  Moreover, recall that $\gamma_{1,2}(\vecx) = \conj{\gamma_{1,1}(\vecx)}$.
  Therefore, we have
  \[ \begin{array}{rl}
           & \constrainttermgroup_{(10\sqrt{2},\nicefrac{1}{2}),1,\vecx} \neq \emptyset                                                                                                                                                                                         \\
      \iff & \text{there is a $f \in \Path$ with
      $\zeta_{1,\matA}^{|f|_{\symMat{A}}} \cdot \zeta_{1,\matB}^{|f|_{\symMat{B}}}
        \cdot \gamma_{1,1}(\vecx) + \zeta_{2,\matA}^{|f|_{\symMat{A}}} \cdot
      \zeta_{2,\matB}^{|f|_{\symMat{B}}} \cdot \gamma_{1,2}(\vecx) \neq 0$}                                                                                                                                                                                                     \\
      \iff & \text{there is a $f \in \Path$ with
      $\zeta_{1,\matA}^{|f|_{\symMat{A}}} \cdot \zeta_{1,\matB}^{|f|_{\symMat{B}}}
        \cdot \gamma_{1,1}(\vecx) + \conj{\zeta_{1,\matA}}^{|f|_{\symMat{A}}} \cdot
        \conj{\zeta_{1,\matB}}^{|f|_{\symMat{B}}} \cdot \conj{\gamma_{1,1}(\vecx)} \neq
      0$}                                                                                                                                                                                                                                                                       \\
      \iff & \text{there is a $f \in \Path$ with $\zeta_{1,\matA}^{|f|_{\symMat{A}}} \cdot \zeta_{1,\matB}^{|f|_{\symMat{B}}} \cdot \gamma_{1,1}(\vecx) + \conj{\zeta_{1,\matA}^{|f|_{\symMat{A}}} \cdot \zeta_{1,\matB}^{|f|_{\symMat{B}}} \cdot \gamma_{1,1}(\vecx)} \neq 0$} \\
      \iff & \text{there is a $f \in \Path$ with $\Re(\zeta_{1,\matA}^{|f|_{\symMat{A}}} \cdot \zeta_{1,\matB}^{|f|_{\symMat{B}}} \cdot \gamma_{1,1}(\vecx)) \neq 0$}                                                                                                           \\
      \iff & \text{there is a $f \in \Path$ with } \zeta_{1,\matA}^{|f|_{\symMat{A}}} \cdot \zeta_{1,\matB}^{|f|_{\symMat{B}}} \cdot
      \gamma_{1,1}(\vecx) \neq 0 \qquad\qquad\qquad \text{(by \Cref{rem:sums_of_constraint_term_groups_are_real_valued})}
      \\
      \iff & \gamma_{1,1}(\vecx) \neq 0
    \end{array}\]
  For the last step, the direction ``$\Longrightarrow$'' is clear, since $\gamma_{1,1}(\vecx) = 0$ implies $\zeta_{1,\matA}^{|f|_{\symMat{A}}} \cdot \zeta_{1,\matB}^{|f|_{\symMat{B}}} \cdot \gamma_{1,1}(\vecx) = 0$ for all $f \in \Path$.
  The direction ``$\Longleftarrow$'' is also clear by choosing $f$ to be the empty path.

  Hence, $\constrainttermgroup_{(10\sqrt{2},\nicefrac{1}{2}),1,\vecx} \neq \emptyset \iff \gamma_{1,1}(\vecx) \neq 0 \iff 10 x_{3} \neq x_{1} + x_{2} \vee 7 x_1 \neq 3 x_2$.
  Similarly, $\constrainttermgroup_{(10\sqrt{2},2),1,\vecx} \neq \emptyset \iff \gamma_{1,3} (\vecx) \neq 0 \iff x_{1} + x_{2} \neq 0$ for $\vecx\in\AA^{3}$.
\end{example}

By inspecting~\cref{eq:complete_constraint_in_grouped_normal_form} again, one observes that if $\URVar(f)$ is sufficiently close to $0$, then the terms belonging to the non-empty constraint term group $\constrainttermgroup_{(\scaleinside,\scaleoutside),c,\vecx}$ with maximal $(\scaleinside,\scaleoutside)$ in the lexicographic ordering will at some point (when $|f| \to \infty$) outgrow all other terms whenever $\URVar(f) \cdot |f| > 0$ is \underline{\textbf{p}}ositive and sufficiently large.
If $\URVar(f) \cdot |f| < 0$ is \underline{\textbf{n}}egative and sufficiently small, then, however, we have to focus our attention on the non-empty group $\constrainttermgroup_{(\scaleinside,\scaleoutside),c,\vecx}$ for which $(\scaleinside,-\scaleoutside)$ is maximal in the lexicographic ordering.\footnote{Since $\scaleinside$ and $\scaleoutside$ are always positive reals, the expressions $\sum_{i \in \constrainttermgroup_{(\scaleinside,\scaleoutside),c,\vecx}}
  \zeta_{i,\matA}^{|f|_{\symMat{A}}} \zeta_{i,\matB}^{|f|_{\symMat{B}}} \gamma_{c,i}
  (\vecx)$ determine whether ${\left( \Val_{\vecx}(f) \right)}_{c}$ is positive or negative. However, for $|f| \to \infty$, this expression (and also $\URVar(f)$) could alternate between being positive, negative, or even 0. This will be regarded in \Cref{sect:Positive Eigenvalues,sec:towards_witnesses_for_non-termination}.}

\begin{definition}[Eventually Dominating Constraint Term Group]
  \label{def:dominating_constraint_term_groups}
  $\!$Let ${\leq_{\mathrm{lex},\pmax}}\!=\!\{((\scaleinside_{1},\scaleoutside_{1}),(\scaleinside_{2},\scaleoutside_{2}))\!\mid \scaleinside_{1}<\scaleinside_{2}$ or both $\scaleinside_{1}=\scaleinside_{2}
    \text{ and }
    \scaleoutside_{1} \leq \scaleoutside_{2}\}$ denote the usual lexicographic ordering on $\RR$ and let ${\leq_{\mathrm{lex},\nmax}} = \{((\scaleinside_{1},\scaleoutside_{1}),(\scaleinside_{2},\scaleoutside_{2})) \mid \scaleinside_{1}<\scaleinside_{2}
    \text{ or both }
    \scaleinside_{1}=\scaleinside_{2}
    \text{ and }
    \scaleoutside_{2} \leq \scaleoutside_{1}\}$ be the lexicographic ordering where the comparison on the second component is flipped.

  For all $d\in\{\nmax,\pmax\}$, $c\in[m]$, and $\vecx\in\AA^{n}$ we define $\constrainttermgroup_{d,c,\vecx}
    = \emptyset$ if $\constrainttermgroup_{(\scaleinside,\scaleoutside),c,\vecx} = \emptyset$ for all $(\scaleinside,\scaleoutside) \in \mathcal{I}$, and $\constrainttermgroup_{d,c,\vecx} = \constrainttermgroup_{(\scaleinside,\scaleoutside),c,\vecx}$ for $(\scaleinside,\scaleoutside) = \max\nolimits_{\mathrm{lex},d}
    \{(\scaleinside,\scaleoutside) \in \mathcal{I} \mid \constrainttermgroup_{(\scaleinside,\scaleoutside),c,\vecx} \neq \emptyset\}$ otherwise, where $\max\nolimits_{\mathrm{lex},d}$ denotes the maximum w.r.t.\ the ordering $\leq_{\mathrm{lex},d}$.
\end{definition}

It is not a priori clear, how, for a constraint index $c\in[m]$ and a given input $\vecx\in\AA^{n}$, the sets $\constrainttermgroup_{\nmax,c,\vecx},\constrainttermgroup_{\pmax,c,\vecx}$ can be computed automatically since one has to decide whether $\frac{a_{1}^{p}}{b_{1}^{p-1}} < \frac{a_{2}^{p}}{b_{2}^{p-1}}$ where $a_{1},b_{1},a_{2},b_{2} \in \AA_{>0}$ are assumed to be positive algebraic reals, but $a_i^p$ and $b_i^{p-1}$ are in general non-algebraic reals.
This is due to the well-known Gelfond-Schneider theorem (see, e.g., \cite[Thm.\ 3.0.1]{transcendentalnumbers}), which states that $a^{p} \not \in \QQbar$ whenever $a \in \QQbar \setminus \{0,1\}$ and $p \in \QQbar$ is irrational.
However, \Cref{lem:comparison_of_constraint_term_groups} ensures the decidability of such comparisons.

\begin{restatable}[Comparing Constraint Term Groups]{lemma}{onthecomparisonofconstrainttermgroups}
  \label{lem:comparison_of_constraint_term_groups}
  Let $T_{1} = \frac{a_{1}^{p}}{b_{1}^{p-1}}$ and $T_{2} = \frac{a_{2}^{p}}{b_{2}^{p-1}}$ for positive algebraic reals $a_{1},b_{1},a_{2},b_{2} \in \AA_{>0}$ and $p\in\AA\cap(0,1)$.
  Then the statement $T_{1} < T_{2}$ is decidable.
\end{restatable}

\begin{example}
  \label{ex:eventually_dominating_constraint_term_groups}
  Recall the two non-empty constraint term groups $\constrainttermgroup_{(10\sqrt{2},\nicefrac{1}{2}),1,\vecx}$ and $\constrainttermgroup_{(10\sqrt{2},2),1,\vecx}$ from \Cref{ex:constraint_term_groups}.
  Then, for $\vecx \in \AA^{3}$, we have
  \begin{align*} \textstyle
    \constrainttermgroup_{\nmax,1,\vecx} & =
    \begin{cases}
      \constrainttermgroup_{(10\sqrt{2},\nicefrac{1}{2})}= \{1,2\} & \text{if $10 x_{3} \neq x_{1} + x_{2}$ or $7 x_1 \neq 3 x_2$} \\
      \constrainttermgroup_{(10\sqrt{2},2)} \;\; = \{3\}           & \text{if
      $10 x_{3} = x_{1} + x_{2}$, $7 x_1 = 3 x_2$, and $x_{1} + x_{2} \neq 0$}                                                     \\
      \emptyset                                                    & \text{otherwise}
    \end{cases}
    \\
    \constrainttermgroup_{\pmax,1,\vecx} & =
    \begin{cases}
      \constrainttermgroup_{(10\sqrt{2},2)} \;\; = \{3\}           & \text{if $x_{1} + x_{2} \neq 0$} \\
      \constrainttermgroup_{(10\sqrt{2},\nicefrac{1}{2})}= \{1,2\} & \text{if
        $x_{1} + x_{2} = 0$ and
      ($10x_{3} \neq x_{1} + x_{2}$ or $7 x_1 \neq 3 x_2$)}                                           \\
      \emptyset                                                    & \text{otherwise}
    \end{cases}
  \end{align*}
\end{example}

\begin{figure}
  \begin{center}
    \input{illustration_full.tex}
  \end{center}
  \caption{\label{fig:safe} Illustration of the General Idea Using Safe Regions}
\end{figure}
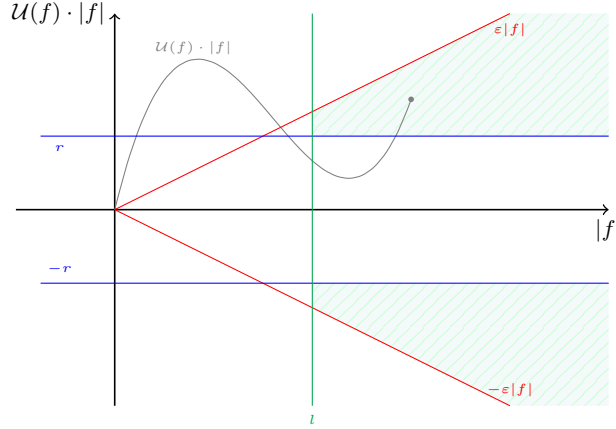

\Cref{lem:domination_of_eventually_dominating_constraint_term_groups} states the main property of the eventually dominating constraint term groups $\constrainttermgroup_{\nmax,c,\vecx}$ and $\constrainttermgroup_{\pmax,c,\vecx}$.
\Cref{fig:safe}
depicts this lemma (and also the following lemmas) graphically.
Here, the horizontal axis represents the length $|f|$ of the path and the vertical axis represents the value of the function $\URVar(f) \cdot |f| = |f|_{\symMat{A}} - p \cdot |f|$, which expresses the deviation of the number of $\symMat{A}$-symbols in the execution $f$ from the expected number of $\symMat{A}$-symbols.
We depicted $\URVar(f) \cdot |f|$ by a gray line.
The lemma essentially states that whenever $\URVar(f) \cdot |f|$ reaches one of the two ``safe'' regions marked in green, then the coefficient $v(f)$ of the dominant addend determines the sign of ${(\Val_{\vecx}(f))}_{c}$, provided that its absolute value $|v(f)|$ is large enough.
The upper safe region is the one for $d=\pmax$, i.e., here the path $f$ is long enough (i.e., $|f| \geq l$), $\URVar(f) \in [0,\eps]$ (i.e., $0 \leq \URVar(f) \cdot |f| \leq \eps \cdot |f|$), and $\URVar(f) \cdot |f| \geq r$.
Similarly, the lower safe region corresponds to the case $d=\nmax$.
This lemma also indicates why an extension of our approach to programs with three instead of two update matrices would be problematic.
Then instead of $\URVar(f)$ we would need a vector to express how much an execution deviates from the probabilities in the program.
This would break our concepts of eventually dominating constraint term groups and safe regions, since instead of $d\in\{\nmax,\pmax\}$, we would have to consider the ``direction'' of this deviation.

\begin{restatable}[Domination of Eventually Dominating Constraint Term Groups]{lemma}{dominationofeventuallydominatingconstrainttermgroups}
  \label{lem:domination_of_eventually_dominating_constraint_term_groups}
  Let $c\in[m]$, $d\in\{\nmax,\pmax\}$, and $\vecx \in \AA^{n}$.
  We define $v\colon \Path \to \AA$ (see \Cref{rem:sums_of_constraint_term_groups_are_real_valued}) as
  \[ \textstyle
    v (f) = \sum_{i \in \constrainttermgroup_{d,c,\vecx}} \zeta_{i,\matA}^{|f|_{\symMat{A}}} \zeta_{i,\matB}^{|f|_{\symMat{B}}} \gamma_{c,i} (\vecx)
  \]
  Then for every $\rho \in \AA_{>0}$,
  \pagebreak[2]
  there exist constants $\eps\in\AA_{>0}$, $r\in\NN$, and $l\in\NN_{>0}$ (for a bound on the \underline{l}ength of the path), such that for all $f\in\Path$ with $|f| \geq l$, $|v (f)| \geq \rho$, $|\URVar(f) \cdot |f|| \ge r$, $\URVar(f) \in [-\eps,0]$ if $d=\nmax$, and $\URVar(f) \in [0,\eps]$ if $d=\pmax$, we have $\sign ({(\Val_{\vecx}(f))}_{c}) = \sign(v (f))$.
\end{restatable}
\report{
  \begin{proof}[Proof Sketch]
    We refer to \cref{app:proofs_for_on_constraint_terms} for the full proof.
    W.l.o.g.\ assume $d = \pmax$ as the other case is symmetric.
    Moreover, we assume $\constrainttermgroup_{\pmax,c,\vecx} \neq \emptyset$ as otherwise $|v(f)| = 0 \not \geq \rho > 0$.
    So, let $(\scaleinside_{\max},\scaleoutside_{\max}) \in \mathcal{I}$ such that $\constrainttermgroup_{(\scaleinside_{\max},\scaleoutside_{\max})} = \constrainttermgroup_{(\scaleinside_{\max},\scaleoutside_{\max}), c, \vecx} = \constrainttermgroup_{\pmax,c,\vecx}$.
    The idea is to show that for all considered $f$, when $|f| \to \infty$, then the absolute value of the addend belonging to $(\scaleinside_{\max},\scaleoutside_{\max})$ in \cref{eq:complete_constraint_in_grouped_normal_form} will outgrow the absolute value of all other addends combined, thereby determining the sign of ${(\Val_{\vecx}(f))}_{c}$.
    As this addend's sign is the sign of $v(f)$, this implies the lemma.
    From \cref{eq:complete_constraint_in_grouped_normal_form} we obtain
    \begin{equation}
      \begin{split}
        {\left( \Val_{\vecx}(f) \right)}_{c}
        = \!\!\!\!\!
        \sum_{\substack{(\scaleinside,\scaleoutside) \in \mathcal{I}\\ \scaleinside < \scaleinside_{\max}}}\!\!\!\!\! {\left(\scaleinside \cdot \scaleoutside^{\URVar(f)}\right)}^{|f|} \!\!\!\!\!\!\!\! \sum_{i \in \constrainttermgroup_{(\scaleinside,\scaleoutside),c,\vecx}}\!\!\!\!\!\!\!\! \zeta_{i,\matA}^{|f|_{\symMat{A}}} \zeta_{i,\matB}^{|f|_{\symMat{B}}} \gamma_{c,i} (\vecx) \\
        + \sum_{\scaleoutside \in \RR_{>0}} {(\scaleinside_{\max} \cdot \scaleoutside^{\URVar(f)})}^{|f|} \!\!\!\!\!\!\!\! \sum_{i \in \constrainttermgroup_{(\scaleinside_{\max},\scaleoutside),c,\vecx}}\!\!\!\!\!\!\!\! \zeta_{i,\matA}^{|f|_{\symMat{A}}} \zeta_{i,\matB}^{|f|_{\symMat{B}}} \gamma_{c,i} (\vecx)
      \end{split}
      \label{eq:proof_eventual_domination_main_part_eq0}
    \end{equation}

    In the first step, we show that the absolute value of the addend belonging to $(\scaleinside_{\max},\scaleoutside_{\max})$ is larger than the absolute value of the rest of the second line of \cref{eq:proof_eventual_domination_main_part_eq0}.
    To this end, we choose $r \in \NN$ large enough such that $t \cdot \scaleoutside_{\max}^{r} > \scaleoutside_{\max 2}^{r} \cdot \sum_{i\in[n]} |\gamma_{c,i} (\vecx)|$ for some $t\in(0,\rho)$, where $\scaleoutside_{\max 2} < \scaleoutside_{\max}$ is the second largest number such that $\constrainttermgroup_{(\scaleinside_{\max},\scaleoutside_{\max 2}),c,\vecx} \neq \emptyset$ or $\scaleoutside_{\max 2} = 0$ if no such value exists.
    Thus, $t \cdot \scaleoutside_{\max}^{\URVar(f) \cdot |f|} > \scaleoutside_{\max 2}^{\URVar(f) \cdot |f|} \cdot \sum_{i\in[n]} |\gamma_{c,i} (\vecx)|$.
    Now, for all $f$ considered by \cref{lem:domination_of_eventually_dominating_constraint_term_groups}, the conditions on $r$ and $t$ imply
    \begin{align*}
      \rho \cdot \scaleinside_{\max}^{|f|} \cdot \scaleoutside_{\max}^{\URVar(f)\cdot|f|} & > t \cdot \scaleinside_{\max}^{|f|} \cdot \scaleoutside_{\max}^{\URVar(f)\cdot|f|} \\
                                                                                          & > \scaleinside_{\max}^{|f|} \cdot \scaleoutside_{\max 2}^{\URVar(f)\cdot|f|} \cdot \sum_{i\in[n]} |\gamma_{c,i} (\vecx)| \\
                                                                                          &
      \geq | \smashoperator{\sum_{\scaleoutside \in \RR_{>0},\, \scaleoutside < \scaleoutside_{\max}}} \,\,\scaleinside_{\max}^{|f|} \cdot \scaleoutside^{\URVar(f)\cdot|f|} \,\,\,\,\,\,\,\smashoperator{\sum_{i \in \constrainttermgroup_{(\scaleinside_{\max},\scaleoutside),c,\vecx}}}
      \zeta_{i,\matA}^{|f|_{\symMat{A}}} \zeta_{i,\matB}^{|f|_{\symMat{B}}}
      \gamma_{c,i} |,
    \end{align*}
    so that the absolute value of the addend with $\scaleoutside = \scaleoutside_{\max}$ in the second line of \cref{eq:proof_eventual_domination_main_part_eq0} is larger than the absolute value of all other addends of this second sum combined.
    Note that while the variable $\scaleoutside$ from the sum in the last line ranges over an uncountable set, only finitely many of its addends are non-zero as for such non-zero addends the pairs $(\scaleinside_{\max},\scaleoutside)$ must be contained in the finite set $\mathcal{I}$.
    Thus, for all such $f$, the sign of the second line equals that of $v(f)$.

    In the second step, we show that the absolute value of the second line is larger than the absolute value of the first line of \cref{eq:proof_eventual_domination_main_part_eq0}.
    As $t < \rho$, there is some constant $C > 0$ such that the absolute value of the second line is at least ${C \cdot (\scaleinside_{\max} \cdot \scaleoutside_{\max}^{\URVar(f)})}^{|f|}$, i.e., the absolute value of the second line increases at least linearly in ${(\scaleinside_{\max} \cdot \scaleoutside_{\max}^{\URVar(f)})}^{|f|}$.
    As $\scaleinside_{\max}$ is the ``maximal $\scaleinside$'', the absolute value of the second line will outgrow that of all other addends in the first line of \cref{eq:proof_eventual_domination_main_part_eq0} combined as $|f| \to \infty$, for small enough $\URVar(f) \leq \eps$.
    Therefore, the sign of the second line (and thus, the sign of $v(f)$) determines the sign of ${(\Val_{\vecx}(f))}_{c}$.

    To that end, one chooses $\eps > 0$ suitably small such that $\max_{u\in[0,\eps]}
      \frac{\scaleinside \cdot \scaleoutside^{u}}{\scaleinside_{\max} \cdot \scaleoutside_{\max}^{u}} < 1$ for all $(\scaleinside,\scaleoutside) \in \mathcal{I}$ with $\scaleinside \neq \scaleinside_{\max}$ and $\constrainttermgroup_{(\scaleinside,\scaleoutside),c,\vecx} \neq \emptyset$ (i.e., $\scaleinside < \scaleinside_{\max}$).
  \end{proof}
}

%% file: illustration_full.tex
{
\begin{tikzpicture}[scale=1.3,every node/.style={font=\tiny}]
  \draw [draw=none, fill=SeaGreen!5] (2,0.75) -- (5,0.75) -- (5,2) -- (4,2) -- (2,1) -- (2,0.75);
  \draw [draw=none, line space=5pt,pattern=my north east lines, pattern color=SeaGreen!30] (2,0.75) -- (5,0.75) -- (5,2) -- (4,2) -- (2,1) -- (2,0.75);

  \draw [draw=none, fill=SeaGreen!5] (2,-0.75) -- (5,-0.75) -- (5,-2) -- (4,-2) -- (2,-1) -- (2,-0.75);
  \draw [draw=none,line space=5pt,pattern=my north east lines, pattern color=SeaGreen!30] (2,-0.75) -- (5,-0.75) -- (5,-2) -- (4,-2) -- (2,-1) -- (2,-0.75);

  \draw [->,line width=0.6pt] (-1,0)--(5,0) node[below]{\footnotesize $|f|$};
  \draw [->,line width=0.6pt] (0,-2)--(0,2) node[left]{\footnotesize $\URVar(f) \cdot |f|$};

  \draw[domain=0:3, smooth, variable=\x, gray] plot ({\x}, {0.6875*\x*\x*\x - 3.313*\x*\x + 4.125*\x});
  \node[gray] at (0.8, 1.65) (A) {$\URVar(f)\cdot|f|$};
  \node[gray] at (3,1.125)[circle,fill,inner sep=0.75pt]{};

  \draw [red] (0,0) -- (4,2) node[below]{$\eps |f|$};
  \draw [red] (0,0) -- (4,-2) node[above]{$-\eps |f|$};

  \draw [blue] (-0.75, 0.75) node[below,xshift=0.25cm]{$ r$} -- (5, 0.75);
  \draw [blue] (-0.75,-0.75) node[above,xshift=0.25cm]{$-r$} -- (5,-0.75);

  \draw [Green] (2,-2) node[below]{$l$} -- (2,2);

\end{tikzpicture}
}

%% file: positive_eigenvalues.tex
\section{Positive Eigenvalues}
\label{sect:Positive Eigenvalues}
Recall that we are interested in $\sign({\left( \Val_{\vecx}(f) \right)}_{c})$ as the execution progresses, i.e., for $|f| \to \infty$.
By \Cref{lem:domination_of_eventually_dominating_constraint_term_groups}, to this end we have to consider the sign of $v(f) = \sum_{i \in \constrainttermgroup_{d,c,\vecx}} \zeta_{i,\matA}^{|f|_{\symMat{A}}} \zeta_{i,\matB}^{|f|_{\symMat{B}}} \gamma_{c,i} (\vecx)$, where $\zeta_{i,\matA} = \frac{a_i}{|a_i|}$ and $\zeta_{i,\matB} = \frac{b_i}{|b_i|}$.
If both $a_i$ and $b_i$ are positive reals, then for $|f| \to \infty$, the sign of $\zeta_{i,\matA}^{|f|_{\symMat{A}}} \zeta_{i,\matB}^{|f|_{\symMat{B}}}
  \gamma_{c,i}(\vecx)$ does not change. Thus, we now investigate $\sum_{i \in \constrainttermgroup_{d,c,\vecx}}
  \zeta_{i,\matA}^{|f|_{\symMat{A}}} \zeta_{i,\matB}^{|f|_{\symMat{B}}} \gamma_{c,i}
  (\vecx)$ restricted to all $i \in \constrainttermgroup_{d,c,\vecx}$ where $a_i$ or $b_i$ is not from $\AA_{>0}$. In \Cref{lem:braverman_generalisation}, we will show that this sum is either always 0 (for all paths $f$) or it becomes negative for large enough $|f|$.

Assume that for some constraint $c\in[m]$, $d\in\{\nmax,\pmax\}$, and input $\vec{x}$, some eigenvalue of each constraint term in the eventually dominating constraint term group is not positive real, i.e., for all $i\in\constrainttermgroup_{d,c,\vecx}$ one has $\zeta_{i,\matA} \neq 1$ or $\zeta_{i,\matB} \ne 1$.
Then, the sign of the real part of this constraint term will change throughout the program's execution (i.e., for $|f| \to \infty$).
\Cref{lem:braverman_generalisation} shows that if the sum of these constraint terms is not always 0, then irrespective of the updates that were already performed in previous iterations, this sum becomes smaller than some negative constant $C$ after a number of further iterations.
This is expressed in \Cref{lem:braverman_generalisation}(\labelcref{it:braverman_generalisation_itemB}), where we have already performed $j_0$ updates with the matrix $\matA$ and $k_0$ updates with the matrix $\matB$.
Then by extending the run long enough, the real part of the sum becomes smaller than a constant $C$ that does not depend on $j_0$ and $k_0$.
Our \Cref{lem:braverman_generalisation} is a generalization of a similar result by Braverman~\cite[Lemma 4]{braverman06} to products of orbits of complex units, i.e., to products of $\zeta^j$ for $|\zeta| = 1$.\footnote{Note that \Cref{lem:braverman_generalisation} allows us to check whether $\zeta_{i,\matA}^{|f|_{\symMat{A}}}\zeta_{i,\matB}^{|f|_{\symMat{B}}}\gamma_{c,i}(\vecx) = 0$ holds for all $f \in \Path$, see \Cref{footnote:comment_on_constrainttermgroups}.
By \Cref{rem:sums_of_constraint_term_groups_are_real_valued}, the sum $\sum_{{i\in\constrainttermgroup_{(\scaleinside,\scaleoutside)}}}
  \zeta_{i,\matA}^{|f|_{\symMat{A}}}\zeta_{i,\matB}^{|f|_{\symMat{B}}}\gamma_{c,i}(\vecx)$ is a real number and thus, in the case of (\labelcref{it:braverman_generalisation_itemA}), $\sum_{{i\in\constrainttermgroup_{(\scaleinside,\scaleoutside)}}}
  \zeta_{i,\matA}^{|f|_{\symMat{A}}}\zeta_{i,\matB}^{|f|_{\symMat{B}}}\gamma_{c,i}(\vecx) = 0$ holds for all $f \in \Path$. So given an actual input $\vecx$, one just has to check the condition of \Cref{lem:braverman_generalisation}(\labelcref{it:braverman_generalisation_itemA}). If that condition does not hold, then by \Cref{lem:braverman_generalisation}(\labelcref{it:braverman_generalisation_itemB}), $\sum_{{i\in\constrainttermgroup_{(\scaleinside,\scaleoutside)}}}
  \zeta_{i,\matA}^{|f|_{\symMat{A}}}\zeta_{i,\matB}^{|f|_{\symMat{B}}}\gamma_{c,i}(\vecx) = 0$ does not hold for every $f \in \Path$.}

\begin{restatable}[Coefficients of Complex Eigenvalues Become Negative]{lemma}{bravermangeneralisation}
  \label{lem:braverman_generalisation}
  Let $\gamma_{1},\ldots,\gamma_{l} \in \CC$ be complex numbers and let $\zeta_{1,1},\ldots,\zeta_{1,l},\zeta_{2,1},\ldots,\zeta_{2,l} \in \{z \in \CC \mid |z| = 1\}$ be complex units such that $\zeta_{1,i}\neq1$ or $\zeta_{2,i}\neq1$ for all $i\in\{1,\ldots,l\}$.
  For all $j,k \in \NN$, let $z_{j,k} = \sum_{i=1}^{l} \zeta_{1,i}^{j} \zeta_{2,i}^{k} \gamma_{i}$.
  If all tuples $(\zeta_{1,i},\zeta_{2,i})$ for $i\in\{1,\ldots,l\}$ are pairwise different, then there exist constants $C \in \AA_{< 0}$ and $K\in\NN$ such that we either have (\labelcref{it:braverman_generalisation_itemA}) or (\labelcref{it:braverman_generalisation_itemB}):
  \begin{alphaenumerate}
    \item For all $i\in[l]$ with $\gamma_{i} \neq 0$ there is some $i'\in[l]$ with $\zeta_{1,i} = \conj{\zeta_{1,i'}}$, $\zeta_{2,i} = \conj{\zeta_{2,i'}}$, and $\gamma_{i} = -\conj{\gamma_{i'}}$, which implies $\Re(z_{j,k}) = \frac{z_{j,k} + \conj{z_{j,k}}}{2} = 0$ for all $j,k\in\NN$.\label{it:braverman_generalisation_itemA}
    \item For all $j_{0},k_{0}\in\NN$ there exist $j,k \in \{0,\ldots,K\}$ such that $\Re(z_{j_{0}+j,k_{0}+k}) \leq C$ and there are $j, j', k, k' \in \NN$ such that $z_{j,k} \neq z_{j',k'}$.\label{it:braverman_generalisation_itemB}
  \end{alphaenumerate}
\end{restatable}

\begin{example}
  To illustrate \Cref{lem:braverman_generalisation}, we continue \Cref{ex:eventually_dominating_constraint_term_groups} and consider the constraint term group $\constrainttermgroup_{(10\sqrt{2},\nicefrac{1}{2})} = \{1,2\}$.
  Let $\vecx \in \AA^{3}$.
  According to \Cref{ex:normal_form}, the coefficient of this constraint term group (for $c = 1$) and $f \in \Path$ is
  \begin{align*} \textstyle
    \sum_{i\in\constrainttermgroup_{(10\sqrt{2},\nicefrac{1}{2})}} \zeta_{i,\matA}^{|f|_{\symMat{A}}} \zeta_{i,\matB}^{|f|_{\symMat{B}}} \gamma_{1,i}(\vecx)
     & \textstyle = {(e^{-\im \arctan (\nicefrac{4}{3})})}^{|f|_{\symMat{A}}} {(e^{\im (\pi - \arctan(\nicefrac{4}{3}))})}^{|f|_{B}} \gamma_{1,1} (\vecx) \\
     & \textstyle \phantom{{}={}} + {(e^{\im \arctan{\nicefrac{4}{3}}})}^{|f|_{\symMat{A}}} {(e^{\im (\arctan(\nicefrac{4}{3}) - \pi)})}^{|f|_{\symMat{B}}} \gamma_{1,2} (\vecx)
  \end{align*}

  As already
  \pagebreak[2]
  explored in \Cref{ex:constraint_term_groups}, this coefficient is $0$ for all $f \in \Path$ whenever $\gamma_{1,1} (\vecx) = 0 \iff \gamma_{1,1} (\vecx) = - \conj{\gamma_{1,2} (\vecx)}$ which corresponds to \Cref{lem:braverman_generalisation}(\labelcref{it:braverman_generalisation_itemA}).
  On the other hand, (\labelcref{it:braverman_generalisation_itemB}) states\linebreak
  that whenever this is not the case, i.e., $\gamma_{1,1} (\vecx) \neq 0$, then there are constants $C \in \AA_{<0}, \, K \in \NN$ such that for every $f\in\Path$ there are $f',f'' \in \Path$ with $|f'|_{\symMat{A}}, |f'|_{\symMat{B}} \leq K$ where
  \begin{equation*}
    \textstyle
    0 \;\; > \;\; C \;\; \geq \;\; \sum_{i\in\constrainttermgroup_{(10\sqrt{2},\nicefrac{1}{2})}} \zeta_{i,\matA}^{|ff'|_{\symMat{A}}} \zeta_{i,\matB}^{|ff'|_{\symMat{B}}} \gamma_{1,i}(\vecx)
    \;\; \neq \;\; \sum_{i\in\constrainttermgroup_{(10\sqrt{2},\nicefrac{1}{2})}} \zeta_{i,\matA}^{|ff''|_{\symMat{A}}} \zeta_{i,\matB}^{|ff''|_{\symMat{B}}} \gamma_{1,i}(\vecx) ,
  \end{equation*}
  which ends our example to illustrate \cref{lem:braverman_generalisation}.
\end{example}

Let $\mathfrak{R} = \{i\in [n] \mid \zeta_{i,\matA} = \zeta_{i,\matB} = 1\}$ be the set of indices $i$ such that both eigenvalues $a_{i}$ and $b_{i}$ are positive reals and let $\mathfrak{C} = [n] \setminus \mathfrak{R} = \{i\in[n] \mid \zeta_{i,\matA} \neq 1 \text{ or }
  \zeta_{i,\matB}\neq1\}$ be the set where at least one of the eigenvalues is not a positive real. To simplify the notation we also denote $\constrainttermgroup_{d,c,\vecx}\cap\mathfrak{R}$ and $\constrainttermgroup_{d,c,\vecx}\cap\mathfrak{C}$ by $\mathfrak{R}_{d,c,\vecx}$ and $\mathfrak{C}_{d,c,\vecx}$, respectively, for $(d,c)\in\{\nmax,\pmax\}\times[m]$. So for all $i \in \mathfrak{R}_{d,c,\vecx}$, we have $\zeta_{i,\matA}^{|f|_{\symMat{A}}} \zeta_{i,\matB}^{|f|_{\symMat{B}}}
  \gamma_{c,i}
  (\vecx) = \gamma_{c,i}
  (\vecx)$, i.e., the sign of the corresponding addend does not change for $|f| \to \infty$. For the other eigenvalues, by \Cref{lem:braverman_generalisation}, $\sum_{i \in \mathfrak{C}_{d,c,\vecx}}
  \zeta_{i,\matA}^{|f|_{\symMat{A}}} \zeta_{i,\matB}^{|f|_{\symMat{B}}}
  \gamma_{c,i}
  (\vecx)$ is either always 0 (for all paths $f$) or it becomes negative for suitable $|f|_{\symMat{A}}$ and $|f|_{\symMat{B}}$.

When executing the loop $\program$ on input $\vecx$, one expects that eventually (for $|f| \to \infty$) the constraint term group $\constrainttermgroup_{d,c,\vecx}$ for either $d=\nmax$ or $d=\pmax$ dominates the sign of constraint $c\in[m]$ (\Cref{lem:domination_of_eventually_dominating_constraint_term_groups}).
Whenever $\sum_{i \in \mathfrak{R}_{d,c,\vecx}} \gamma_{c,i} (\vecx) \leq 0$ and $\URVar(f) \cdot |f|$ has reached one of the two ``safe'' regions marked in green in \Cref{fig:safe}, by \Cref{lem:braverman_generalisation}
one can extend the current path $f$ by a path $g_f$ such that the coefficient $v(f \, g_f) = \sum_{i\in\constrainttermgroup_{d,c,\vecx}} \zeta_{i,\matA}^{|f\, g_f|_{\symMat{A}}} \zeta_{i,\matB}^{|f \, g_f|_{\symMat{B}} }\gamma_{c,i}
  (\vecx)$ of the dominating constraint term group is negative. Thus, the execution $f$ can be extended by a path $g_f$ such that it leads to termination. This observation is captured in \Cref{lem:braverman_generalisation_corollary}.

\begin{restatable}[Finite Execution leading to Termination]{lemma}{finiteexecutionleadingtotermination}
  \label{lem:braverman_generalisation_corollary}
  Let $c \in [m]$, $\vecx\in\AA^{n}$, and $d\in\{\nmax,\pmax\}$, such that $\sum_{i \in \mathfrak{R}_{d,c,\vecx}} \gamma_{c,i}(\vecx) \leq 0$.
  Then there are constants $\eps \in \AA_{>0}$, $r,u\in\NN$, and $l \in \NN_{>0}$, such that for all $f\in\Path$ with $|f| \geq l$, $|\URVar(f) \cdot |f|| \geq r$, $\URVar(f) \in [-\eps,0]$ if $d=\nmax$, and $\URVar(f) \in [0,\eps]$ if $d=\pmax$, there is a finite execution $g_f \in \Path$ of length $|g_f| \leq u$ with ${\left(\Val_{\vecx}(f\,g_f)\right)}_{c} \leq 0$.
\end{restatable}

Finally, \Cref{lem:dual_positive_eigenvalues_for_eventually_dominating_constraints}
builds upon \Cref{lem:braverman_generalisation_corollary}
and gives a sufficient criterion for termination of an input $\vecx\in\AA^{n}$.
The negation of this criterion is a necessary criterion for every input $\vecx\in\AA^{n}$ that is eventually non-terminating.
This necessary criterion states that if $\vecx$ is eventually\linebreak
non-terminating, then for all constraints $c$, the sum $\sum_{i\in\mathfrak{R}_{d,c,\vecx}} \gamma_{c,i}(\vecx)$ of the addends for the ``dual positive eigenvalues'' (where both $a_i$ and $b_i$ are positive reals) must be positive.

So whenever $\sum_{i \in \mathfrak{R}_{d,c,\vecx}} \gamma_{c,i} (\vecx) \leq 0$, \Cref{lem:dual_positive_eigenvalues_for_eventually_dominating_constraints} states that the expected number of steps until $\URVar(f) \cdot |f|$ reaches a ``safe'' (green) area in \Cref{fig:safe}
and executes $g_f$ afterwards is finite.
In other words, the expected number of steps $\E (\LRVar_{\vecx})$ until termination is finite.

\begin{example}
  \label{ex:dual_positive_eigenvalues_for_eventually_dominating_constraints}
  To motivate \Cref{lem:dual_positive_eigenvalues_for_eventually_dominating_constraints}
  further, we continue \Cref{ex:eventually_dominating_constraint_term_groups}.
  Let $\vecx \in \AA^{3}$.
  We have $\mathfrak{R} = \{ 3 \}$, as only $a_3$ and $b_3$ are positive real eigenvalues.

  First, suppose $\constrainttermgroup_{\nmax,1,\vecx} \neq \constrainttermgroup_{(10\sqrt{2},2)} = \{ 3 \}$ and $\constrainttermgroup_{\pmax,1,\vecx} \neq \constrainttermgroup_{(10\sqrt{2},2)} = \{ 3 \}$.
  Then, for all $d \in\linebreak
    \{\nmax,\pmax\}$ we have $\mathfrak{R}_{d,1,\vecx} = \emptyset$ and hence $\sum_{i \in \mathfrak{R}_{d,1,\vecx}} \gamma_{1,i} (\vecx) = 0$. Thus, $\vecx \not \in \ENT$ by \Cref{lem:dual_positive_eigenvalues_for_eventually_dominating_constraints}.

  On the other hand, if for some $d\in\{\nmax,\pmax\}$ we have $\constrainttermgroup_{d,1,\vecx} = \constrainttermgroup_{(10\sqrt{2},2)}= \{ 3 \}$, i.e., $\mathfrak{R}_{d,1,\vecx} = \{3\}$, then \Cref{lem:dual_positive_eigenvalues_for_eventually_dominating_constraints}
  states that $\gamma_{1,3} (\vecx) \leq 0$ (i.e., $ \frac{11}{10}x_{1} + \frac{11}{10}x_{2} \leq 0$) implies $\vecx \not \in \ENT$.
  However, if $\gamma_{1,3} (\vecx) > 0$, then \Cref{lem:dual_positive_eigenvalues_for_eventually_dominating_constraints} does not make any statement about whether $\vecx \in \ENT$ or $\vecx \not \in \ENT$.\footnote{\Cref{sec:towards_witnesses_for_non-termination} will show that in this case one indeed has $\vecx \in \ENT$, see \Cref{ex:witnesses}.}
\end{example}

\begin{restatable}[Dual Positive Eigenvalues for Eventually Dominating Constraints]{lemma}{dualpositiveeigenvaluesforeventuallydominatingconstraints}
  \label{lem:dual_positive_eigenvalues_for_eventually_dominating_constraints}
  Let $\vecx \in \AA^{n}$.
  If for every $d\in\{\nmax,\pmax\}$ there is a $c\in[m]$ with $\sum_{i\in\mathfrak{R}_{d,c,\vecx}} \gamma_{i}(\vecx) \leq 0$, then $\vecx \notin \ENT$.
  Thus, if $\vecx \in\linebreak
    \ENT$, then there is some $d\in\{\nmax,\pmax\}$ such that for all $c\in[m]$ we have $\sum_{i\in\mathfrak{R}_{d,c,\vecx}} \gamma_{c,i}(\vecx) > 0$.
\end{restatable}
\report{
  \begin{proof}[Proof Sketch]
    The full proof can be found in \cref{app:proofs_for_positive_eigenvalues}.
    \Cref{lem:dual_positive_eigenvalues_for_eventually_dominating_constraints} is proved by contradiction.
    To that end, let $\vecx\in\ENT$ and let $d\in\{\nmax,\pmax\}$ such that for all $c\in[m]$ we have $\sum_{i\in\mathfrak{R}_{d,c,\vecx}} \gamma_{c,i}(\vecx) > 0$.

    W.l.o.g., we may assume $\vecx \in \NT$ as otherwise one could consider $\vecy = A^{j}B^{k}\vecx \in \NT$ for some $j,k \in \NN$ according to the definition of $\ENT$.
    This is because $\sum_{i\in\mathfrak{R}_{d,c,\vecy}} \gamma_{c,i}(\vecy) > 0$ implies $\sum_{i\in\mathfrak{R}_{d,c,\vecx}} \gamma_{c,i}(\vecx) > 0$ (see \cref{app:proofs_for_positive_eigenvalues}).

    For a proof by contradiction, we assume the existence of $c_{\nmax},c_{\pmax} \in [m]$ such that
    \[\sum_{i\in\mathfrak{R}_{\nmax,c_{\nmax},\vecx}} \gamma_{c_{\nmax},i}(\vecx) \leq 0
      \quad \text{and} \quad
      \sum_{i\in\mathfrak{R}_{\pmax,c_{\pmax},\vecx}} \gamma_{c_{\pmax},i}(\vecx) \leq 0.\]

    We denote the constants from \Cref{lem:braverman_generalisation_corollary} as $\eps_{\nmax}, r_{\nmax}, l_{\nmax}, u_{\nmax}$ for $c = c_{\nmax}, d=\nmax$ and $\eps_{\pmax},r_{\pmax},l_{\pmax},u_{\pmax}$ for $c = c_{\pmax}, d=\pmax$.
    We define $\eps = \min\{\eps_{\nmax},\eps_{\pmax}\}$, $r = \max\{r_{\nmax},r_{\pmax}\}$, $u = \max\{u_{\nmax},u_{\pmax}\}$, and choose $l \geq \max\{l_{\nmax},l_{\pmax},u\}$ large enough such that $\eps \cdot l \geq r+1$.\footnote{This ensures that if $|f|$ is increased by 1, then the value of $\URVar(f) \cdot |f|$ changes by at most 1.
      Therefore, if $\URVar(f) \cdot |f| \geq r$ for the first time, then we also have $\URVar(f) \cdot |f| \leq \eps \cdot |f|$, i.e, $\URVar(f) \cdot |f|$ has reached a safe region.}
    Thus, for every path $f\in\Path$ leading to one of the two ``safe regions'', i.e., $|f| \geq l$, $|\URVar(f) \cdot |f|| \geq r$, and $\URVar(f) \in [-\eps,\eps]$, there is some $g_f$ with $|g_f| \leq u$ such that ${(\Val_{\vecx} (f \, g_f))}_{c_{\nmax}} \leq 0$ or ${(\Val_{\vecx} (f \, g_f))}_{c_{\pmax}} \leq 0$ and hence the path $f \,g_f$ is terminating.

    For all $i\in\NN$, we introduce the ($\F$-measurable) random variables $B_{i}\colon \Runs \to \NN \cup \{\infty\}$ such that $B_{i}$ maps runs to an index where one is in one of the ``safe regions'' specified by the constants $\eps$, $l$, and $r$ for at least the $i$-th time: \input{./dual_positive_eigenvalues_for_eventually_dominating_constraints_def_Xi.tex}
    We ensure that $B_{i}(\run_1 \ldots) - B_{i-1}(\run_1 \ldots) \geq l \geq u$ to provide enough ``space'' for a terminating continuation $g$ of length $|g| \leq u$ to occur between the indices $B_{i}(\run_1 \ldots)$ and $B_{i-1}(\run_1 \ldots)$.
    Every time one is in a safe region, there is the possibility to execute a terminating path of length at most $l$.
    So for every path $f$ that leads to some index $B_{i}$, and thus to a safe region, there is a probability of at least ${(\min(p,1-p))}^{l} > 0$ that this path terminates within its first $B_{i} + l$ steps.
    If $f$ is not followed by $g_{f}$, then one might not terminate and even leave the safe region.
    However, in this case it only takes a constant number of steps (on expectation) to re-enter a safe region, i.e., to reach the index $B_{i+1}$.
    Then one has again a probability of ${(\min(p,1-p))}^{l} > 0$ to terminate in the following $l$ steps, by executing a terminating path $g_{f}$ next.
    So while the expectations of indices $B_{i}$ grow linearly in $i$, the probability of non-termination decreases exponentially in $i$, i.e., it is bounded by ${(1 - {(\min(p,1-p))}^{l})}^{i}$.
    The full proof in \cref{app:proofs_for_positive_eigenvalues} formalizes this argument to show $\E(\LRVar_{\vecx}) < \infty$.
    Hence, $\vecx \not \in \ENT$, contradicting our initial assumption.
  \end{proof}
}

%% file: dual_positive_eigenvalues_for_eventually_dominating_constraints_def_Xi.tex
\begin{align*}
  B_{0}(\run_1 \ldots) & = 0 \\
  B_{i}(\run_1 \ldots) & = \min \{j \in \NN \mid j > B_{i-1}(\run_1 \ldots) + l,\; |\URVar(\run_{1}\ldots\run_{j}) \cdot j| \geq r,\; \URVar (\run_{1}\ldots\run_{j}) \in [-\eps,\eps] \}
\end{align*}

%% file: towards_ent_witnesses.tex
\section{Towards Non-Termination Witnesses}
\label{sec:towards_witnesses_for_non-termination}

\Cref{lem:dual_positive_eigenvalues_for_eventually_dominating_constraints} provides a necessary condition that must hold for all $\vecx\in\ENT$.
It requires that the\linebreak
sum of the addends $\gamma_{c,i}(\vecx)$ for all positive real eigenvalues $a_i, b_i$ must be $> 0$.
This condition is however not sufficient for $\vecx\in\ENT$.
To turn this into a sufficient criterion, we now increase the lower bound $0$.
More precisely, we replace $0$ by the sum of the addends $|\gamma_{c,i} (\vecx)|$ for all those eigenvalues where $a_i$ or $b_i$ are not a positive real number.
In this way, we obtain a sufficient (but no longer necessary) criterion for $\ENT$.
To turn this into a sufficient and necessary\linebreak
criterion, we then introduce a ``boosting lemma'' (\Cref{lem:boosting}), which states that if there is an\linebreak
input in $\ENT$, then there is also a (possibly different) input in $\ENT$ that satisfies our sufficient criterion.
To prove this boosting lemma, we need the necessary condition of \Cref{lem:dual_positive_eigenvalues_for_eventually_dominating_constraints}.

For our sufficient (but not necessary) condition for $\ENT$, we define the set of \emph{witnesses} for eventual non-termination as those inputs meeting this criterion.

\begin{definition}[Witnesses for Eventual Non-Termination]
  \label{def:witnesses_for_ent}
  We define the set $W = W_{\nmax} \cup W_{\pmax}$ of \emph{witnesses for eventual non-termination}, where for $d\in\{\nmax,\pmax\}$, we have
  \[ \textstyle
    W_{d} = \bigcap_{c\in[m]} \left\{\vecx \in \AA^{n} \mid
    \sum_{i\in\mathfrak{R}_{d,c,\vecx}} \gamma_{c,i}(\vecx) >
    \sum_{i\in\mathfrak{C}_{d,c,\vecx}} \left|\gamma_{c,i} (\vecx)\right| \right\}
    .
  \]
\end{definition}
Note that the sum on the left-hand side in the definition of $W_{d}$ is real-valued due to \Cref{lem:sums_of_conjugated_addends_are_real_valued}.

So $W_d$ are all inputs $\vecx$ where for all constraints $c$, the sum of the dominating addends $\gamma_{c,i} (\vecx)$ for positive real eigenvalues $a_i, b_i$ is greater than the sum of the $\left|\gamma_{c,i} (\vecx)\right|$ for the other eigenvalues $a_i, b_i$.
\Cref{lem:soundness_of_witnesses} shows that the witness condition of \Cref{def:witnesses_for_ent} is indeed a sufficient criterion for $\ENT$.
Before presenting this lemma, we will apply it to our running example.

\begin{example}
  \label{ex:witnesses}
  We continue \Cref{ex:dual_positive_eigenvalues_for_eventually_dominating_constraints}.
  Let $\vecx \in \AA^{3}$ with $\gamma_{1,3}(\vecx) > 0$ such that for some $d \in \{\nmax,\pmax\}$ we have $\constrainttermgroup_{d,1,\vecx} = \constrainttermgroup_{(10\sqrt{2},2)}$.
  Then we have $\constrainttermgroup_{d,1,\vecx} = \mathfrak{R}_{d,1,\vecx} \,\uplus\, \mathfrak{C}_{d,1,\vecx} = \{3\} \,\uplus\, \emptyset$ and thus $\sum_{i \in \mathfrak{R}_{d,1,\vecx}} \gamma_{1,i}(\vecx) = \gamma_{1,3} (\vecx) > 0 = \sum_{i \in \mathfrak{C}_{d,1,\vecx}} |\gamma_{1,i} (\vecx)|$.
  Hence, $\vecx \in W_{d} \subseteq W$ and thus by the following \Cref{lem:soundness_of_witnesses}, we obtain $\vecx \in \ENT$, answering the question from \Cref{ex:dual_positive_eigenvalues_for_eventually_dominating_constraints}.

  Hence, $\vecx \in \ENT$ iff $\gamma_{1,3} (\vecx) > 0$ and there is $d \in \{\nmax,\pmax\}$ with $\constrainttermgroup_{d,1,\vecx} = \constrainttermgroup_{(10\sqrt{2},2)}$.
  (The ``only if'' direction is due to \Cref{lem:dual_positive_eigenvalues_for_eventually_dominating_constraints}, see \Cref{ex:dual_positive_eigenvalues_for_eventually_dominating_constraints}.)
\end{example}

\begin{restatable}[Witness Criterion is Sufficient for Eventual Non-Termination]{lemma}{soundnessofwitnesses}
  \label{lem:soundness_of_witnesses}
  Let $\vecx \in \AA^{n}$ be a witness for eventual non-termination, i.e., $\vecx \in W$.
  Then we have $\vecx\in\ENT$, i.e., $\vecx$ is indeed an eventually non-terminating input.
\end{restatable}
\report{
  \begin{proof}[Proof Sketch]
    The full proof can be found in \cref{app:proofs_for_towards_witnesses_for_non-termination}.
    Let $\vecx \in W$.
    Thus, there is some $d\in\{\nmax,\pmax\}$ such that for all $c\in[m]$ we have
    \begin{equation}
      \label{realGreaterComplex main part}
      \sum_{i\in\mathfrak{R}_{d,c,\vecx}} \gamma_{c,i}(\vecx) > \sum_{i\in\mathfrak{C}_{d,c,\vecx}} |\gamma_{c,i}(\vecx)|.
    \end{equation}
    W.l.o.g., we assume $d=\pmax$ as the other case is symmetric.

    For all $c\in[m]$ and $f\in\Path$, we define $v_{c}(f) = \sum_{i\in\constrainttermgroup_{\pmax,c,\vecx}} \zeta_{i,\matA}^{|f|_{\symMat{A}}} \zeta_{i,\matB}^{|f|_{\symMat{B}}} \gamma_{c,i}(\vecx)$.
    We have
    \begin{align*}
      v_{c}(f) & = \sum_{i\in\mathfrak{R}_{\pmax,c,\vecx}} \gamma_{c,i}(\vecx) +
      \sum_{i\in\mathfrak{C}_{\pmax,c,\vecx}} \zeta_{i,\matA}^{|f|_{\symMat{A}}}
      \zeta_{i,\matB}^{|f|_{\symMat{B}}} \gamma_{c,i}(\vecx) \\
               & \geq \sum_{i\in\mathfrak{R}_{\pmax,c,\vecx}} \gamma_{c,i}(\vecx) -
      \left| \sum_{i\in\mathfrak{C}_{\pmax,c,\vecx}} \zeta_{i,\matA}^{|f|_{\symMat{A}}}
      \zeta_{i,\matB}^{|f|_{\symMat{B}}} \gamma_{c,i}(\vecx) \right| \\
               & \geq \sum_{i\in\mathfrak{R}_{\pmax,c,\vecx}} \gamma_{c,i}(\vecx) -
      \sum_{i\in\mathfrak{C}_{\pmax,c,\vecx}} \left|\zeta_{i,\matA}^{|f|_{\symMat{A}}}
      \zeta_{i,\matB}^{|f|_{\symMat{B}}}\right| \cdot \left|\gamma_{c,i}(\vecx) \right| \\
               & = \sum_{i\in\mathfrak{R}_{\pmax,c,\vecx}} \gamma_{c,i}(\vecx) -
      \sum_{i\in\mathfrak{C}_{\pmax,c,\vecx}} \left|\gamma_{c,i}(\vecx) \right|
    \end{align*}
    where we used the triangle inequation and that $|\zeta_{i,\matA}| = |\zeta_{i,\matB}| = 1$ for all $i\in[n]$.
    For all $c\in[m]$, we define $\rho_{c}$ as the value of this inequation's right-hand side, where $\rho_c > 0$ by \cref{realGreaterComplex main part}.

    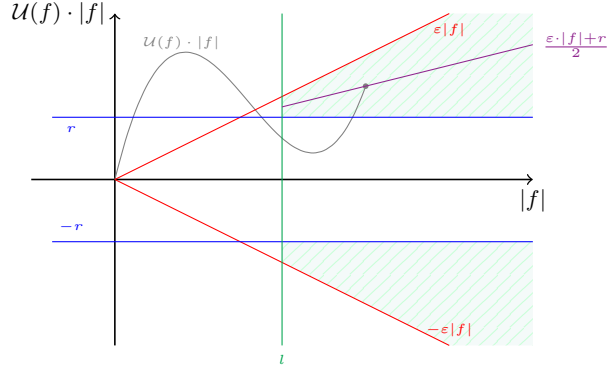
\begin{figure}
      \begin{center}
        \input{illustration_full_middle.tex}
      \end{center}
      \caption{\label{fig:safe_with_middle_sketch} Illustration of the General Idea Using Safe Regions with Middle Line}
    \end{figure}

    Then, by \Cref{lem:domination_of_eventually_dominating_constraint_term_groups}, for all $c\in[m]$ there are constants $\eps_{c} \in \AA_{>0}$, $r_{c} \in \NN$, and $l_{c} \in \NN_{>0}$ such that $\sign({(\Val_{\vecx} (f))}_{c}) = \sign(v_{c}(f)) = 1$ for all $f \in \Path$ with $|f| \geq l_c$, $\URVar(f) \cdot |f| \geq r_c$, and $\URVar\in[0,\eps_{c}]$.
    We select $r\in\NN$ with $r \geq \max\{r_{1},\ldots, r_{m}\}$, $\eps \in (0,1-p]$ with $\eps \leq \min \{\eps_{1},\ldots,\eps_{m}\}$, and $l \geq \max\{l_{1}, \ldots,l_{m}\}$ large enough such that $\eps l \geq r+2$.
    Then, $\Val_{\vecx} (f) > \veczero$ whenever $|f| \geq l$, $\URVar(f) \cdot |f| \geq r$, and $\URVar(f) \in [0,\eps]$.
    The intuition is that we regard the line $\frac{\eps \cdot |f| + r}{2}$ in the middle of the upper safe region (the purple line in \Cref{fig:safe_with_middle_sketch}) and the \emph{middle region} which is between 1 above and 1 below this middle line.
    We need $\eps \leq 1-p$ to ensure that $\URVar(f) \cdot |f|$ can reach the whole (upper) safe region (recall that $\URVar(f) \leq 1-p$).
    Moreover, we need $\eps l \geq r+2$ (i.e., $\eps \cdot l - r \geq 2$) to ensure that the whole middle region is in the safe region.

    Let $\hat{f} \in \Path$ satisfy these requirements such that additionally $\abs{\URVar(\hat{f}) \cdot |\hat{f}| - \frac{\eps |\hat{f}| + r}{2}} \leq 1$, i.e., $\URVar(\hat{f}) \cdot |\hat{f}|$ reaches the middle region.
    Such an $\hat{f}$ always exists (see \cref{app:proofs_for_towards_witnesses_for_non-termination}).

    Consider the ($\F$-measurable) random variable $T\colon \Runs \to \NN\cup\{\infty\}$ defined as \input{./ent_witnesses_definition_T}
    So $T$ maps any run $\run$ to the smallest $i$ such that if $\hat{f}$ takes place before $\run$, then $\run$ reaches an unsafe region after performing $i$ steps from $\hat{f}$ (according to \Cref{fig:safe_with_middle_sketch}).
    The expectation of the random variable $T$ is infinite, i.e.,
    \begin{equation}
      \E (T) = \infty. \label{eq:soundness_of_witnesses1 main part}
    \end{equation}
    A proof of \cref{eq:soundness_of_witnesses1 main part} can be found in \Cref{app:proofs_for_towards_witnesses_for_non-termination}.

    Let $\vecy = \matA^{|\hat{f}|_{\symMat{A}}} \matB^{|\hat{f}|_{\symMat{B}}} \vecx$, i.e., $\vecy$ results from $\vecx$ by executing the program according to the path $\hat{f}$.
    We now show $\vecy \in \NT$, which implies $\vecx \in \ENT$.

    If $T(\run) = k$ for some arbitrary run $\run\in\Runs$ and $k\in\NN\cup\{\infty\}$, then $\URVar(\hat{f}\run_{1}\ldots\run_{i}) \cdot i \geq r$ as well as $\URVar(\hat{f}\run_{1}\ldots\run_{i}) \in [0,\eps]$ for all $i\in\{1,\ldots,k-1\}$.
    Note that since $\URVar(\hat{f}\run_{1}\ldots\run_{i}) \geq 0$, this implies $\URVar(\hat{f}\run_{1}\ldots\run_{i}) \cdot (|\hat{f}| + i) \geq \URVar(\hat{f}\run_{1}\ldots\run_{i}) \cdot i \geq r$.
    It follows that $\Val_{\vecy} (\run_{1}\ldots\run_{i}) = \Val_{\vecx} (\hat{f}\run_{1}\ldots\run_{i}) > \veczero$.
    Therefore, $\LRVar_{\vecy} (\run) > i$ and hence, $\LRVar_{\vecy} (\run) \geq k$.
    As this holds for arbitrary runs, one concludes $\LRVar_{\vecy} \geq T$.
    Combined with \cref{eq:soundness_of_witnesses1 main part}, one obtains $\E \left(\LRVar_{\vecy}\right) \geq \E (T) = \infty$, completing the proof of \Cref{lem:soundness_of_witnesses}.
  \end{proof}

  \medskip

  In the proof of \cref{lem:soundness_of_witnesses}, the fact that we are interested in finite expected runtimes (positive almost sure termination) and not termination with probability $1$ (almost sure termination) is crucial, as in expectation, infinitely many steps are required to leave the safe region.
  However, the safe region will eventually be left with probability $1$.
  Thus, the constructed input $\vecy$ might in general terminate with probability $1$.
}

\Cref{def:witnesses_for_ent} introduced a set $W\subseteq\AA^{n}$ that, as shown by \Cref{lem:soundness_of_witnesses}, under-approximates the set of eventually non-terminating inputs $\ENT \supseteq W$.
While in general we may have $\ENT \supsetneq W$, for the program from \Cref{ex:witnesses}
we have $\ENT = W$ as for every $d\in\{\nmax,\pmax\}$ one either has $\constrainttermgroup_{d,1,\vecx} = \mathfrak{R}_{d,1,\vecx}$ or $\constrainttermgroup_{d,1,\vecx} = \mathfrak{C}_{d,1,\vecx}$.
So here, $\vecx \in \ENT$ implies $\sum_{i \in \mathfrak{R}_{d,1,\vecx}} \gamma_{1,i}(\vecx) > 0 \Rightarrow \mathfrak{R}_{d,1,\vecx} \neq \emptyset \Rightarrow \mathfrak{C}_{d,1,\vecx} = \emptyset$ for some $d\in\{\nmax,\pmax\}$ by \Cref{lem:dual_positive_eigenvalues_for_eventually_dominating_constraints} and hence $\sum_{i \in \mathfrak{R}_{d,1,\vecx}} \gamma_{1,i}(\vecx) > \sum_{i \in \mathfrak{C}_{d,1,\vecx}} |\gamma_{1,i}(\vecx)|$, i.e., $\vecx \in W$.

As the set $W$ is rather simple to characterize in contrast to $\ENT$, our goal is to only check for the existence of some $\vecx \in W$.
This input then witnesses the eventual non-termination of the loop $\program$.
The following \Cref{lem:boosting} establishes that whenever $\program$ is eventually non-terminating, then such a witness $\vecx \in W$ does indeed exist.
This then leads to our overall decision procedure, because we have that $\program$ is non-terminating $\iff$ $\program$ is eventually non-terminating $\iff$ $W \neq \emptyset$, see \Cref{CharacterizingTermination}.
In \Cref{sect:Deciding PAST}, we will show that emptiness of $W$ is decidable (not only over the algebraic reals, but also over different sub-semirings $\semiring$ of $\AA$ such as the naturals, integers, or rationals) and if $W \neq \emptyset$, then an element of $W$ is computable.

The intuition behind \cref{lem:boosting} is as follows: Given an input $\vecx \in \NT\cap\semiring^{n}$, we want to construct
\pagebreak[2]
a non-terminating input in $W\cap\semiring^{n}$.
Recall that for any constraint $c \in [m]$ and $d \in \{\nmax,\pmax\}$, the set $\mathfrak{R}_{d,c,\vecx}$ contains those indices $i \in \{1,\ldots,n\}$ from the dominant constraint term group where the corresponding eigenvalues $a_{i}$ and $b_{i}$ of both update matrices $\matA$ and $\matB$ are positive reals.
On the other hand, $\mathfrak{C}_{d,c,\vecx}$ contains the remaining indices from the dominant constraint term group.
Moreover, the $\gamma_{c,i}(\vecx)$ help to determine the sign of the corresponding dominant pair's coefficient.
If $x \notin W$, then
\begin{equation}
  \sum_{i \in \mathfrak{C}_{d,c,\vecx}} |\gamma_{c,i}(\vecx)| \geq \sum_{i \in \mathfrak{R}_{d,c,\vecx}} \gamma_{c,i}(\vecx) \label{eq:boosting_idea_motivation}
\end{equation}
for some $c \in [m]$.
One can now modify $\vecx$ to make \cref{eq:boosting_idea_motivation}'s left-hand side smaller.
For every $i \in \mathfrak{C}_{d,c,\vecx}$, if $a_{i}$ is not positive real, then we multiply $\vecx$ by $\matA$, and otherwise by $\matB$.
Since we have $\gamma_{c,i} (\matA \vecx) = a_{i} \cdot \gamma_{c,i}(\vecx)$ by \cref{GammaAB}, this ``shifts'' the phase of at least $\gamma_{c,i}(\vecx)$ on the left-hand side, but not for the addends on the right-hand side of \cref{eq:boosting_idea_motivation}.
By performing such multiplications repeatedly and taking a linear combination of the obtained inputs, \cref{eq:boosting_idea_motivation}'s left-hand side becomes arbitrarily small since addends ``cancel out'', whereas this is not the case for the right-hand side.
So one obtains a non-terminating input where \cref{eq:boosting_idea_motivation} does not hold for any $c \in [m]$.
Thus, $\vecx \in W$.

\begin{restatable}[Boosting]{lemma}{boosting}
  \label{lem:boosting}
  Let $\program$ be a non-terminating loop over $\semiring$, i.e., $\matA,\matB \in \semiring^{n \times n}$ and $\NT\cap\semiring^{n} \neq \emptyset$.
  Then there is a corresponding witness in $W\cap\semiring^{n}$.
\end{restatable}
\report{
  \begin{proof}[Proof Sketch]
    The full version of this proof can again be found in \cref{app:proofs_for_towards_witnesses_for_non-termination}.
    Let $\vecx \in \semiring^{n}$ be a non-terminating input, i.e., $\vecx \in \NT$.
    Then, according to \Cref{lem:dual_positive_eigenvalues_for_eventually_dominating_constraints}, there is a $d\in\{\nmax,\pmax\}$ such that for all $c\in[m]$ we have $\sum_{i\in\mathfrak{R}_{d,c,\vecx}} \gamma_{c,i}(\vecx) > 0$.
    In particular, $\constrainttermgroup_{d,c,\vecx} \supseteq \mathfrak{R}_{d,c,\vecx}\neq \emptyset$.
    For every $c \in [m]$, we define the sets $\{I_{c,1}, \ldots, I_{c,l_{c}}\} = \{i \in \mathfrak{C}_{d,c,\vecx} \mid \gamma_{c_h,i} (\vecx) \neq 0\}$.
    We prove the lemma by inductively constructing inputs $\vecx_{\!c,j} \in \semiring^{n}$ for all $(c,j) \in [m]\times ([l_c] \cup \{0\})$ such that
    \begin{equation}
      \begin{split}
        \text{for all $(c',j') \leq_{\mathrm{lex}} (c,j):$ If $j' \neq 0$, then }
        \frac{j'}{l_{c'}} \sum_{i\in\mathfrak{R}_{d,c',\vecx_{\!c,j}}} \!\!\!\!\!\!\!\! \gamma_{c',i} (\vecx_{\!c,j})
        > \sum_{\tilde{j}\in [j']} \! \left| \gamma_{c',I_{c',\tilde{j}}} \left(\vecx_{\!c,j}\right) \right|
      \end{split}
      \label{eq:boosting_main_part_claim}
    \end{equation}
    where $\leq_{\mathrm{lex}} {}={} \leq_{\mathrm{lex},\pmax}$ is the usual non-strict lexicographic order.
    Hence, for the final input $\vecx_{\!m,l_{m}}$ we have
    \[
      \sum_{i\in\mathfrak{R}_{d,c',\vecx_{\!m,l_m}}} \!\!\!\!\!\!\!\! \gamma_{c',i} (\vecx_{\!m,l_{m}})
      > \sum_{\tilde{j}\in [l_{c'}]} \! \left| \gamma_{c',I_{c',\tilde{j}}} \left(\vecx_{\!m,l_{m}}\right) \right|
    \]
    for all $c' \in [m]$, where the right-hand side is equal to $\sum_{i \in \mathfrak{C}_{d,c',\vecx_{\!m,l_m}}} |\gamma_{c',i}\left(\vecx_{\!m,l_{m}}\right)|$.
    Thus, $\vecx_{\!m,l_{m}} \in W\cap\semiring^{n}$.

    We set $\vecx_{1,0} = \vecx \in \semiring^{n}$ and $\vecx_{\!c,0} = \vecx_{\!c-1, l_{c-1}} \in \semiring^{n}$ for $c > 1$.
    Next consider $c \in [m]$ and $j \in \{1,\ldots,l_{c}\}$ such that \cref{eq:boosting_main_part_claim} holds for $\vecx_{\!c,j-1}$ (by the induction hypothesis).
    W.l.o.g.\ we may assume $a_{I_{c,j}} \not \in \AA_{>0}$.
    Otherwise we have $b_{I_{c,j}} \not \in \AA_{>0}$ as $I_{c,j} \in \constrainttermgroup_{d,c,\vecx}$ and one can replace $\matA$ with $\matB$ below.
    As shown in the full proof, there are constants $e_{1} \ldots e_{k} \in \NN$ and $p_{1},\ldots,p_{k},q \in \NN_{>0} \subseteq \semiring$ depending on the choice of $c$ and $j$ such that \cref{eq:boosting_main_part_claim} holds for
    \[ \vecx_{\!c,j} = (p_{1} \matA^{e_{1}} + \cdots + p_{k} \matA^{e_{k}})\cdot \vecx_{\!c,j-1} + q \cdot \vecx_{\!c,j-1} \,\in \semiring^{n}. \]
    By \cref{GammaAB} and linearity of all $\gamma_{c',i'}$ we have
    \[\gamma_{c',i'}(\vecx_{\!c,j}) = \left(\sum_{l = 1}^{k} p_{l} a_{i'}^{e_{l}} + q\right) \cdot \gamma_{c',i'}(\vecx_{\!c,j-1}) \]
    for all $(c',i') \in [m] \times [n]$.
    The intuition behind the choice of these constants is as follows: While for all $i' \in [n]$ with $a_{i'} \not \in \AA_{>0}$ the complex-valued addends in the coefficient $(\sum_{l = 1}^{k} p_{l} a_{i'}^{e_{l}} + q)$ of $\gamma_{c',i'}(\vecx_{\!c,j-1})$ in \cref{eq:boosting_main_part_claim} cancel each other out, such that the choice of the constants can make the sum's value arbitrarily close to $0$, this is not the case for the $i' \in [n]$ with $a_{i'} \in \AA_{>0}$, as those addends are real-valued and therefore cannot cancel each other out.
    As the $\gamma_{c',i'}$ with $a_{i'} \in \AA_{>0}$ make up the left-hand side in \cref{eq:boosting_main_part_claim} whereas those with complex-valued $a_{i'} \not \in \AA_{>0}$ constitute the sum on the right-hand side, \cref{eq:boosting_main_part_claim} can be re-established for $\vecx_{\!c,j}$, as shown in the proof's full version.
  \end{proof}

  \medskip
}

The following corollary summarizes our results so far, i.e., it shows that non-termination is equivalent to the existence of an element in $W$.

\begin{restatable}[Characterizing Termination]{corollary}{charactTermination}
  \label{CharacterizingTermination}
  A loop is terminating over a semiring $\semiring$ iff $W \cap \semiring^n =\emptyset$.
\end{restatable}
\report{
  \begin{proof}
    We have to show the equivalence of the statements $\NT \cap \semiring^{n} = \emptyset$ and $W \cap \semiring^{n} = \emptyset$.
    For the ``if'' direction, \Cref{lem:soundness_of_witnesses} states $W\cap\semiring^{n} \subseteq \ENT\cap\semiring^{n}$.
    Following \Cref{lem:nt_vs_ent}, $\emptyset \neq W\cap\semiring^{n} \subseteq \ENT\cap\semiring^{n}$ then implies $\NT\cap\semiring^{n} \neq \emptyset$.

    For the ``only if'' direction, let $\emptyset \neq \NT\cap\semiring^{n}$.
    Then, $W\cap\semiring^{n} \neq \emptyset$ according to \Cref{lem:boosting}.
  \end{proof}
}

%% file: illustration_full_middle.tex
{
\begin{tikzpicture}[scale=1.1,every node/.style={font=\tiny}]
  \draw [draw=none, fill=SeaGreen!5] (2,0.75) -- (5,0.75) -- (5,2) -- (4,2) -- (2,1) -- (2,0.75);
  \draw [draw=none, line space=5pt,pattern=my north east lines, pattern color=SeaGreen!30] (2,0.75) -- (5,0.75) -- (5,2) -- (4,2) -- (2,1) -- (2,0.75);

  \draw [draw=none, fill=SeaGreen!5] (2,-0.75) -- (5,-0.75) -- (5,-2) -- (4,-2) -- (2,-1) -- (2,-0.75);
  \draw [draw=none,line space=5pt,pattern=my north east lines, pattern color=SeaGreen!30] (2,-0.75) -- (5,-0.75) -- (5,-2) -- (4,-2) -- (2,-1) -- (2,-0.75);

  \draw [->,line width=0.6pt] (-1,0)--(5,0) node[below]{\footnotesize $|f|$};
  \draw [->,line width=0.6pt] (0,-2)--(0,2) node[left]{\footnotesize $\URVar(f) \cdot |f|$};

  \draw[domain=0:3, smooth, variable=\x, gray] plot ({\x}, {0.6875*\x*\x*\x - 3.313*\x*\x + 4.125*\x});
  \node[gray] at (0.8, 1.65) (A) {$\URVar(f)\cdot|f|$};
  \node[gray] at (3,1.125)[circle,fill,inner sep=0.75pt]{};

  \draw [red] (0,0) -- (4,2) node[below]{$\eps |f|$};
  \draw [red] (0,0) -- (4,-2) node[above]{$-\eps |f|$};

  \draw [blue] (-0.75, 0.75) node[below,xshift=0.25cm]{$ r$} -- (5, 0.75);
  \draw [blue] (-0.75,-0.75) node[above,xshift=0.25cm]{$-r$} -- (5,-0.75);

  \draw [Green] (2,-2) node[below]{$l$} -- (2,2);

  \draw [Plum] (2,0.875) -- (5,1.625) node[right] {$\frac{\eps \cdot |f| + r}{2}$};
\end{tikzpicture}
}

%% file: ent_witnesses_definition_T.tex
\begin{equation*}
  T = \min \{i\in\NN_{>0} \mid \URVar(\hat{f}\run_{1}\ldots\run_{i}) \cdot (|\hat{f}| + i) < r \text{ or } \URVar(\hat{f}\run_{1}\ldots\run_{i}) \not \in [0,\eps]\}.
\end{equation*}

%% file: deciding_past.tex
\section{Deciding PAST}
\label{sect:Deciding PAST}
Finally, we present our novel technique for deciding whether a loop is (positively almost surely) terminating, i.e., whether its expected runtime is finite for every input.
As discussed in \Cref{sec:towards_witnesses_for_non-termination}, to this end we only have to show decidability of $W \neq \emptyset$ for the set of witnesses $W$ for eventual non-termination from \Cref{def:witnesses_for_ent}.
We now explain how to translate this emptiness problem into an SMT problem.
More precisely, we show that the witness set $W$ is semialgebraic, i.e., it corresponds to a formula over polynomial arithmetic (which is linear in the variables $\vecx$).
For this we have to take into account that for different values of $\vecx$, different addends may be eventually dominating.
Then, decidability over the algebraic reals is clear.

As before, $\mathfrak{R}$ are those indices from $[n]$ where both eigenvalues $a_i$ and $b_i$ are positive reals, and $\mathfrak{C}$ are the remaining indices.

\begin{restatable}[Semialgebraic Sets of Witnesses for Algebraic Loops]{lemma}{semialgebraicwitnesssetforalgebraicloops}
  \label{lem:witness_set_is_semialgebraic}
  Let $\matA,\matB \in \semiring^{n \times n}$, $\matC \in \AA^{m \times n}$.
  We define the sets $\CEqZero_{c,(\scaleinside,\scaleoutside)}, \CGTZero_{c,(\scaleinside,\scaleoutside)} \subseteq \AA^{n}$ as\report{\footnote{In $H_{i,(\scaleinside,\scaleoutside)}$, due to \Cref{lem:sums_of_conjugated_addends_are_real_valued}
      it would be possible to replace the condition $\sum_{\substack{j \in \constrainttermgroup_{(\scaleinside,\scaleoutside)} \\
            \zeta_{j,\matA} = \zeta_{i,\matA} \\
            \zeta_{j,\matB} = \zeta_{i,\matB}}} \gamma_{c,j}(\vecx) + \conj{\sum_{\substack{j \in \constrainttermgroup_{(\scaleinside,\scaleoutside)} \\
              \zeta_{j,\matA} = \conj{\zeta_{i,\matA}} \\
              \zeta_{j,\matB} = \conj{\zeta_{i,\matB}}}} \gamma_{c,j}(\vecx)} = 0$ by ``$\Re(\sum_{\substack{j \in \constrainttermgroup_{(\scaleinside,\scaleoutside)} \\
            \zeta_{j,\matA} = \zeta_{i,\matA} \\
            \zeta_{j,\matB} = \zeta_{i,\matB}}} \gamma_{c,j}(\vecx)) = 0$''.}}
  \begin{align*}
    \CEqZero_{c,(\scaleinside,\scaleoutside)} & = \textstyle
    \left\{\vecx \in \AA^{n} \mid \sum_{i \in
      \constrainttermgroup_{(\scaleinside,\scaleoutside)}\cap\mathfrak{R}} \; \gamma_{c,i}(\vecx) = 0 \right\} \cap \bigcap_{i \in \constrainttermgroup_{(\scaleinside,\scaleoutside)}\cap\mathfrak{C}} H_{i,(\scaleinside,\scaleoutside)}
    \\
    H_{i,(\scaleinside,\scaleoutside)}        & =
    \left\{\vecx \in
    \AA^{n}\mid \sum\limits_{\substack{j \in
    \constrainttermgroup_{(\scaleinside,\scaleoutside)} \\
    \zeta_{j,\matA} = \zeta_{i,\matA} \\
        \zeta_{j,\matB} = \zeta_{i,\matB}}} \gamma_{c,j}(\vecx) \; + \;
    \conj{\sum\limits_{\substack{j \in \constrainttermgroup_{(\scaleinside,\scaleoutside)} \\
    \zeta_{j,\matA} = \conj{\zeta_{i,\matA}} \\
          \zeta_{j,\matB} = \conj{\zeta_{i,\matB}}}} \gamma_{c,j}(\vecx)} = 0\right\}
    \\
    \CGTZero_{c,(\scaleinside,\scaleoutside)} & = \textstyle
    \left\{ \vecx \in \AA^{n}\mid
    \sum_{i\in\constrainttermgroup_{(\scaleinside,\scaleoutside)}\cap\mathfrak{R}} \;
    \gamma_{c,i} (\vecx) >
    \sum_{i\in\constrainttermgroup_{(\scaleinside,\scaleoutside)}\cap\mathfrak{C}} \; \left|\gamma_{c,i} (\vecx)\right|\right\}
  \end{align*}
  for all $(c,(\scaleinside,\scaleoutside)) \in [m]\times\mathcal{I}$.
  Then, for all $d\in\{\nmax,\pmax\}$ and $\mathfrak{c} = ((\scaleinside_{1},\scaleoutside_{1}),\ldots,(\scaleinside_{m},\scaleoutside_{m})) \in \mathcal{I}^{m}$ we define $W_{d,\mathfrak{c}} \subseteq W_{d}$ as
  \[ \textstyle
    W_{d,\mathfrak{c}} = \left\{ \vecx \in W_{d} \Bigm\vert \bigwedge_{c \in [m]}
    \constrainttermgroup_{d,c,\vecx} = \constrainttermgroup_{(\scaleinside_{c},\scaleoutside_{c}),c,\vecx} \right\} .\]
  Then
  \begin{equation}
    W_{d,\mathfrak{c}} = \bigcap_{c\in[m]}
    \left(\CGTZero_{c,(\scaleinside_{c},\scaleoutside_{c})} \cap
    \bigcap_{\substack{(\scaleinside',\scaleoutside') \in \mathcal{I}\\
        (\scaleinside',\scaleoutside') >_{\mathrm{lex},d} (\scaleinside_{c},\scaleoutside_{c})}} \CEqZero_{c,(\scaleinside',\scaleoutside')} \right) \label{eq:witness_set_is_semialgebraic1} .
  \end{equation}

  Furthermore, we have $W = W_{\nmax} \cup W_{\pmax}= \biguplus_{\mathfrak{c} \in \mathcal{I}^{m}} W_{\nmax,\mathfrak{c}} \,\cup\, \biguplus_{\mathfrak{c} \in \mathcal{I}^{m}} W_{\pmax,\mathfrak{c}}$.
  The sets $W_{d,\mathfrak{c}}$ and the set $W$ are moreover semialgebraic.
\end{restatable}

\begin{example}
  We continue \Cref{ex:witnesses} and consider $\vecx =
    \begin{pmatrix}
      1 & 1 & 0
    \end{pmatrix}
    ^{\mathrm{T}}\in \NN^{3}$. By \Cref{ex:witnesses}
  we have $\vecx \in \CGTZero_{1,(10\sqrt{2},2)}$.
  Moreover, there is no $(\scaleinside',\scaleoutside') \in \mathcal{I} = \{(10\sqrt{2},2), (10\sqrt{2},\nicefrac{1}{2})\}$ with $(\scaleinside',\scaleoutside') >_{\mathrm{lex},\pmax} (10\sqrt{2},2)$.
  For $\mathfrak{c} = (10\sqrt{2},2)$ this implies $\vecx \in W_{\pmax,\mathfrak{c}} \subseteq W_{\pmax} \subseteq W \subseteq \ENT \neq \emptyset \Rightarrow \NT \neq \emptyset$.

  Thus, the loop initially introduced in \Cref{ex:initial_program} is non-terminating for all $\semiring \subseteq \{\ZZ, \QQ, \AA\}$.\footnote{We did not consider $\semiring \in \{\NN,\QQ_{\geq 0}, \AA_{\geq 0}\}$ as for such a choice of $\semiring$ we do not have $\matA,\matB \in \semiring^{n \times n}$.}
\end{example}

Note that while for $(\scaleinside,\scaleoutside) \in \mathcal{I}$ the number $\scaleinside$ is not necessarily algebraic, the representation of $W_{d} \cap \AA^{n}$ as a finite union/intersection of the semialgebraic sets $\CEqZero_{(\scaleinside,\scaleoutside),c}, \CGTZero_{(\scaleinside,\scaleoutside),c}$ is still computable by \Cref{lem:comparison_of_constraint_term_groups} as one simply has to determine the corresponding ordering $>_{\mathrm{lex},d}$ on $\mathcal{I}$.
This is the reason why we restricted $p$ to the set of algebraic reals.

To show that emptiness of $W$ is also decidable over various sub-semirings $\semiring$ of the algebraic reals, we prove the convexity of the sets $W_{d,\mathfrak{c}}$.
Note that the set $W$ as well as the sets $W_{\nmax}, W_{\pmax}$ themselves are in general \emph{not} convex.

\begin{restatable}[$W_{d}$ as Finite Union of Convex Sets]{lemma}{wdasfiniteunionofconvexsets}
  \label{lem:witness_set_is_union_of_convex_sets}
  For $d \in \{\nmax,\pmax\}$ and $\mathfrak{c} \in \mathcal{I}^{m}$, the set $W_{d,\mathfrak{c}}$ is convex, i.e., $t \vecx + (1-t)\vecy \in W_{d,\mathfrak{c}}$ for all $\vecx,\vecy \in W_{d,\mathfrak{c}}$ and $t \in (0,1)$.
\end{restatable}

Note that \Cref{lem:witness_set_is_semialgebraic,lem:witness_set_is_union_of_convex_sets} imply that for $d\in \{\nmax,\pmax\}$ the set $W \cap \AA^{n}$ is semialgebraic and a finite union of convex sets.

\begin{theorem}[Deciding PAST]
  \label{thm:deciding_past}
  Let $\semiring \in \{\NN,\ZZ, \QQ_{\geq0}, \QQ, \AA_{\geq0}, \AA\}$.
  Then, the question whether a loop is terminating on $\semiring^{n}$ is decidable, and if the loop is non-terminating, then a witness $\vecx \in W \cap \semiring^{n}$ for eventual non-termination can be computed.
\end{theorem}
\report{
  \begin{proof}
    By \Cref{CharacterizingTermination}, emptiness of the set $W\cap\semiring^{n}$ is equivalent to termination over $\semiring$.
    So we only have to show that $W\cap\semiring^{n} = \emptyset$ is decidable.

    If $\semiring = \AA$, then one can use a decision procedure (e.g., cylindrical algebraic decomposition) for the first-order theory of the reals \cite{tarski51,collins75}
    to decide whether $W\cap\AA^{n} = \emptyset$.
    Note that for the computation of the corresponding formula according to \Cref{lem:witness_set_is_semialgebraic}, in general one has to compare transcendental numbers in order to compute $\constrainttermgroup_{d,c,\vecx}$ for $(d,c) \in \{\nmax,\pmax\}\times[m]$.
    However, in our case this is decidable as stated by \Cref{lem:comparison_of_constraint_term_groups}.
    In addition, one can compute a witness $\vecx \in W\cap\AA^{n}$ whenever $W\cap\AA^{n} \neq \emptyset$.

    Next, consider $\semiring = \ZZ$.
    Given a convex semialgebraic set $A\subseteq\AA^{k}$, the result~\cite[Thm.\ 1.2]{khachiyan97} by Khachiyan and Porklab gives an algorithm to decide whether $A\cap\ZZ^{n} \neq \emptyset$ and in addition, it allows the computation of an integral point $\vecx \in A\cap\ZZ^{n}$ whenever this is the case.
    The theorem then follows in this case as $W = \bigcup_{d\in\{\nmax,\pmax\}} \biguplus_{\mathfrak{c} \in \mathcal{I}^{m}}W_{d,\mathfrak{c}}$ and the finitely many sets $W_{d,\mathfrak{c}}$ are each convex and semialgebraic.
    Note that~\cite[Thm.\ 1.2]{khachiyan97} has been previously applied to the termination analysis of linear loops in~\cite{hosseini19,ouaknine15}.

    For $\semiring = \QQ$ one observes~\cite{braverman06} that $W\cap\ZZ \subseteq W\cap\QQ$ and $\vecx \in W\cap\QQ$ entails $q\vecx \in W\cap\ZZ$ where $q \in \NN_{>0}$ is the product of the denominators of all rational entries of $\vecx$.
    Here, $q \vecx \in W$ is due to the linearity of all $\gamma_{c,i}$.
    So the loop is terminating on $\QQ^n$ iff it is terminating on $\ZZ^n$, i.e., we can use the decision procedure for $\semiring = \ZZ$ again.

    Finally, consider $\semiring \in \{\NN,\QQ_{\geq0},\AA_{\geq0}\}$.
    Here, one uses the same procedures as outlined above but instead considers the set $W' = W'_{\nmax} \cup W'_{\pmax}$ with $W'_{d} = W_{d} \cap \{\vecx\in\AA^{n} \mid \vecx \geq \veczero\}$ for $d\in\{\nmax,\pmax\}$.
    Note that the sets $W',W'_{\nmax},W'_{\pmax}$ are again semialgebraic and $W_{d}$ as an intersection of convex sets is a convex set again.
  \end{proof}
}

\begin{remark}
  \label{rem:on_computations}
  The theory of the reals and the algebraic reals are elementary equivalent as both are real closed fields.
  Thus, \Cref{thm:deciding_past} directly entails that the question whether there exists a non-terminating non-negative real input $\vecx\in\RR_{\geq0}^{n}$ or real input $\vecx\in\RR^{n}$ for an algebraic loop $\program$ is decidable as well, if one extends the set $\NT$ and the corresponding definitions to real inputs $\vecx \in \RR^{n}$.
  Note that in this case $\NT \neq \emptyset$ iff $\NT \cap \AA^{n} \neq \emptyset$.
\end{remark}

While the procedure outlined in the proof of \Cref{thm:deciding_past} only allows for the computation of a witness $\vecx \in \ENT \cap \semiring^{n}$ for eventual non-termination, one can lift this to the computation of a witness $\vecy \in \NT \cap \semiring^{n}$ according to the constructive\paper{ \pagebreak}
proofs of \Cref{lem:domination_of_eventually_dominating_constraint_term_groups,lem:soundness_of_witnesses}.
\begin{restatable}[Computing Witnesses for Non-Termination]{corollary}{computingwitnessesfornontermination}
  \label{cor:computing_witnesses_for_nontermination}
  Let $\semiring \in \{\NN,\ZZ, \QQ_{\geq0},\QQ, \AA_{\geq0}, \AA\}$.
  If a loop is non-terminating, then a witness for non-termination from $\NT \cap \semiring^{n}$ can be computed.
\end{restatable}

%% file: conclusion.tex
\section{Implementation and Conclusion}
\label{Implementation and Conclusion}

\emph{Prototype Implementation:} To demonstrate the practical applicability of our decision procedure, we implemented it in our prototype tool \textsf{\toolname}
(for ``\underline{Si}mple \underline{R}andomized Lo\underline{op}s'').
The tool and a corresponding collection of exemplary programs can be obtained from
\[ \mbox{\url{https://github.com/aprove-developers/SiRop}} \]
The tool is implemented in \textsf{Python} and uses the \textsf{SageMath} open-source mathematics software system~\cite{sagemath} in order to perform necessary computations such as simultaneous diagonalization and determining the mappings $\gamma_{c,i}$.
\textsf{\toolname}
tries to compute a witness $\vecx \in W$ for eventual non-termination by creating a corresponding SMT problem which is then solved using the \textsf{SMT-RAT}~\cite{SMT-RAT}
solver.
If the SMT problem is unsatisfiable, then the program is terminating.
In contrast, if such a witness $\vecx$ is found, then the program is non-terminating and the tool computes a non-terminating input $\vecy \in \NT$ from $\vecx$.
Currently, \textsf{\toolname}
handles loops over the algebraic reals only, i.e., $\semiring = \AA$, as for all other considered sub-semirings of $\AA$, the decision procedure relies on the technique presented in~\cite{khachiyan97}
which (to the best of our knowledge) has not yet been implemented.

\medskip

\emph{Conclusion:}
We have shown the decidability of universal positive almost sure termination (\textsf{UPAST}) for the class of simple randomized loop ranging over numerous semirings $\semiring$, thereby transferring a line of research started in 2004 by Tiwari~\cite{tiwari04} on universal termination of linear loops to the realm of randomized programs.
To that end, we devised a corresponding decision procedure and presented a prototype implementation for the case $\semiring = \AA$, showing the practical applicability of the presented approach.
In particular, our tool managed to find a non-terminating algebraic input for one\footnote{\url{https://github.com/TermCOMP/TPDB/blob/11.3/C_Integer/Stroeder_15/ChenFlurMukhopadhyay-SAS2012-Ex2.06_false-termination.c}} of the only two problems from the category \texttt{C Integer} which were not solved by any tool at the 2023 Termination Competition \cite{TermComp19},\footnote{This category was not part of the 2024 competition.} the other one being the Collatz problem.
While our tool only considers $\semiring = \AA$ (whereas the problem is formulated over the integers), the constraints generated by \textsf{\toolname} are unsatisfiable over $\ZZ$, which implies universal termination of the program over the integers.

\medskip

\emph{Future Work}: While our procedure can decide positive almost sure termination for all inputs, in the future we want to improve it such that it can also compute bounds on expected runtimes.
Moreover, decision procedures for termination or complexity of subclasses of non-randomized programs (e.g., \cite{floriantriangularinitial,hark20,hark23,solvable-loops}) have been integrated in (incomplete) tools that analyze general programs \cite{koat-targeting-completeness,koat-twn}, and we would like to investigate such an integration for randomized programs as well.
Finally, we plan to adapt our approach to a decision procedure for universal almost sure termination (\textsf{UAST}), i.e., whether a program terminates with probability $1$ on all inputs.
Clearly, \textsf{UPAST} implies \textsf{UAST} but the converse does not hold in general.

\bigskip

\noindent
\textbf{Acknowledgements:} We thank Sophia Greiwe for her help with the implementation of our decision procedure in \textsf{\toolname}.

%% file: probability_theory_app.tex
\section{Preliminaries from Probability Theory}
\label{sect:Preliminaries from Probability Theory}
This section provides a brief introduction to concepts from probability theory used in this paper.
It is based on the book by Grimmett and Stirzaker \cite{grimmettprobability}.
Throughout this section we regard an abstract but fixed set $\Omega$ of events.

\begin{definition}[$\sigma$-Field \protect{\cite[Def.\ 1.2(5)]{grimmettprobability}}]
  A set $\F$ of subsets of $\Omega$ is called a \emph{$\sigma$-field} over $\Omega$ if it
  \begin{enumerate}
    \item contains the empty set: $\emptyset \in \F$
    \item is closed under (countable) union: $\bigcup_{i=0}^{\infty} A_i \in \F$ for all $A_{0},A_{1},\ldots \in \F$
    \item is closed under complement: $\Omega \setminus A \in \F$ for all $A \in \F$
  \end{enumerate}
\end{definition}

The sets in $\F$ are commonly referred to as \emph{measurable} sets.
It is well known that the intersection $\cap_{i \in I} \F_{i}$ of a family ${(\F_{i})}_{i \in I}$ of $\sigma$-fields is again a $\sigma$-field~\cite[1.6.(1)]{grimmettprobability}.
This fact motivates the following definition of \emph{generated $\sigma$-fields}.

\begin{definition}[$\sigma$-Field Generated by Sets]
  \label{def:sigma_field_generated_by_family_of_sets}
  For a set $\mathcal{A} \subseteq 2^{\Omega}$, we denote by
  \[\sigma(\mathcal{A}) = \bigcap_{\substack{\mathcal{A} \subseteq \F\\\text{$\F$ $\sigma$-field over $\Omega$}}} \F\]
  the $\sigma$-field generated by all $A \in \mathcal{A}$, i.e., the smallest $\sigma$-field containing $\mathcal{A}$.
\end{definition}

A tuple $(\Omega,\F)$ where $\F\subseteq2^\Omega$ is a $\sigma$-field is called a \emph{measurable space}.
Given such a measurable space, a \emph{probability measure} $\P$ then assigns a probability to every measurable set $A \in \F$.

\begin{definition}[Probability Measure \protect{\cite[Def.\ 1.3(1)]{grimmettprobability}}]
  \label{def:probability_measures}
  A function $\P: \F \to [0,1]$ is called a \emph{probability measure} (on $(\Omega,\F)$) if
  \begin{enumerate}
    \item $\P (\Omega) = 1$
    \item $\P \left(\biguplus_{i=0}^{\infty} A_{i}\right) = \sum_{i=0}^{\infty} \P (A_{i})$ for all pairwise disjoint sets $A_{0},A_{1},\ldots \in \F$
  \end{enumerate}
\end{definition}

From $\P (\Omega) = 1$ one directly concludes $\P(\emptyset) = 0$ as $\P (\emptyset) + \P (\Omega) = \P (\emptyset \uplus \Omega) = \P(\Omega)$.
By equipping a measurable space $(\Omega,\F)$ with a probability measure $\P$ on $(\Omega,\F)$, a \emph{probability space} $(\Omega,\F,\P)$ is obtained.

\begin{definition}[Probability Space]
  A tuple $(\Omega,\F,\P)$ consisting of a $\sigma$-field $\F$ and a probability measure $\P$ on $(\Omega,\F)$ is called a \emph{probability space}.
\end{definition}

In the remainder of this section, we consider some fixed probability space $(\Omega,\F,\P)$.
Often, one is not interested in specific outcomes of a random experiment but rather in the values taken by a function $X: \Omega \rightarrow \RR$, called a \emph{random variable}, on the observed outcome.
In this paper, we will restrict ourselves to discrete random variables for the sake of simplicity, i.e., those that take on only a discrete set of reals.
A well-known discrete random variable is the indicator function $\ind_{A}: \Omega \to \{0,1\}$ of a set $A \subseteq \Omega$.
It is defined such that $\ind_{A} (\omega) = 1$ iff $\omega \in A$.

\begin{definition}[Measurability, Discrete Random Variables, \& Stochastic Processes \protect{\cite[Def.\ 2.1(3) and 2.3(1)]{grimmettprobability}}]
  Let $X: \Omega \to \RR$ such that the image of $X$ is countable.
  If moreover, $X^{-1}(r) \in \F$ for all $r \in \RR$ and some $\sigma$-field $\F$, then $X$ is said to be ($\F$-)\emph{measurable}.
  The function $X$ is called a \emph{(discrete) random variable.}

  A (countable) collection ${(X_{t})}_{t=0}^{\infty}$ of discrete random variables is called a \emph{stochastic process}.
\end{definition}

As we deal with discrete random variables only, for the $\F$-measurability of $X$, it suffices to require $X^{-1}(r) \in \F$ for all $r \in \RR$, instead of the usual (stronger) requirement $\{ \omega \in \Omega \mid X(\omega) \leq r\}
  \in \F$ from~\cite[2.1(3)]{grimmettprobability}. In addition to \cref{def:sigma_field_generated_by_family_of_sets}, we define the $\sigma$-field generated by a family of (discrete) random variables below.
\begin{definition}[$\sigma$-Field Generated by Random Variables]
  Given a family ${(X_{i})}_{i \in I}$ of (discrete) random variables $X_{i} \colon \Omega \to \RR$, we denote by
  \[\sigma\left({(X_{i})}_{i \in I}\right) = \sigma\left(\{X_{i}^{-1} (r) \mid i \in I,\, r \in \RR\} \right)\]
  the $\sigma$-field generated by ${(X_{i})}_{i \in I}$, i.e., the smallest $\sigma$-field such that all random variables $X_{i}$, for $i \in I$, are measurable.
\end{definition}

Note that usually equations and inequations $X \mathrel{\bowtie} X'$ between two random variables $X,X'$ for ${\mathrel{\bowtie}} \in \{<, \leq,=, \geq,>\}$ only hold at most with probability $1$, i.e., $\P (\{ \omega \in \Omega \mid X(\omega) \mathrel{\bowtie} X'(\omega)\}) = 1$.
Then, one says that $X \bowtie X'$ holds \emph{almost surely}.
Similarly, we say that a statement $\varphi$ holds for \emph{almost all} $\omega \in \Omega$ iff the set $A = \{\omega \in \Omega \mid \neg\varphi(\omega)\}$ is a null-set, i.e., $\P(A) = 0$.

The sum of all values taken by a random variable weighed with their corresponding probabilities is known as its \emph{expectation}.
Note that a sum (instead of the more general Lebesgue integral) suffices here as we only regard discrete random variables.

\begin{definition}[Integrability \& Expectations (\protect{\cite[Def.\ 3.3(1)]{grimmettprobability}})]
  Let $X: \Omega \to \RR$ be a discrete random variable and consider the sum $S = \sum_{x \in \RR} x \cdot \P \left(X^{-1} (x)\right)$.
  $X$ is said to be \emph{integrable} whenever this sum converges absolutely.
  In this case, the \emph{expectation} of $X$ is defined as $\E(X) = S$.
\end{definition}

For non-integrable \emph{non-negative} random variables $X: \Omega \to \RR_{\geq 0}$ one defines $\E (X) = \lim_{n\to\infty} \E (\min(X,n))= \infty$.

Often instead of relying on the expectation of a random variable alone, one focuses on its expectation in view of some already ``revealed'' information.
To that end, one considers a sub $\sigma$-field $\mathcal{G} \subseteq \F$.
The \emph{conditional expectation} of $X : \Omega \to \RR$ given $\mathcal{G}$ is then a $\mathcal{G}$-measurable random variable $\E (X \mid \mathcal{G}) : \Omega \to \RR$.
Intuitively, the value of $\E (X \mid \mathcal{G}) (\omega)$ is the ``closest approximation'' to the value $X(\omega)$ when limited to the information contained in the $\sigma$-field $\mathcal{G}$.
In particular, $\E (X \mid \F) = X$ as $X$ is by definition $\F$-measurable, so the ``available information'' is not limited.
\begin{definition}[Conditional Expectation of a Random Variable \protect{\cite[Thm.\ 7.9(26)]{grimmettprobability}}]
  \label{def:conditional_expectations}
  Let $X: \Omega \to \RR$ be an integrable random variable and $\mathcal{G} \subseteq \F$ be a $\sigma$-field.
  Then, the \emph{conditional expectation} of $X$ given $\mathcal{G}$ is the unique random variable $\E (X \mid \mathcal{G}): \Omega \to \RR$ such that
  \begin{enumerate}
    \item[(1)] $\E (X \mid \mathcal{G})$ is $\mathcal{G}$-measurable\label{it:conditional_expectations_item1}
    \item[(2)] $\E (|\E (X \mid \mathcal{G})|) < \infty$ and $\E (\E(X \mid \mathcal{G}) \cdot \ind_{G}) = \E (X \cdot \ind_{G})$ for all $G \in \mathcal{G}$\label{it:conditional_expectations_item2}
  \end{enumerate}
\end{definition}
Note that the conditional expectation of $X$ given $\mathcal{G}$ is only unique in the sense that for all such conditional expectations $Y,Y'$ we have $\P (\{ \omega \in \Omega \mid Y(\omega) = Y'(\omega) \}) = 1$.
Moreover, as for ordinary expectations, we will also consider conditional expectations $\E (X \mid \mathcal{G}): \Omega \to \RR_{\geq 0} \cup \{\infty\}$ of non-integrable \emph{non-negative} random variables $X: \Omega \to \RR$ by setting $\E (X \mid \mathcal{G}) = \lim_{n\to\infty} \E (\min(X,n) \mid \mathcal{G})$.
Of course, the first part of \Cref{def:conditional_expectations} does not apply anymore in this case.

\Cref{def:conditional_expectations} is not constructive.
However, often one considers the conditional expectation $\E (X \mid \mathcal{G})(\omega)$ for some $\omega \in \Omega$, random variable $X\colon \Omega \to \RR$, $\sigma$-field $\mathcal{G}$, and $\omega \in A \in \mathcal{G}$ such that no non-empty proper subset of $A$ is contained in $\mathcal{G}$.
In this case, we call $A$ an \emph{atom} of $\mathcal{G}$, inspired by the literature on measure theory, e.g.,~\cite[Thm.\ 5.5.8]{ashprobabilityandmeasuretheory}.
Furthermore, the value of the conditional expectation $\E (X \mid \mathcal{G})(\omega)$ is straightforward for all such $\omega \in A$, as stated by the following lemma.
\begin{lemma}[Atoms and Conditional Expectations]
  \label{lem:atoms_and_conditional_expectations}
  Let $\mathcal{G}$ be a $\sigma$-field and $A \in \mathcal{G}$, such that for all $B \in \mathcal{G}$ we have either $A \subseteq B$ or $A \cap B = \emptyset$.
  Then, we call $A$ an \emph{atom} of $\mathcal{G}$ and we have $\E (X \mid \mathcal{G})(\omega) = \E (X \mid \mathcal{G})(\omega')$ for all $\omega,\omega' \in A$.
  In particular, $\P(A) \cdot \E(X \mid \mathcal{G})(\omega) = \E(\ind_{A} \cdot X)$ for all $\omega \in A$.

  Moreover, if $\mathcal{G} =\sigma\left(\{A_{i} \mid i \in I\}\right)$ where $A_{i}
    \subseteq \Omega$ for all $i \in I$, and there is some $i \in I$ such that there is no $A_{j}$ with $j \in I$ and $\emptyset \subsetneq A_{j} \subsetneq A_{i}$, then $A_{i}$ is an atom of $\mathcal{G}$.
\end{lemma}
\begin{proof}
  We first prove the first statement in the first paragraph of \cref{lem:atoms_and_conditional_expectations}.
  To that end, let $\omega, \omega' \in A$ (the statement is trivial for $A = \emptyset$) and consider $r_{\omega} = \E(X \mid \mathcal{G}) (\omega)$ as well as $r_{\omega'} = \E(X \mid \mathcal{G}) (\omega')$.
  Consider the sets $A_{\omega} = A \cap {\left(\E(X \mid \mathcal{G})\right)}^{-1} (r_{\omega})$ and $A_{\omega'} = A \cap {\left(\E(X \mid \mathcal{G})\right)}^{-1} (r_{\omega'})$.
  By \cref{def:conditional_expectations}(1) the conditional expectation $\E(X \mid \mathcal{G})$ is $\mathcal{G}$-measurable.
  Thus, ${\left(\E(X \mid \mathcal{G})\right)}^{-1} (r) \in \mathcal{G}$ for $r \in \{r_{\omega},r_{\omega'}\}$.
  This implies $A_{\omega}, A_{\omega'} \in \mathcal{G}$.
  As $\emptyset \subsetneq A_{\omega},A_{\omega'} \subseteq A$, we must have $A_{\omega} = A_{\omega'} = A$ since $A$ is an atom of $\mathcal{G}$ and hence $\E(X \mid G)(\omega) = \E(X \mid \mathcal{G}) (\omega')$.

  To prove the second statement of the first paragraph, we consider $\omega \in A$ and
  \begin{align*}
    \E(\ind_{A} \cdot X)
     & = \E ( \ind_{A} \cdot \E(X \mid \mathcal{G}) ) \tag{by \cref{def:conditional_expectations}(2)} \\
     & = \E ( \ind_{A} \cdot \E(X \mid \mathcal{G})(\omega) ) ,
  \end{align*}
  where in the last line we used the already proven statement of \cref{lem:atoms_and_conditional_expectations}.
  We continue our computation and obtain
  \begin{equation*}
    \E(\ind_{A} \cdot X) = \E ( \ind_{A} \cdot \E(X \mid \mathcal{G})(\omega) ) = \E ( \ind_{A}) \cdot \E(X \mid \mathcal{G})(\omega) = \P (A) \cdot \E(X \mid \mathcal{G})(\omega)
  \end{equation*}
  which ends the proof of the first paragraph of \cref{lem:atoms_and_conditional_expectations}.

  To see why the last paragraph of \cref{lem:atoms_and_conditional_expectations} holds, we note that by \cref{def:sigma_field_generated_by_family_of_sets}, $\mathcal{G} = \sigma\left(\{A_{i} \mid i \in I\}\right)$ is the smallest set with $\emptyset \in \mathcal{G}$ and $A_{i} \in \mathcal{G}$ for all $i \in I$ that is closed under countable union and complement.
  As there is no $j \in I$ such that $\emptyset \subsetneq A_{j} \subsetneq A_{i}$, there is no proper non-empty subset of $A_i$ in $\mathcal{G}$, i.e., $A_{i}$ is an atom of $\mathcal{G}$.
\end{proof}

\medskip

A stochastic process ${(X_{t})}_{t=0}^{\infty}$ can be thought of as ``a value evolving over time''.
This can be mirrored by a sequence $\F_{0} \subseteq \F_{1} \subseteq \ldots \subseteq \F$ of increasingly finer $\sigma$-fields where $\F_{t}$ represents the information revealed up to time step $t \in \NN$.
\begin{definition}[Filtration \& Adaptedness \protect{\cite[p.\ 473]{grimmettprobability}}]
  \label{def:filtration_and_adaptedness}
  A sequence $\F_{0} \subseteq \F_{1} \subseteq \ldots \subseteq \F$ of $\sigma$-fields is called a \emph{filtration} (of $\F$).
  A stochastic process ${(X_{t})}_{t=0}^{\infty}$ is said to be \emph{adapted} to $\F_{0} \subseteq \ldots$ whenever $X_{t}$ is $\F_{t}$-measurable for all $t \in \NN$.
\end{definition}

Now consider a stochastic process ${(X_{t})}_{t=0}^{\infty}$ adapted to a corresponding filtration $\F_{0} \subseteq \ldots$.
The values of the random variable $X_{t}$ correspond to the values of the process at time $t \in \NN$, whereas the conditional expectation $\E (X_{t+1} \mid \F_{t})$ intuitively corresponds to the ``best forecast'' for the next value of the process, i.e., $X_{t+1}$, given the information revealed up to (and including) time $t$, i.e., $\F_{t}$.
If for all $t \in \NN$ one expects that the value of the process does not change within the next step, i.e., $\E (X_{t+1} \mid \F_{t}) = X_{t}$, then ${(X_{t})}_{t=0}^{\infty}$ constitutes a so-called \emph{martingale} (w.r.t.\ the filtration $\F_{0} \subseteq \ldots$).

\begin{definition}[Martingale \protect{\cite[Def.\ 12.1(8)]{grimmettprobability}}]
  \label{def:martingale}
  Let ${(X_{t})}_{t=0}^{\infty}$ be a stochastic process adapted to a filtration $\F_{0} \subseteq \F_{1} \subseteq \ldots$.
  Then, ${(X_{t})}_{t=0}^{\infty}$ is called a \emph{martingale} (w.r.t.\ filtration $\F_{0} \subseteq \ldots$) if for all $t \in \NN$
  \begin{enumerate}
    \item $\E (|X_{t}|) < \infty$
    \item $\E (X_{t+1} \mid \F_{t}) = X_{t}$
  \end{enumerate}
\end{definition}

The following inequation, also known under the names ``Azuma's Inequation'' and ``Azuma-Hoeffding Inequation'' provides a concentration bound for martingales, i.e., a bound on the probability that a martingale ${(X_{t})}_{t=0}^{\infty}$ has moved away $\eps$ steps from its initial value $X_{0}$ after $t$ time steps.
\begin{theorem}[(Azuma-)Hoeffding's Inequation \protect{\cite[Thm.\ 12.2(3)]{grimmettprobability}}]
  \label{thm:azumahoeffdinginequality}
  Let ${(X_{t})}_{t=0}^{\infty}$ be a martingale w.r.t.\ a filtration $\F_{0} \subseteq \F_{1} \subseteq \ldots$.
  Let $c_{0}, c_{1}, \ldots \in \RR_{\geq 0}$ such that $|X_{t+1} - X_{t}| \leq c_{t}$ for all $t \in \NN$.
  Then,
  \[ \P (X_{t} - X_{0} \geq \eps) \leq \exp\left(\frac{-\eps^{2}}{2 \sum_{i=0}^{t-1} c_{i}^{2}}\right)\]
  for all $\eps > 0$.
\end{theorem}

For a martingale ${(X_{t})}_{i=0}^{\infty}$ w.r.t.\ $\F_{0} \subseteq \F_{1} \subseteq \ldots$ one deduces $\E (X_{t}) = \E (X_{t} \cdot \ind_{\Omega}) = \E (\E(X_{t} \mid \F_{t-1}) \cdot \ind_{\Omega}) = \E (\E(X_{t} \mid \F_{t-1})) = \E(X_{t-1}) = \ldots = \E(X_{0})$ for all $t \in \NN$.
But what if the time $T$ at which the value of the process is inspected is itself a random variable $T: \Omega \to \NN$?
To answer this question, it is necessary that at every time step $t \in \NN$ it must be known whether we want to inspect the value of $X_{t}$, i.e., $\{\omega \in \Omega \mid T(\omega) = t\} \in \F_{t}$ for all $t \in \NN$.
Such a random variable $T$ is called a \emph{stopping time} (w.r.t.\ $\F_{0} \subseteq \ldots$).
The following special case of a theorem by Doob sets out conditions under which $\E(X_{T})$ not only exists but is equal to $\E(X_{0})$.

\begin{theorem}[Optional Stopping Theorem (\protect{\cite[Thm.\ 12.5(1) and 12.5(9)]{grimmettprobability}})]
  \label{thm:optionalstopping}
  Let ${(X_{t})}_{t=0}^{\infty}$ be a martingale w.r.t.\ $\F_{0} \subseteq \F_{1} \subseteq \ldots$ and $T$ a corresponding stopping time such that
  \begin{enumerate}
    \item $\E (T) < \infty$
    \item $\E (| X_{t+1} - X_{t}| \mid \F_{t}) \leq c$ for all $t < T$ and some constant $c \in \RR$.
  \end{enumerate}
  Then, $\E (X_{T}) = \E (X_{0})$.
\end{theorem}

%% file: linear_algebra_app.tex
\section{Preliminaries from Linear Algebra}
\label{sec:preliminaries_from_linear_algebra}
\begin{lemma}[Complex Eigenvalues of Real Matrices Appear in Conjugate Pairs]
  \label{lem:complex_eigenvalues_of_real_matrices_appear_in_conjugate_pairs}
  Let $\matA \in \RR^{n \times n}$ and let $\vec{v} =
    \begin{pmatrix}
      v_{1} & \cdots & v_{n}
    \end{pmatrix}
    \in \CC^{n}$ be an eigenvector of $\matA$ for some eigenvalue $\lambda$. Then, $\conj{\vec{v}} =
    \begin{pmatrix}
      \conj{v_{1}} & \cdots & \conj{v_{n}}
    \end{pmatrix}
  $ is an eigenvector of $\matA$ for the eigenvalue $\conj{\lambda}$.
\end{lemma}
\begin{proof}
  As $\matA \vec{v} = \lambda \vec{v}$ and $\matA \in \RR^{n \times n}$ we have $\matA \conj{\vec{v}} = \conj{\matA \vec{v}} = \conj{\lambda \vec{v}} = \conj{\lambda}
    \, \conj{\vec{v}}$.
\end{proof}

\begin{lemma}[Eigenvectors for Real Eigenvalues]
  \label{lem:eigenvectors_for_real_eigenvalues}
  Let $\matA \in \RR^{n \times n}$ and let $\lambda \in \RR$ be a real eigenvalue of $\matA$ with geometric multiplicity $m$.
  Then, there are $m$ corresponding real eigenvectors $\vec{v_{1}},\ldots,\vec{v_{m}}$ for the eigenvalue $\lambda$, i.e., $\vec{v_{i}} \in \ker(A - \lambda I)$ for all $i\in[m]$.
\end{lemma}
\begin{proof}
  A basis for the kernel of the matrix $A - \lambda I$ can for example be computed by considering its reduced row echelon form.
  As $A - \lambda I \in \RR^{n \times n}$, the computation can entirely be performed over the real numbers giving rise to real eigenvectors $v_{1}, \ldots, v_{m} \in \RR^{n}$.
\end{proof}

\begin{lemma}[Eigenvectors for Complex Eigenvalues]
  \label{lem:eigenvectors_for_complex_eigenvalues}
  Let $\matA \in \RR^{n \times n}$, $\lambda \in \CC$ be an eigenvalue of $\matA$, and $(\vec{v_{1}},\ldots,\vec{v_{m}})$ be a basis consisting of eigenvectors for the eigenvalue $\lambda$, i.e., a basis of $\ker (A - \lambda I)$.
  Then, $(\conj{\vec{v_{1}}}, \ldots, \conj{\vec{v_{m}}})$ is an eigenbasis for the conjugated eigenvalue $\conj{\lambda}$.
\end{lemma}
\begin{proof}
  The linear independence of $(\vec{v_{1}},\ldots,\vec{v_{m}})$ directly implies the linear independence of $(\conj{\vec{v_{1}}}, \ldots, \conj{\vec{v_{m}}})$.
  By \Cref{lem:complex_eigenvalues_of_real_matrices_appear_in_conjugate_pairs}, all $\conj{\vec{v_{i}}}$ are eigenvectors for $\conj{\lambda}$.
\end{proof}

\medskip

Throughout this paper we will always assume that the eigenbasis used in the diagonalization of a matrix $\mat{A}$ corresponds to \Cref{lem:eigenvectors_for_real_eigenvalues,lem:eigenvectors_for_complex_eigenvalues}.
Here, $\matS_{\_{},i}$ denotes the $i$-th column vector of the matrix $\matS$.

\begin{lemma}[Conjugate Coefficients in Eigenbasis Representation]
  \label{lem:conjugate_coefficients_in_eigenbasis_representation}
  Let $\matA \in \RR^{n \times n}$ and $\mat{D} \in \CC^{n \times n}$ be the diagonal form of $\matA$, i.e., $\matA = \matS \mat{D} \matS^{-1}$ for some regular $\matS \in \CC^{n \times n}$.
  Then, for all $\vecx\in\RR^{n}$ and all $i\in[n]$ there is some unique $j\in[n]$ such that $\matS_{\_{},i} = \conj{\matS_{\_{},j}}$ and ${(\matS^{-1}\vecx)}_{i} = \conj{{(\matS^{-1}\vecx)}_{j}}$.
\end{lemma}
\begin{proof}
  Following \Cref{lem:eigenvectors_for_real_eigenvalues,lem:eigenvectors_for_complex_eigenvalues}, we partition the set $\{\matS_{\_{},1},\ldots,\matS_{\_{},n}\}$ of $\matS$'s column vectors into three sets: $\{\matS_{\_{},i} \mid \matS_{\_{},i} \in \RR^{n} \} = \{\vec{r_{1}},\ldots,\vec{r_{k}}\}$, $\{\vec{v_{1}}, \ldots, \vec{v_{m}}\}$, and $\{\conj{\vec{v_{1}}},\ldots,\conj{\vec{v_{m}}}\}$ such that $\{\vec{v_{1}}, \ldots, \vec{v_{m}}, \conj{\vec{v_{1}}},\ldots,\conj{\vec{v_{m}}}\} = \{\matS_{\_{},i} \mid \matS_{\_{},i} \not \in \RR^{n}\}$.

  For $j \in [m]$ consider the vectors $\vec{u_{j}} = \vec{v_{j}} + \conj{\vec{v_{j}}} \in \RR^{n}$ and $\vec{w_{j}} = i\vec{v_{j}} + (-i) \conj{\vec{v_{j}}} \in \RR^{n}$.
  The vectors in $\{\matS_{\_{},1}, \ldots, \matS_{\_{},n}\}$ are linearly independent as they constitute an eigenbasis.
  Hence, the $k + 2m = n$ vectors in $\{\vec{r_{1}},\ldots,\vec{r_{k}},\vec{u_{1}},\ldots,\vec{u_{m}}, \vec{w_{1}}, \ldots, \vec{w_{m}}\} \subseteq \RR^{n}$ are linearly independent as well and form a basis for the $\RR$-vector space $\RR^{n}$.

  Let $\vecx \in \RR^{n}$.
  The previous paragraph implies that $\vecx$ can be uniquely represented as a linear combination
  \begin{align*}
    \vecx & = \sum_{j\in[k]} c_{j} \vec{r_{k}} + \sum_{j\in[m]} (c'_{j} \vec{u_{i}} + c''_{j} \vec{w_{m}}) \\
          & = \sum_{j\in[k]} c_{j} \vec{r_{k}} + \sum_{j\in[m]} ((c'_{j} + ic''_{j}) \vec{v_{j}} + (c'_{j} - ic''_{j}) \conj{\vec{v_{j}}})
  \end{align*}
  of the column vectors $\{\vec{r_{1}}, \ldots, \vec{r_{k}}, \vec{v_{1}}, \ldots, \vec{v_{m}}, \conj{\vec{v_{1}}}, \ldots, \conj{\vec{v_{m}}}\}$ of the matrix $\matS$ for real scalars $c_{1},\ldots,c_{k},\allowbreak c'_{1},\ldots,c'_{m},\allowbreak c''_{1},\ldots, c''_{m} \in \RR$.

  Let $\operatorname{J}\colon \{\matS_{\_{},1}, \ldots, \matS_{\_{},n}\} \to \{1,\ldots,n\}$ be the bijection with $\operatorname{J}(\matS_{\_{},j}) = j$ for all $j \in [n]$.
  Note that $\operatorname{J}$ is well defined as the column vectors of $\matS$ are linearly independent and hence pairwise different.
  Thus, $\matS^{-1} \vecx = \vecy$ with $\vecy_{\operatorname{J}(\vec{r_{j}})} = c_{j}$ for all $j\in[k]$ and $\vecy_{\operatorname{J}(\vec{v_{j}})} = c'_{j} + i c''_{j}, \vecy_{\operatorname{J} (\conj{\vec{v_{j}}})} = c'_{j} - i c''_{j} = \conj{\vecy_{\operatorname{J} (\vec{v_{j}})}}$ for all $j\in[m]$.

  One then concludes the proof by setting $j = i$ if $i = \operatorname{J}
    (\vec{r_{h}})$ for some $h \in [k]$, $j = \operatorname{J} (\conj{\vec{v_{h}}})$ if $i = \operatorname{J}(\vec{v_{h}})$ for some $h \in [m]$, and $j = \operatorname{J}
    (\vec{v_{h}})$ if $i = \operatorname{J} (\conj{\vec{v_{h}}})$ for some $h \in [m]$.
\end{proof}

%% file: programs_and_termination_app.tex
\proofsForSection{sect:programsTermination}
\label{app:programsTermination}

\valuesofconstraints*{}
\begin{proof}
  Let $\matS$ be the matrix that simultaneously diagonalizes $\matA$ and $\matB$.
  Then, for all inputs $\vecx\in\AA^{n}$ and $f\in\Path$ we have
  \begin{align*}
    \Val_{\vecx} (f) & = \matC \cdot \matA^{|f|_{\symMat{A}}} \cdot \matB^{|f|_{\symMat{B}}} \cdot \vecx \\
                     & = \matC \cdot {(\matS \matAD \matS^{-1})}^{|f|_{\symMat{A}}} \cdot {(\matS \matBD \matS^{-1})}^{|f|_{\symMat{B}}}\cdot \vecx \\
                     & = \matC \cdot \matS \cdot \diag(a_{1}^{|f|_{\symMat{A}}} \cdot b_{1}^{|f|_{\symMat{B}}}, \ldots, a_{n}^{|f|_{\symMat{A}}} \cdot b_{n}^{|f|_{\symMat{B}}}) \cdot \matS^{-1} \cdot \vecx \qedhere
  \end{align*}
\end{proof}

\sumsofconjugatedaddendsarerealvalued*
\begin{proof}
  W.l.o.g.\ we can assume that $a$ is an eigenvalue of $\matA$ and $b$ is an eigenvalue of $\matB$ as otherwise $\gamma_{1} = \gamma_{2} = 0$.

  Recall that the $\gamma_{c,i} (\vecx)$ are obtained by computing the closed form of the expression ${(\Val_{\vecx}(f))}_{c} = \matC \cdot \matA^{|f|_{\symMat{A}}} \cdot \matB^{|f|_{\symMat{B}}} \cdot \vecx$.
  To that end, we considered the simultaneous diagonalization $\mat{A_{D}} = \diag{(a_{1},\ldots,a_{n})} = \mat{S}^{-1} \matA \mat{S}$ and $\mat{B_{D}} = \diag(b_{1},\ldots,b_{n}) = \mat{S}^{-1} \matB \mat{S}$ of $\matA,\matB$ and in $\gamma_{c,i} (\vecx)$, we collected the resulting coefficients of the exponential expressions $a_{i}^{|f|_{\symMat{A}}} b_{i}^{|f|_{\symMat{B}}}$, i.e., $\gamma_{c,i}(\vecx) = {(\matC \cdot \matS)}_{c,i} \cdot {(\matS^{-1} \cdot \vecx)}_{i}$, see \cref{def:gammaci}.

  Consider $
    \begin{pmatrix}
      c_{1}(\vecx) & \ldots & c_{n}(\vecx)
    \end{pmatrix}
    ^{\T} = \mat{S}^{-1} \vecx$ where $c_{i} \colon \AA^{n} \to \CC$ with $i\in[n]$ are linear functions. According to \Cref{lem:conjugate_coefficients_in_eigenbasis_representation} for every $i \in [n]$ there is some unique $j \in [n]$ with conjugated column vectors $\matS_{\_{},i} = \conj{\matS_{\_{},j}}$ and $c_{i}(\vecx) = \conj{c_{j}(\vecx)}$. Moreover, $\mat{S}_{\_{},i}$ is an eigenvector for the eigenvalue $a_i$ of $\matA$ as well as $b_i$ of $\matB$ and $\mat{S}_{\_{},j}$ is an eigenvector for the conjugated eigenvalue $a_j = \conj{a_i}$ of $\matA$ as well as $b_j = \conj{b_i}$ of $\matB$. Hence,
  \begin{equation}
    \mat{A_{D}}^{|f|_{\symMat{A}}} \cdot \mat{B_{D}}^{|f|_{\symMat{B}}} \cdot \mat{S}^{-1} \cdot \vecx =
    \begin{pmatrix}
      a_{1}^{|f|_{\symMat{A}}} b_{1}^{|f|_{\symMat{B}}} c_{1} (\vecx) \\
      \vdots                                                          \\
      a_{n}^{|f|_{\symMat{A}}} b_{n}^{|f|_{\symMat{B}}} c_{n} (\vecx)
    \end{pmatrix}
    \label{eq:real_valuedness_of_conjugated_constraint_terms_proof1}
  \end{equation}
  such that for all $i\in[n]$ there is some unique $j\in[n]$ with
  \[
    {(\mat{A_{D}}^{|f|_{\symMat{A}}} \cdot \mat{B_{D}}^{|f|_{\symMat{B}}} \cdot
    \mat{S}^{-1}\cdot \vecx)}_{i}
    = \conj{{(\mat{A_{D}}^{|f|_{\symMat{A}}} \cdot \mat{B_{D}}^{|f|_{\symMat{B}}}
    \cdot \mat{S}^{-1} \cdot \vecx)}_{j}}
  \] for all $\vecx \in \AA^{n}$.

  By multiplying \cref{eq:real_valuedness_of_conjugated_constraint_terms_proof1} with $\mat{S}$ from the left one obtains
  \begin{align*}
        & \mat{S} \cdot \mat{A_{D}}^{|f|_{\symMat{A}}} \cdot \mat{B_{D}}^{|f|_{\symMat{B}}} \cdot \mat{S}^{-1} \cdot \vecx \\
    ={} &
    \begin{pmatrix}
      \mat{S}_{1,1} a_{1}^{|f|_{\symMat{A}}} b_{1}^{|f|_{\symMat{B}}} c_{1} (\vecx)+ \cdots + \mat{S}_{1,n} a_{n}^{|f|_{\symMat{A}}} b_{n}^{|f|_{\symMat{B}}} c_{n} (\vecx) \\
      \vdots                                                                                                                                                                \\
      \mat{S}_{n,1} a_{1}^{|f|_{\symMat{A}}} b_{1}^{|f|_{\symMat{B}}} c_{1}(\vecx) + \cdots + \mat{S}_{n,n} a_{n}^{|f|_{\symMat{A}}} b_{n}^{|f|_{\symMat{B}}} c_{n} (\vecx) \\
    \end{pmatrix}
    \\
    ={} &
    \begin{pmatrix}
      \sum_{i=1}^{n} \left(a_{i}^{|f|_{\symMat{A}}} b_{i}^{|f|_{\symMat{B}}} c'_{1,i} (\vecx) + \conj{a_{i}}^{|f|_{\symMat{A}}} \conj{b_{i}}^{|f|_{\symMat{B}}} \conj{c'_{1,i} (\vecx)} \right) \\
      \vdots                                                                                                                                                                                    \\
      \sum_{i=1}^{n} \left(a_{i}^{|f|_{\symMat{A}}} b_{i}^{|f|_{\symMat{B}}} c'_{n,i} (\vecx) + \conj{a_{i}}^{|f|_{\symMat{A}}} \conj{b_{i}}^{|f|_{\symMat{B}}} \conj{c'_{n,i} (\vecx)} \right)
    \end{pmatrix}
    \in \AA^{n}
  \end{align*}
  for linear functions $c'_{i,i'}\colon \AA^{n} \to \CC$ for $i,i' \in [n]$.

  By multiplying with $\matC \in \AA^{m \times n}$ from the left one then obtains
  \begin{align*}
        & {\left(\matC \cdot \mat{S} \cdot \mat{A_{D}}^{|f|_{\symMat{A}}} \cdot \mat{B_{D}}^{|f|_{\symMat{B}}} \cdot \mat{S}^{-1} \cdot \vecx\right)}_{c} \\
    {}= & \sum_{i=1}^{n} \matC_{c,i} \sum_{j=1}^{n} \left(a_{j}^{|f|_{\symMat{A}}} b_{j}^{|f|_{\symMat{B}}} c'_{i,j} (\vecx) + \conj{a_{j}}^{|f|_{\symMat{A}}} \conj{b_{j}}^{|f|_{\symMat{B}}} \conj{c'_{i,j} (\vecx)} \right)
  \end{align*}
  for all $c \in [m]$.
  Hence, for all $a,b \in \CC$ and $c \in [m]$ we have
  \begin{align*}
    \gamma_{1} & = \sum_{\substack{i \in [n] \\ (a_{i},b_{i}) = (a,b)}} \gamma_{c,i} (\vecx) \\
               & = \sum_{i=1}^{n} \mat{C}_{c,i} \left( \sum_{\substack{j \in [n] \\ (a_{j},b_{j}) = (a,b)}} c'_{i,j} (\vecx) + \sum_{\substack{j \in [n]\\ (a_{j},b_{j}) = (\conj{a},\conj{b})}} \conj{c'_{i,j} (\vecx)} \right)\\
    \gamma_{2} & = \sum_{\substack{i \in [n] \\ (a_{i},b_{i}) = (\conj{a},\conj{b})}} \gamma_{c,i} (\vecx) \\
               & = \sum_{i=1}^{n} \mat{C}_{c,i} \left( \sum_{\substack{j \in [n] \\ (a_{j},b_{j}) = (\conj{a},\conj{b})}} c'_{i,j} (\vecx) + \sum_{\substack{j \in [n]\\ (a_{j},b_{j}) = (a,b)}} \conj{c'_{i,j} (\vecx)} \right)
  \end{align*}
  by the definition of $\gamma_{c,i}$.
  This implies $\gamma_{2} = \conj{\gamma_{1}}$ since $\matC$ is real and every individual addend $c'_{i,j} (\vecx)$ ($\conj{c'_{i,j}(\vecx)}$) in the sum of the expression for $\gamma_{1}$ above directly corresponds to the conjugated addend $\conj{c'_{i,j} (\vecx)}$ ($c'_{i,j} (\vecx)$) in the sum of the expression for $\gamma_{2}$.
\end{proof}

\correspondencebetweenentandnt*{}
\begin{proof}
  Clearly, $\NT \subseteq \ENT$ and hence $\NT \cap \semiring^{n} \neq \emptyset$ implies $\ENT \cap \semiring^{n} \neq \emptyset$.
  On the other hand, let $\ENT \cap \semiring^{n} \neq \emptyset$.
  Then, there is some $\vecx\in\ENT\cap\semiring^{n}$ and by \Cref{def:ent} there are $j,k\in\NN$ such that $\vecy = \matA^{j} \matB^{k} \vec{x} \in \NT$.
  As $\matA,\matB \in \semiring^{n \times n}$, we have $\vecy \in \semiring^{n}$ as well and hence $\ENT\cap\semiring^{n} \neq \emptyset$ implies $\NT \cap \semiring^{n} \neq \emptyset$.
\end{proof}

%% file: on_constraint_terms_app.tex
\proofsForSection{sect:On Constraint Terms}
\label{app:proofs_for_on_constraint_terms}

\normalformofconstraintterms*{}
\begin{proof}
  Pick an arbitrary $f\in\Path$ with $\URVar(f)\in(-p, 1-p)$ and let $a^{|f|_{\symMat{A}}}b^{|f|_{\symMat{B}}}\gamma$ be defined as required for the application of \Cref{lem:normal_form_of_constraint_terms}.
  We show \[a^{|f|_{\symMat{A}}}b^{|f|_{\symMat{B}}} = \zeta_{\matA}^{|f|_{\symMat{A}}}\zeta_{\matB}^{|f|_{\symMat{B}}} {\left( \frac{|a|^{p}}{|b|^{p-1}} {\left(\frac{|a|}{|b|}\right)}^{\URVar(f)}\right)}^{|f|}.\]
  We have
  \begin{align*}
    a^{|f|_{\symMat{A}}} b^{|f|_{\symMat{B}}} & = {(|a|\zeta_{\matA})}^{|f|_{\symMat{A}}} {(|b|\zeta_{\matB})}^{|f|_{\symMat{B}}} \\
                                              & =
    \zeta_{\matA}^{|f|_{\symMat{A}}}\zeta_{\matB}^{|f|_{\symMat{B}}}
    {\left(|a|^{\frac{|f|_{\symMat{A}}}{|f|}}
    |b|^{\frac{|f|_{\symMat{B}}}{|f|}}\right)}^{|f|} \tag{as $|f| > 0$} \\
                                              & = \zeta_{\matA}^{|f|_{\symMat{A}}}\zeta_{\matB}^{|f|_{\symMat{B}}} {\left(|a|^{\URVar(f)}|a|^{p}|b|^{-\URVar(f)}|b|^{1-p}\right)}^{|f|} \\
                                              & = \zeta_{\matA}^{|f|_{\symMat{A}}}\zeta_{\matB}^{|f|_{\symMat{B}}} {\left(\frac{|a|^{p}}{|b|^{p-1}} {\left(\frac{|a|}{|b|}\right)}^{\URVar(f)}\right)}^{|f|}
  \end{align*}
  where we used $\tfrac{|f|_{\symMat{A}}}{|f|} + \tfrac{|f|_{\symMat{B}}}{|f|} = 1$ in the penultimate line.
\end{proof}

\equalityofeigenvalues*{}
\begin{proof}
  From $\frac{a_{1}^{p}}{b_{1}^{p-1}} = \frac{a_{2}^{p}}{b_{2}^{p-1}}$ we conclude $a_{1} = a_{2} {\left(\frac{b_{1}}{b_{2}}\right)}^{\frac{p-1}{p}}$ and thus, $a_{1} = a_{2} {\left(\frac{b_{2}}{b_{1}}\right)}^{\frac{1-p}{p}}$ with $p\neq0$.
  Substituting this into the remaining equation $\frac{a_{1}}{b_{1}} = \frac{a_{2}}{b_{2}}$, one obtains
  \[
    \frac{a_{2} {\left(\frac{b_{2}}{b_{1}}\right)}^{\frac{1-p}{p}}}{b_{1}} =
    \frac{a_{2}}{b_{2}} \iff \frac{b_{2}}{b_{1}}
    {\left(\frac{b_{2}}{b_{1}}\right)}^{\frac{1-p}{p}} \!\!\!= 1 \iff
    {\left(\frac{b_{2}}{b_{1}} \right)}^{\frac{1}{p}} \!\! = 1 .
  \]
  Thus, $b_{1} = b_{2}$ and from $\frac{a_{1}}{b_{1}} = \frac{a_{2}}{b_{2}}$ we directly obtain that also $a_{1} = a_{2}$.
\end{proof}

\coefficientsofconstrainttermgroupsarerealvalued*
\begin{proof}
  Let $(\scaleinside,\scaleoutside) \in \mathcal{I}$, $\vecx \in \AA^{n}$, and $c \in [m]$.
  First note that for all $i,j \in [n]$ with $(|a_{i}|,|b_{i}|) = (|a_{j}|,|b_{j}|)$ we have $i \in \constrainttermgroup_{(\scaleinside,\scaleoutside),c,\vecx}$ iff $j \in \constrainttermgroup_{(\scaleinside,\scaleoutside),c,\vecx}$ by \Cref{def:constraint_term_groups}.
  We partition the set $\{ (a_{i},b_{i}) \mid i \in \constrainttermgroup_{(\scaleinside,\scaleoutside),c,\vecx}\}$ into the sets $\mathfrak{R} \subseteq \AA_{>0}^{2}$ and $\mathfrak{C}_{1}, \mathfrak{C}_{2} \subseteq \{ (a,b) \mid \{a,b\} \not \subseteq \AA_{>0}\}$ such that $\mathfrak{C}_{2} = \{ (\conj{a},\conj{b}) \mid (a,b) \in \mathfrak{C}_{1}\}$.
  Then, we have
  \begin{align*}
        & \sum_{i \in \constrainttermgroup_{(\scaleinside,\scaleoutside),c,\vecx}} \zeta_{i,\matA}^{|f|_{\symMat{A}}} \zeta_{i,\matB}^{|f|_{\symMat{B}}} \gamma_{c,i} (\vecx) \\
    ={} & \sum_{(a,b) \in \mathfrak{R}} \sum_{\substack{i \in [n] \\ (a_{i},b_{i}) = (a,b)}} \gamma_{c,i} (\vecx) \\
        & {}+ \sum_{(a,b) \in \mathfrak{C}_{1}} \sum_{\substack{i \in [n] \\ (a_{i},b_{i}) = (a,b)}} \zeta_{i,\matA}^{|f|_{\symMat{A}}} \zeta_{i,\matB}^{|f|_{\symMat{B}}} \gamma_{c,i} (\vecx) + \sum_{(a,b) \in \mathfrak{C}_{2}} \sum_{\substack{i \in [n] \\ (a_{i},b_{i}) = (a,b)}} \zeta_{i,\matA}^{|f|_{\symMat{A}}} \zeta_{i,\matB}^{|f|_{\symMat{B}}} \gamma_{c,i} (\vecx) \\
    ={} & \sum_{(a,b) \in \mathfrak{R}} \!\!\!\! \sum_{\substack{i \in [n] \\ (a_{i},b_{i}) = (a,b)}}\!\!\!\!\!\! \gamma_{c,i} (\vecx) \!\!\!\!\sum_{(a,b) \in \mathfrak{C}_{1}} \left(\sum_{\substack{i \in [n] \\ (a_{i},b_{i}) = (a,b)}}\!\!\!\!\!\! \zeta_{i,\matA}^{|f|_{\symMat{A}}} \zeta_{i,\matB}^{|f|_{\symMat{B}}} \gamma_{c,i} (\vecx) \right. \\ &\left.{}+ \sum_{\substack{i \in [n]\\ (a_{i},b_{i}) = (\conj{a},\conj{b})}} \zeta_{i,\matA}^{|f|_{\symMat{A}}} \zeta_{i,\matB}^{|f|_{\symMat{B}}} \gamma_{c,i} (\vecx) \right)
  \end{align*}
  Now, \Cref{lem:sums_of_conjugated_addends_are_real_valued} directly implies that the left addend is real.
  To prove the same for the right addend, let $(a,b) \in \mathfrak{C}_{1}$.
  For the inner sum of the right addend we then obtain
  \begin{align*}
        & \sum_{\substack{i \in [n] \\ (a_{i},b_{i}) = (a,b)}}\!\!\!\!\!\! \zeta_{i,\matA}^{|f|_{\symMat{A}}} \zeta_{i,\matB}^{|f|_{\symMat{B}}} \gamma_{c,i} (\vecx) + \!\!\!\!\sum_{\substack{i \in [n]\\ (a_{i},b_{i}) = (\conj{a},\conj{b})}}\!\!\!\!\!\! \zeta_{i,\matA}^{|f|_{\symMat{A}}} \zeta_{i,\matB}^{|f|_{\symMat{B}}} \gamma_{c,i} (\vecx) \\
    ={} & \left(\frac{a}{|a|}\right)^{|f|_{\symMat{A}}}
    \left(\frac{b}{|b|}\right)^{|f|_{\symMat{B}}} \!\!\!\!\sum_{\substack{i \in [n] \\ (a_{i},b_{i}) = (a,b)}}\!\!\!\!\!\! \gamma_{c,i}
    (\vecx)
    + \left(\frac{\conj{a}}{|a|}\right)^{|f|_{\symMat{A}}} \left(\frac{\conj{b}}{|b|}\right)^{|f|_{\symMat{B}}} \!\!\!\!\sum_{\substack{i \in [n] \\ (a_{i},b_{i}) =
    (\conj{a},\conj{b})}}\!\!\!\!\!\! \gamma_{c,i} (\vecx) \tag{Def.\ of $\zeta_{i,\matA}$ and
    $\zeta_{i,\matB}$} \\
    ={} & \left(\frac{a}{|a|}\right)^{|f|_{\symMat{A}}} \left(\frac{b}{|b|}\right)^{|f|_{\symMat{B}}} \!\!\!\! \sum_{\substack{i \in [n] \\ (a_{i},b_{i}) = (a,b)}} \!\!\!\!\!\! \gamma_{c,i} (\vecx)
    + \left(\frac{\conj{a}}{|a|}\right)^{|f|_{\symMat{A}}} \left(\frac{\conj{b}}{|b|}\right)^{|f|_{\symMat{B}}} \conj{\sum_{\substack{i \in [n] \\ (a_{i},b_{i}) = (a,b)}} \!\!\!\!\!\! \gamma_{c,i} (\vecx)} \tag{\Cref{lem:sums_of_conjugated_addends_are_real_valued}} \\
    ={} & 2 \Re \left(\left(\frac{a}{|a|}\right)^{|f|_{\symMat{A}}} \left(\frac{b}{|b|}\right)^{|f|_{\symMat{B}}} \sum_{\substack{i \in [n] \\ (a_{i},b_{i}) = (a,b)}} \gamma_{c,i} (\vecx) \right)
  \end{align*}
  which is real-valued as well.
  Thus, as $\zeta_{i,\matA},\zeta_{i,\matB},\gamma_{c,i}(\vecx) \in \QQbar$ for all $i \in \constrainttermgroup_{(\scaleinside,\scaleoutside),c,\vecx}$, we have
  \[\sum_{i \in \constrainttermgroup_{(\scaleinside,\scaleoutside),c,\vecx}}
    \zeta_{i,\matA} \zeta_{i,\matB} \gamma_{c,i} (\vecx) \in \QQbar \cap \RR = \AA .\]
\end{proof}

\onthecomparisonofconstrainttermgroups*{}
\begin{proof}
  Whenever $T_{1} \neq T_{2}$ one can simply check whether $T_{1} < T_{2}$ by iteratively computing closer rational approximations.
  Thus, consider the case $T_{1} = T_{2}$.
  If $p \in \QQ$, then $T_{1}$ and $T_{2}$ are computable algebraic reals for which equality can easily be checked.
  So assume $p \in \AA\setminus\QQ$.
  Consider the rearrangement
  \begin{alignat*}{3}
     &      & T_{1}                                            & = T_{2} \\
     & \iff & \frac{a_{1}^{p}}{b_{1}^{p-1}}                    & = \frac{a_{2}^{p}}{b_{2}^{p-1}} \\
     & \iff & \frac{a_{1}^{p}b_{1}}{b_{1}^{p}}                 & = \frac{a_{2}^{p}b_{2}}{b_{2}^{p}} \\
     & \iff & {\left(\frac{a_{1}b_{2}}{b_{1}a_{2}}\right)}^{p} & = \frac{b_{2}}{b_{1}}
  \end{alignat*}
  of the equation $T_{1} = T_{2}$.
  Recall that $b_{1},b_{2} \in \AA_{>0}$ and hence $\frac{b_{2}}{b_{1}} \in \AA_{>0}$.
  However, following the Gelfond-Schneider Theorem~\cite[Thm.\ 3.0.1]{transcendentalnumbers}, ${\left(\frac{a_{1}b_{2}}{b_{1}a_{2}}\right)}^{p} \in \AA$ iff $\frac{a_{1}b_{2}}{b_{1}a_{2}} \in \{0,1\}$.
  Hence, the equation can also be checked in a trivial manner for (algebraic) irrational $p$.
\end{proof}

\dominationofeventuallydominatingconstrainttermgroups*{}
\begin{proof}
  W.l.o.g.\ assume $d = \pmax$ as the other case is symmetric.
  Moreover, we assume $\constrainttermgroup_{\pmax,c,\vecx} \neq \emptyset$ as otherwise $|v(f)| = 0 \not \geq \rho > 0$.
  So, let $(\scaleinside_{\max},\scaleoutside_{\max}) \in \mathcal{I}$ such that $\constrainttermgroup_{(\scaleinside_{\max},\scaleoutside_{\max})} = \constrainttermgroup_{(\scaleinside_{\max},\scaleoutside_{\max}), c, \vecx} = \constrainttermgroup_{\pmax,c,\vecx}$.
  In the spirit of \cref{eq:complete_constraint_in_grouped_normal_form} let
  \[
    w = \!\!\!\!\!\!\sum_{\substack{(\scaleinside,\scaleoutside)\in\mathcal{I}\\
        (\scaleinside,\scaleoutside) \neq (\scaleinside_{\max},\scaleoutside_{\max})}}\!\!\!\!
    \left( {\left(\scaleinside \cdot \scaleoutside^{\URVar(f)}\right)}^{|f|} \!\!\!\!\sum_{i\in\constrainttermgroup_{(\scaleinside,\scaleoutside)}}\!\!\!\!\!\! \zeta_{i,\matA}^{|f|_{\symMat{A}}} \zeta_{i,\matB}^{|f|_{\symMat{B}}} \gamma_{c,i}(\vecx)\right).
  \]

  Note that
  \begin{align*}
    {(\Val_{\vecx}(f))}_{c} & = {(\scaleinside_{\max} \cdot \scaleoutside_{\max}^{\URVar(f)})}^{|f|}
    \!\!\!\!\sum_{i\in\constrainttermgroup_{\pmax,c,\vecx}}\!\!\!\!\!\! \zeta_{i,\matA}^{|f|_{\symMat{A}}}
    \zeta_{i,\matB}^{|f|_{\symMat{B}}} \gamma_{c,i}(\vecx) + w \\
                            & = {(\scaleinside_{\max} \cdot \scaleoutside_{\max}^{\URVar(f)})}^{|f|} \cdot v(f) + w.
  \end{align*}
  As $|v(f)| \geq \rho$, to prove $\sign ({(\Val_{\vecx}(f))}_{c}) = \sign(v (f))$, it suffices to show ${(\scaleinside_{\max} \cdot \scaleoutside_{\max}^{\URVar(f)})}^{|f|}
    \rho > |w| \iff 0 < {(\scaleinside_{\max} \cdot \scaleoutside_{\max}^{\URVar(f)})}^{|f|} \rho - |w|$ for all $f\in\Path$ meeting the requirements.

  The goal is to show that the addend for $(\scaleinside_{\max}, \scaleoutside_{\max})$ dominates the sum of all other addends.
  To over-estimate the sum of the other addends, we consider the ``second largest'' pairs of numbers, i.e., pairs of numbers which are as large as possible, but smaller than $(\scaleinside_{\max}, \scaleoutside_{\max})$.
  There are two reasons why a pair of numbers can be smaller than $(\scaleinside_{\max}, \scaleoutside_{\max})$ in the lexicographic ordering.
  The first reason is that the number in the first component can be smaller than $\scaleinside_{\max}$.
  Thus, let $\scaleinside_{\max 2} < \scaleinside_{\max}$ be the second largest number such that $\constrainttermgroup_{(\scaleinside_{\max 2},\scaleoutside),c,\vecx} \neq \emptyset$ for some $\scaleoutside\in\RR_{>0}$ or $\scaleinside_{\max 2} = 0$ if no such $\scaleoutside$ exists.
  To obtain a pair $(\scaleinside_{\max 2},\ldots)$ which is ``as large as possible'', in the second component, we use the largest number $\scaleoutside_{\max'} \geq \scaleoutside_{\max}$ such that $\constrainttermgroup_{(\scaleinside,\scaleoutside_{\max'}),c,\vecx} \neq \emptyset$ for some $\scaleinside\in\RR_{>0}$.
  So the first ``second largest'' pair of number that we regard is $(\scaleinside_{\max 2},\scaleoutside_{\max'})$.
  The second reason why a pair of numbers can be smaller than $(\scaleinside_{\max}, \scaleoutside_{\max})$ in the lexicographic ordering is when the first component stays the same, but the second component gets smaller.
  Hence, let $\scaleoutside_{\max 2} < \scaleoutside_{\max}$ be the second largest number such that $\constrainttermgroup_{(\scaleinside_{\max},\scaleoutside_{\max 2}),c,\vecx} \neq \emptyset$ or $0$ if no such value exists.
  So the other ``second largest'' pair of number that we regard is $(\scaleinside_{\max}, \scaleoutside_{\max 2})$.

  Next, we define a \underline{\textbf{s}}caling factor $\mathfrak{s}$ that compares the largest pair $(\scaleinside_{\max}, \scaleoutside_{\max})$ with the second largest pair $(\scaleinside_{\max 2}, \scaleoutside_{\max'})$.
  To this end, we choose $\eps > 0$ suitably small such that $\mathfrak{s} = \max_{u\in[0,\eps]}
    \frac{\scaleinside_{\max 2} \cdot \scaleoutside_{\max'}^{u}}{\scaleinside_{\max}
      \cdot \scaleoutside_{\max}^{u}} = \frac{\scaleinside_{\max 2} \cdot \scaleoutside_{\max'}^{\eps}}{\scaleinside_{\max}
      \cdot \scaleoutside_{\max}^{\eps}}
    < 1$. Such an $\eps$ always exists, since $\scaleinside_{\max 2} < \scaleinside_{\max}$ (but $\scaleoutside_{\max'} \geq \scaleoutside_{\max}$). Thus, we choose $\eps$ in such a way that as long as $\URVar(f) \in [0,\eps]$, we know that $\scaleinside_{\max 2} \cdot \scaleoutside_{\max'}^{\URVar(f)}$ is smaller than $\scaleinside_{\max} \cdot \scaleoutside_{\max}^{\URVar(f)}$ by a factor of at least $\mathfrak{s}$. Thus, we have
  \begin{align}
    |w| & \leq \!\!\!\!\sum_{\substack{(\scaleinside,\scaleoutside)\in\mathcal{I} \nonumber \\
        (\scaleinside,\scaleoutside) \neq (\scaleinside_{\max},\scaleoutside_{\max})}}\!\!\!\!
    \left(
    {\left(\scaleinside \cdot \scaleoutside^{\URVar(f)}\right)}^{|f|}
    \sum_{i\in\constrainttermgroup_{(\scaleinside,\scaleoutside)}}
    |\gamma_{c,i}(\vecx)| \right) \tag{Triang. Ineq.} \nonumber \\
        & \leq \left( {\left( \scaleinside_{\max} \cdot \scaleoutside_{\max 2}^{\URVar(f)}
      \right)}^{|f|}
    + {\left( \scaleinside_{\max 2} \cdot \scaleoutside_{\max'}^{\URVar(f)}\right)}^{|f|}\right) \cdot \sum_{i\in[n]} |\gamma_{c,i} (\vecx)| \nonumber \\
        & \leq \left( {\left( \scaleinside_{\max} \cdot \scaleoutside_{\max 2}^{\URVar(f)}
      \right)}^{|f|} + {\left( \mathfrak{s} \cdot \scaleinside_{\max} \cdot
      \scaleoutside_{\max}^{\URVar(f)}\right)}^{|f|}\right) \cdot \sum_{i\in[n]} |\gamma_{c,i} (\vecx)| \label{w expression}
  \end{align}
  Here, the first inequation holds because of the triangle inequation ($|x + y| \leq |x| + |y|$) and the fact that $|\zeta_{i,\matA}^{|f|_{\symMat{A}}}
    \zeta_{i,\matB}^{|f|_{\symMat{B}}}| = 1$. In the second inequation, one over-approximates the coefficients $\sum_{i\in\constrainttermgroup_{(\scaleinside,\scaleoutside)}}
    |\gamma_{c,i}(\vecx)|$ by $\sum_{i\in[n]} |\gamma_{c,i} (\vecx)|$ and the expression ${\left(\scaleinside \cdot \scaleoutside^{\URVar(f)}\right)}^{|f|}$ by ${\left( \scaleinside_{\max} \cdot \scaleoutside_{\max 2}^{\URVar(f)}
        \right)}^{|f|}$ and ${\left( \scaleinside_{\max 2} \cdot \scaleoutside_{\max'}^{\URVar(f)}\right)}^{|f|}$. In the last inequation, we use that $\URVar(f) \in [0,\eps]$ and thus, $\scaleinside_{\max 2} \cdot \scaleoutside_{\max'}^{\URVar(f)} \leq \mathfrak{s} \cdot \scaleinside_{\max} \cdot \scaleoutside_{\max}^{\URVar(f)}$.

  To compare the largest pair $(\scaleinside_{\max}, \scaleoutside_{\max})$ with the other second largest pair $(\scaleinside_{\max}, \scaleoutside_{\max 2})$, we choose a large enough number $r$ such that for all $|\URVar(f) \cdot |f|| \ge r$, $\scaleoutside_{\max}^{r}$ multiplied by its coefficient is greater than $\scaleoutside_{\max 2}^r$ multiplied by the sum of the coefficients of all other addends.
  So we choose $r \in \NN$ large enough such that $t \cdot \scaleoutside_{\max}^{r} > \scaleoutside_{\max 2}^{r} \cdot \sum_{i\in[n]} |\gamma_{c,i} (\vecx)|$ for some $t\in(0,\rho)$.
  Here, $t$ is smaller than the coefficient of the dominating addend (i.e., $t < \rho$).
  We need $t < \rho$, because the dominating addend with its actual coefficient must also ``outgrow'' the addend corresponding to $(\scaleinside_{\max 2}, \scaleoutside_{\max'})$.
  Note that by the choice of $r$ above, we also have $t \cdot \scaleoutside_{\max}^{r'} > \scaleoutside_{\max 2}^{r'} \cdot \sum_{i\in[n]} |\gamma_{c,i} (\vecx)|$ for all $r' \geq r$, because $\scaleoutside_{\max} > \scaleoutside_{\max 2}$.
  Thus, since $|\URVar(f) \cdot |f|| = \URVar(f) \cdot |f| \geq r$, we also have $t \cdot \scaleoutside_{\max}^{\URVar(f) \cdot |f|}
    > \scaleoutside_{\max 2}^{\URVar(f) \cdot |f|} \cdot \sum_{i\in[n]} |\gamma_{c,i} (\vecx)|$.

  Substituting \cref{w expression} into the expression ${(\scaleinside_{\max} \cdot \scaleoutside_{\max}^{\URVar(f)})}^{|f|} \rho - |w|$ one obtains
  \begin{align*}
           & {(\scaleinside_{\max} \cdot \scaleoutside_{\max}^{\URVar(f)})}^{|f|} \rho - |w| \\
    \geq{} & {(\scaleinside_{\max} \cdot \scaleoutside_{\max}^{\URVar(f)})}^{|f|} \rho -
    \left( {\left( \scaleinside_{\max} \cdot \scaleoutside_{\max 2}^{\URVar(f)} \right)}^{|f|} + {\left( \mathfrak{s} \cdot \scaleinside_{\max} \cdot \scaleoutside_{\max}^{\URVar(f)}\right)}^{|f|}\right) \sum_{i\in[n]} |\gamma_{c,i} (\vecx)| \\
    \geq{} & {(\scaleinside_{\max} \cdot \scaleoutside_{\max}^{\URVar(f)})}^{|f|} \left(\rho - \mathfrak{s}^{|f|} \sum_{i\in[n]} |\gamma_{c,i} (\vecx)|\right) - \scaleinside_{\max}^{|f|} \cdot \scaleoutside_{\max 2}^{\URVar(f)\cdot |f|} \sum_{i\in[n]} |\gamma_{c,i} (\vecx)| \\
    >{}    & {(\scaleinside_{\max} \cdot \scaleoutside_{\max}^{\URVar(f)})}^{|f|} \left(\rho - \mathfrak{s}^{|f|} \sum_{i\in[n]} |\gamma_{c,i} (\vecx)|\right) - \scaleinside_{\max}^{|f|} \cdot \scaleoutside_{\max}^{\URVar(f) \cdot |f|} \cdot t \\
    \geq{} & {(\scaleinside_{\max} \cdot \scaleoutside_{\max}^{\URVar(f)})}^{|f|} \left(\rho - t - \mathfrak{s}^{|f|} \sum_{i\in[n]} |\gamma_{c,i} (\vecx)|\right)
  \end{align*}
  Clearly, ${\scaleinside_{\max} \cdot \scaleoutside_{\max}^{\URVar(f)}}^{|f|} > 0$.
  By definition of $t$ we have $\rho - t > 0$.
  Recall that $\mathfrak{s}\in(0,1)$.
  The statement then follows by choosing $l$ large enough and considering $|f| \geq l$.
\end{proof}

%% file: positive_eigenvalues_app.tex
\proofsForSection{sect:Positive Eigenvalues}
\label{app:proofs_for_positive_eigenvalues}

\bravermangeneralisation*{}
\begin{proof}
  The proof below follows the one given by Braverman in spirit.

  Similar to the proof by Braverman we will make use of the following adapted key claims.
  However, unlike~\cite{braverman06}, we give a detailed proof for both of them at the end.
  \input{./braverman_key_claims.tex}

  Note that the first key claim implies that in (\labelcref{it:braverman_generalisation_itemB}) the $z_{j,k}$ cannot be equal for all $j,k\in \NN$.
  The reason is that then we would have $\left| \sum_{j=0}^{N-1} \sum_{k=0}^{N-1} z_{j_{0}+j,k_{0}+k} \right| = N^2 \cdot z_{j,k} \leq N C_1$ for all $N \in \NN_{>0}$ which can only happen if $z_{j,k} = 0$.
  Hence, this would be a contradiction to the first part of (\labelcref{it:braverman_generalisation_itemB}).

  Let us assume that the first item of \Cref{lem:braverman_generalisation} does not hold, i.e., the real part of the sum $z_{j,k}$ is not equal to zero for all $j,k$.
  Then it suffices to show (\labelcref{it:braverman_generalisation_itemB}).

  As in~\cite{braverman06}, from now on we only consider real $z_{j,k}$.
  This is not an actual restriction as one can consider $y_{j,k} = z_{j,k} + \conj{z_{j,k}} = \sum_{i=1}^{l} (\gamma_{i}\zeta_{1,i}^{j}\zeta_{2,i}^{k} + \conj{\gamma_{i}} \conj{\zeta_{1,i}}^{j} \conj{\zeta_{2,i}}^{k})$ instead.
  Note that the two claims also hold for $y_{j_{0}+j,k_{0}+k}$.
  To simplify the notation, from now on we write $z_{j,k}$ to denote $y_{j_{0}+j,k_{0}+k}$ and have $z_{j,k} = \Re(z_{j,k})$.

  We choose $N\in\NN_{>0}$ large enough such that $N^{2} C_{2} \geq 3 N l C_{1}$.
  By the first claim we have
  \begin{equation}
    \label{Re-upper bound}
    \left| \sum_{j=0}^{Nl-1}\sum_{k=0}^{Nl-1} \Re(z_{j,k})\right| = \left| \sum_{j=0}^{Nl-1}\sum_{k=0}^{Nl-1} z_{j,k} \right| \leq NlC_{1} .
  \end{equation}

  On the other hand, the second claim implies
  \begin{align}
    \hspace*{-.4cm} \sum_{j=0}^{Nl-1}\sum_{k=0}^{Nl-1} |\Re(z_{j,k})| & = \!\! \sum_{j=0}^{Nl-1}\sum_{k=0}^{Nl-1} |z_{j,k}| \nonumber \\
                                                                      & = \!\!\sum_{j=0}^{N-1}\sum_{k=0}^{N-1} \sum_{j'=0}^{l-1}\sum_{k'=0}^{l-1} \left| z_{lj+j',lk+k'} \right| \nonumber \\
                                                                      & \geq \!\! \sum_{j=0}^{N-1}\sum_{k=0}^{N-1} C_{2} = N^{2}C_{2} \geq 3NlC_{1}
    \label{Re-lower bound}
  \end{align}

  Now assume that
  \begin{equation}
    \label{Re-contradiction}
    \sum_{j=0}^{Nl-1}\sum_{\substack{k\in\{0,\dots,Nl-1\}\\
        \Re(z_{j,k})<0}} \Re(z_{j,k}) > -NlC_{1}
  \end{equation}
  This would imply
  \[ \left|\sum_{j=0}^{Nl-1}\sum_{\substack{k\in\{0,\dots,Nl-1\}\\ \Re(z_{j,k})<0}} \Re(z_{j,k})\right| < NlC_{1} .\]
  Hence, we get
  \begin{equation}
    \label{Re-contradiction-consequence}
    \sum\limits_{j=0}^{Nl-1}\sum\limits_{\substack{k\in\{0,\dots,Nl-1\}\\
        \Re(z_{j,k})\geq 0}} \Re(z_{j,k}) \geq 2NlC_{1}
  \end{equation}
  because
  \begin{align*}
    3NlC_{1} & \leq \sum\limits_{j=0}^{Nl-1}\sum\limits_{k=0}^{Nl-1} |\Re(z_{j,k})|
    \tag{by \Cref{Re-lower bound}} \\
             & = \sum\limits_{j=0}^{Nl-1}\sum\limits_{\substack{k\in\{0,\dots,Nl-1\} \\ \Re(z_{j,k})<0}} |\Re(z_{j,k})| + \sum\limits_{j=0}^{Nl-1}\sum\limits_{\substack{k\in\{0,\dots,Nl-1\} \\ \Re(z_{j,k})\geq 0}} \Re(z_{j,k})
  \end{align*}
  and therefore
  \begin{align*}
    \sum\limits_{j=0}^{Nl-1}\sum\limits_{\substack{k\in\{0,\dots,Nl-1\} \\ \Re(z_{j,k})\geq 0}}\!\!\!\!\!\!\!\! \Re(z_{j,k}) &\geq 3NlC_{1} - \sum\limits_{j=0}^{Nl-1}\sum\limits_{\substack{k\in\{0,\dots,Nl-1\} \\ \Re(z_{j,k})<0}} \!\!\!\!\!\!\!\! |\Re(z_{j,k})| \\
     & = 3NlC_{1} - \left| \sum\limits_{j=0}^{Nl-1}\sum\limits_{\substack{k\in\{0,\dots,Nl-1\} \\ \Re(z_{j,k})<0}} \!\!\!\!\!\!\!\! \Re(z_{j,k}) \right| \\
     & > 3NlC_{1} - NlC_{1} \tag{by \Cref{Re-contradiction}} \\
     & = 2NlC_{1}
  \end{align*}
  But this implies
  \begin{align*}
    \sum\limits_{j=0}^{Nl-1}\sum\limits_{k=0}^{Nl-1} \Re(z_{j,k})
     & = \sum\limits_{j=0}^{Nl-1}\sum\limits_{\substack{k\in\{0,\dots,Nl-1\} \\ \Re(z_{j,k})\geq 0}}\!\!\!\!\!\!\!\! \Re(z_{j,k}) + \sum\limits_{j=0}^{Nl-1}\sum\limits_{\substack{k\in\{0,\dots,Nl-1\}\\ \Re(z_{j,k}) < 0}}\!\!\!\!\!\!\!\! \Re(z_{j,k}) \\
     & > 2NlC_{1} - NlC_1 \tag{by \Cref{Re-contradiction-consequence,Re-contradiction}} \\
     & = NlC_1
  \end{align*}
  and hence, $\left|\sum\limits_{j=0}^{Nl-1}\sum\limits_{k=0}^{Nl-1} \Re(z_{j,k})\right| > NlC_1$ which is a contradiction to \cref{Re-upper bound}.

  So the assumption \cref{Re-contradiction} was wrong and we have
  \begin{align}
    \sum_{j=0}^{Nl-1}\sum_{\substack{k\in\{0,\dots,Nl-1\} \\ \Re(z_{j,k})<0}} \Re(z_{j,k}) &= \sum_{j=0}^{Nl-1}\sum_{\substack{k\in\{0,\dots,Nl-1\}\\z_{j,k}<0}} z_{j,k} \nonumber \\
     & \leq -NlC_{1}.
    \label{Re-consequence}
  \end{align}
  Assume that for all $j,k$ where $z_{j,k} < 0$, we have $z_{j,k} > -\frac{C_{1}}{Nl}$.
  Then we would have
  \begin{align*}
    \sum\limits_{j=0}^{Nl-1}\sum\limits_{\substack{k\in\{0,\dots,Nl-1\} \\ z_{j,k}<0}} z_{j,k} &> \sum\limits_{j=0}^{Nl-1}\sum\limits_{\substack{k\in\{0,\dots,Nl-1\}\\ z_{j,k}<0}} -\frac{C_{1}}{Nl} \\
     & > (Nl)^2 \cdot \left(-\frac{C_{1}}{Nl}\right) \\
     & = -NlC_{1}
  \end{align*}
  in contradiction to \cref{Re-consequence}.
  Thus there exist some $j,k\in\{0,\ldots,Nl-1\}$ with $\Re(z_{j,k}) = z_{j,k} \leq -\frac{C_{1}}{Nl}$.
  We set $K=Nl-1$ and $C = -\frac{C_{1}}{Nl}$ to conclude the proof.
  \input{./braverman_generalisation_claims_app.tex}
\end{proof}

\finiteexecutionleadingtotermination*{}
\begin{proof}
  W.l.o.g.\ assume $d = \pmax$ as the other case is symmetric.
  We exclude the case $\constrainttermgroup_{\pmax,c,\vecx} = \emptyset$ as in this case one can simply set $g_f$ to be the empty path and $u=0$, resulting in ${(\Val_{\vecx}(f\, g_f))}_{c}
    = 0$.

  For all finite executions $f\in\Path$, we have
  \begin{equation*}
    \sum_{i\in\constrainttermgroup_{\pmax,c,\vecx}} \zeta_{i,\matA}^{|f|_{\symMat{A}}} \zeta_{i,\matB}^{|f|_{\symMat{B}}} \gamma_{c,i}(\vecx) = \sum_{i\in\mathfrak{R}_{\pmax,c,\vecx}} \gamma_{c,i}(\vecx) + \sum_{i\in\mathfrak{C}_{\pmax,c,\vecx}} \zeta_{i,\matA}^{|f|_{\symMat{A}}} \zeta_{i,\matB}^{|f|_{\symMat{B}}} \gamma_{c,i}(\vecx) .
  \end{equation*}

  First note that there are constants $C \in \AA_{<0}$ and $K\in\NN$ such that for all $f\in\Path$ we have
  \begin{equation}
    \label{FiniteExecutionTerminationUpperBound}
    \sum_{i\in\constrainttermgroup_{\pmax,c,\vecx}} \zeta_{i,\matA}^{|f|_{\symMat{A}}+j} \zeta_{i,\matB}^{|f|_{\symMat{B}}+k} \gamma_{c,i}(\vecx) \leq C < 0
  \end{equation}
  for some $j,k\in\{0,\ldots,K\}$.
  To see this, note that
  \[\Re(\sum_{i\in\mathfrak{C}_{\pmax,c,\vecx}}
    \zeta_{i,\matA}^{|f|_{\symMat{A}}} \zeta_{i,\matB}^{|f|_{\symMat{B}}}
    \gamma_{c,i}(\vecx)) = \sum_{i\in\mathfrak{C}_{\pmax,c,\vecx}}
    \zeta_{i,\matA}^{|f|_{\symMat{A}}} \zeta_{i,\matB}^{|f|_{\symMat{B}}}
    \gamma_{c,i}(\vecx)\] by \Cref{rem:sums_of_constraint_term_groups_are_real_valued}. If $\Re(\sum_{i\in\mathfrak{C}_{\pmax,c,\vecx}}
    \zeta_{i,\matA}^{|f|_{\symMat{A}}} \zeta_{i,\matB}^{|f|_{\symMat{B}}}
    \gamma_{c,i}(\vecx)) = 0$, then $\constrainttermgroup_{\pmax,c,\vecx} \neq \emptyset$ implies
  \[\sum_{i\in\constrainttermgroup_{\pmax,c,\vecx}} \zeta_{1,\matA}^{|f|_{\symMat{A}}+j}
    \zeta_{2,\matB}^{|f|_{\symMat{B}}+k} \gamma_{c,i}(\vecx) =
    \sum_{i\in\mathfrak{R}_{\pmax,c,\vecx}} \gamma_{c,i}(\vecx) < 0\] which implies
  \cref{FiniteExecutionTerminationUpperBound}.
  Otherwise, \Cref{lem:braverman_generalisation}(\labelcref{it:braverman_generalisation_itemB}) implies
  \begin{align*}
    \sum_{i\in\mathfrak{C}_{\pmax,c,\vecx}} \zeta_{i,\matA}^{|f|_{\symMat{A}}+j}
    \zeta_{i,\matB}^{|f|_{\symMat{B}}+k} \gamma_{c,i}(\vecx)
     & = \Re(\sum_{i\in\mathfrak{C}_{\pmax,c,\vecx}} \zeta_{i,\matA}^{|f|_{\symMat{A}}+j}
    \zeta_{i,\matB}^{|f|_{\symMat{B}}+k} \gamma_{c,i}(\vecx)) \\
     & \leq C' \quad \text{for some $C' \in \AA_{<0}$,}
  \end{align*}
  which implies \cref{FiniteExecutionTerminationUpperBound} since $\sum_{i\in\mathfrak{R}_{\pmax,c,\vecx}} \gamma_{c,i}(\vecx) \leq 0$.
  In order to satisfy the requirement of pairwise different $(\zeta_{i,\matA},\zeta_{i,\matB})$ that is needed for \Cref{lem:braverman_generalisation}, one has to group the terms correspondingly.
  For all $f\in\Path$, we define $g_f = \symMat{A}^{j} \symMat{B}^{k}$ and let $u = 2K$ such that $|g_f| \leq u$.
  Then we have
  \begin{equation}
    \label{sign new path}
    \sum_{i\in\constrainttermgroup_{\pmax,c,\vecx}} \zeta_{i,\matA}^{|f\,
    g_f|_{\symMat{A}}} \zeta_{i,\matB}^{|f \, g_f|_{\symMat{B}} }\gamma_{c,i}
    (\vecx)
    \leq C < 0,
  \end{equation}
  i.e., after the path $f \, g_f$, the dominating constraint is negative.

  Let $\rho = |C|$.
  \Cref{lem:domination_of_eventually_dominating_constraint_term_groups} then ensures the existence of constants $\eps' \in \AA_{> 0}$, $r'\in\NN$, and $l'\in\NN_{>0}$ such that for all $f'\in\Path$ with $|f'| \geq l'$, $\URVar(f') \in [0,\eps']$, and $|\URVar(f')\cdot |f'|| = \URVar(f') \cdot |f'| = |f'|_{\symMat{A}} - p \cdot |f'| \geq r'$ we have $\sign({(\Val_{\vecx} (f'))}_{c}) = \sign(\sum_{i\in\constrainttermgroup_{\pmax,c,\vecx}}
    \zeta_{1,\matA}^{|f'|_{\symMat{A}}} \zeta_{2,\matB}^{|f'|_{\symMat{B}}}
    \gamma_{c,i}(\vecx))$.

  We now want to apply \Cref{lem:domination_of_eventually_dominating_constraint_term_groups} to the actual path $f \, g_f$ (i.e., where $f' = f \, g_f$).
  We choose $\eps$ with $0 < \eps < \eps'$ small enough and $l \geq l'$ large enough such that $\eps l + u \leq \eps' l$ and $r > r'$ large enough such that $r - u \geq r'$.
  Let $f\in\Path$ satisfy the prerequisites of \Cref{lem:braverman_generalisation_corollary}.
  Then, we have
  \begin{equation}
    \label{finite execution 1}
    |fg_f|
    \geq |f| \geq l \geq l'.
  \end{equation}
  Moreover, we also have
  \begin{equation}
    \label{finite execution 2}
    \URVar (f\,g_f) \in [0,\eps'].
  \end{equation}
  To show this, we first prove $\URVar (f\,g_f) \geq 0$.
  We have
  \begin{align*}
    \URVar(f\, g_{f}) \cdot |f \, g_{f}| & = |f\,g_f|_{\symMat{A}} - p \cdot |f \, g_f| \\
                                         & = |f|_{\symMat{A}} - p \cdot |f| + |g_f|_{\symMat{A}} - p \cdot |g_f| \\
                                         & = \URVar(f) \cdot |f|+ |g_f|_{\symMat{A}} - p \cdot |g_f| \\
                                         & \geq \URVar(f) \cdot |f| - |g_f| \\
                                         & \geq \URVar(f) \cdot |f| - u \tag{as $|g_{f}| \leq u$} \\
                                         & \geq r - u \tag{by the prerequisites} \\
                                         & \geq r' \tag{by the choice of $r$} \\
                                         & \geq 0
  \end{align*}
  and thus, $\URVar (f\,g_f) \geq 0$.
  To prove $\URVar (f\,g_f) \leq \eps'$, note that
  \[ \begin{array}{rcl}
       &                     & \URVar (f\,g_f) \leq \eps'                                             \\
       & \Longleftrightarrow & \URVar (f\,g_f) \cdot |f \, g_f| \leq \eps' \cdot |f \, g_f|           \\
       & \Longleftrightarrow & |f\,g_f|_{\symMat{A}} - p \cdot |f \, g_f| \leq \eps' \cdot |f \, g_f| \\
       & \Longleftrightarrow & \underbrace{|f|_{\symMat{A}} - p \cdot |f|}_{= \, \URVar(f) \cdot
      |f| \, \leq \, \eps \cdot |f|} +
      \underbrace{|g_f|_{\symMat{A}} - p \cdot |g_f|}_{\leq \, |g_f| \leq u}
      \leq \eps' \cdot |f \, g_f|                                                                     \\
       & \Longleftarrow      & \eps \cdot |f| + u \leq \eps' \cdot |f \, g_f|                         \\
       & \Longleftrightarrow & \eps \cdot l + \eps \cdot (|f| - l) + u \leq \eps' \cdot (l + |f \,
      g_f| - l)                                                                                       \\
       & \Longleftarrow      & \eps \cdot l + \eps' \cdot (|f| - l) + u \leq \eps' \cdot (l + |f \,
      g_f| - l)                                                                                       \\
       & \Longleftrightarrow & \eps \cdot l + u \leq \eps' \cdot l
    \end{array}
  \]
  Finally, we also have
  \begin{equation}
    \label{finite execution 3}
    |\URVar(f\,g_f) \cdot |f\,g_f|| \geq r'.
  \end{equation}
  The reason is that
  \[ \begin{array}{rcl}
      |\URVar(f\,g_f) \cdot |f\,g_f||
       & =    & \URVar(f\,g_f) \cdot |f\,g_f|                                                    \\
       & =    & |f\,g_f|_{\symMat{A}} - p \cdot |f \, g_f|                                       \\
       & \geq & |f|_{\symMat{A}} - p \cdot |f \, g_f|                                            \\
       & =    & \underbrace{|f|_{\symMat{A}} - p \cdot |f|}_{= \, \URVar(f) \cdot |f| \, \geq \,
      r} - p \cdot |g_f|                                                                         \\
       & \geq & r - p \cdot |g_f|                                                                \\
       & \geq & r - p \cdot u                                                                    \\
       & \geq & r - u                                                                            \\
       & \geq & r'
    \end{array}\]
  Hence, by \cref{finite execution 1}, \cref{finite execution 2}, and \cref{finite execution 3} we can apply \Cref{lem:domination_of_eventually_dominating_constraint_term_groups} to the path $f \, g_f$ and conclude that $\sign( (\Val_{\vecx} (f \, g_f))_{c} ) = \sign(\sum_{i\in\constrainttermgroup_{\pmax,c,\vecx}} \zeta_{i,\matA}^{|f\, g_f|_{\symMat{A}}} \zeta_{i,\matB}^{|f\, g_f|_{\symMat{B}} }\gamma_{c,i}
    (\vecx))$. Thus, by \cref{sign new path} we obtain ${(\Val_{\vecx} (fg_f))}_{c} < 0$.
\end{proof}

\dualpositiveeigenvaluesforeventuallydominatingconstraints*{}
\begin{proof}
  We prove \Cref{lem:dual_positive_eigenvalues_for_eventually_dominating_constraints} by contradiction.
  To that end, let $\vecx\in\ENT$ and let $d\in\{\nmax,\pmax\}$ such that for all $c\in[m]$ we have $\sum_{i\in\mathfrak{R}_{d,c,\vecx}} \gamma_{c,i}(\vecx) > 0$.

  W.l.o.g., we may assume $\vecx \in \NT$ as otherwise one would simply consider $\vecy = A^{j}B^{k}\vecx$ where $j,k\in\NN$ such that $\vecy \in \NT$ according to the definition of $\ENT$.
  Then it suffices to show that $\sum_{i\in\mathfrak{R}_{d,c,\vecy}} \gamma_{c,i}(\vecy) > 0$, because that implies $\sum_{i\in\mathfrak{R}_{d,c,\vecx}} \gamma_{c,i}(\vecx) > 0$.
  To see this, recall the definition of $\gamma_{c,i}$.
  We had ${\left(\Val_{\vecx} (f)\right)}_{c} = \sum_{i\in[n]}
    a_{i}^{|f|_{\symMat{A}}}b_{i}^{|f|_{\symMat{B}}}\gamma_{c,i}(\vecx)$. Hence,
  \begin{equation*}
    {\left(\Val_{\vecy} (f)\right)}_{c}
      = {\left(\Val_{A^j B^k \vecx} (f)\right)}_{c}
    = \sum_{i\in[n]} a_{i}^{|f|_{\symMat{A}}+j}b_{i}^{|f|_{\symMat{B}}+k}\gamma_{c,i}(\vecx)
    = \sum_{i\in[n]} a_{i}^{|f|_{\symMat{A}}}b_{i}^{|f|_{\symMat{B}}}\gamma_{c,i}(\vecy),
  \end{equation*}
  see \Cref{GammaAB}.
  Thus, $\gamma_{c,i}(\vecy) = a_i^j b_i^k \gamma_{c,i}(\vecx)$.
  Note that $\mathfrak{R}_{d,c,\vecy}= \mathfrak{R}_{d,c,\vecx}$ holds, since by the prerequisite $\sum_{i\in\mathfrak{R}_{d,c,\vecy}} \gamma_{c,i}(\vecy) > 0$ we have $\mathfrak{R}_{d,c,\vecy} \neq \emptyset$ and $\gamma_{c,i}(\vecy)$ results from $\gamma_{c,i}(\vecx)$ by multiplication with positive reals $a_i^j, b_i^k$.
  Hence, if $\sum_{i\in\mathfrak{R}_{d,c,\vecy}} \gamma_{c,i}(\vecy) = \sum_{i\in\mathfrak{R}_{d,c,\vecy}} a_i^j b_i^k \gamma_{c,i}(\vecx) = \sum_{i\in\mathfrak{R}_{d,c,\vecx}} a_i^j b_i^k \gamma_{c,i}(\vecx)$ is positive, then $\sum_{i\in\mathfrak{R}_{d,c,\vecx}}
    \gamma_{c,i}(\vecx)$ is positive as well. To see this, note that for $i \in \mathfrak{R}_{d,c,\vecx}$, all $(a_{i}, b_{i})$ are (componentwise) equal. The reason is that that by \Cref{lem:shared_modulus_of_eigenvalues_in_constraint_term_groups}, all $a_{i}$ have the same modulus (absolute value) and similarly, also all $b_i$ have the same modulus. Moreover, as $i \in \mathfrak{R}_{d,c,\vecx}$ we have $a_{i},b_{i} \in \AA_{>0}$ by definition of $\mathfrak{R}_{d,c,\vecx}$.

  For a proof by contradiction, we assume the existence of $c_{\nmax},c_{\pmax} \in [m]$ such that we have $\sum_{i\in\mathfrak{R}_{\nmax,c_{\nmax},\vecx}} \gamma_{c_{\nmax},i}(\vecx) \leq 0$ and $\sum_{i\in\mathfrak{R}_{\pmax,c_{\pmax},\vecx}} \gamma_{c_{\pmax},i}(\vecx) \leq 0$.

  We denote the constants from \Cref{lem:braverman_generalisation_corollary} as $\eps_{\nmax}, r_{\nmax}, l_{\nmax}, u_{\nmax}$ for $c = c_{\nmax}, d=\nmax$ and $\eps_{\pmax},r_{\pmax},l_{\pmax},u_{\pmax}$ for $c = c_{\pmax}, d=\pmax$.
  We define $\eps = \min\{\eps_{\nmax},\eps_{\pmax}\}$, $r = \max\{r_{\nmax},r_{\pmax}\}$, $u = \max\{u_{\nmax},u_{\pmax}\}$, and choose $l \geq \max\{l_{\nmax},l_{\pmax},u\}$ large enough such that $\eps \cdot l \geq r+1$.\footnote{This ensures that if $|f|$ is increased by 1, then the value of $\URVar(f) \cdot |f|$ changes by at most 1.
    Therefore, if $\URVar(f) \cdot |f| \geq r$ for the first time, then we also have $\URVar(f) \cdot |f| \leq \eps \cdot |f|$, i.e., $\URVar(f) \cdot |f|$ has reached a safe region.}
  Thus, for every path $f\in\Path$ leading to one of the two ``safe regions'', i.e., $|f| \geq l$, $|\URVar(f) \cdot |f|| \geq r$, and $\URVar(f) \in [-\eps,\eps]$, there is some $g_f$ with $|g_f| \leq u$ such that ${(\Val_{\vecx} (f \, g_f))}_{c_{\nmax}} \leq 0$ or ${(\Val_{\vecx} (f \, g_f))}_{c_{\pmax}} \leq 0$ and hence the path $f \,g_f$ is terminating.

  For all $i\in\NN$, we introduce the ($\F$-measurable) random variables $B_{i}\colon \Runs \to \NN \cup \{\infty\}$ such that $B_{i}$ maps runs to an index where one is in one of the ``safe regions'' specified by the constants $\eps$, $l$, and $r$ for at least the $i$-th time: \input{./dual_positive_eigenvalues_for_eventually_dominating_constraints_def_Xi.tex}
  We ensure that $B_{i}(\run_1 \ldots) - B_{i-1}(\run_1 \ldots) \geq l \geq u$ to provide enough ``space'' for a terminating continuation $g$ of length $|g| \leq u$ to occur between the indices $B_{i}(\run_1 \ldots)$ and $B_{i-1}(\run_1 \ldots)$.
  Let $\run = \run_{1}\ldots \in \Runs$ such that $h = B_{i} (\run) < \infty$ for some $i\in\NN_{>0}$.
  Then, $(\Val_{\vecx}(\run_{1}\ldots\run_{h} \, g_{\run_{1}\ldots\run_{h}}))_{c} \leq 0$ for some $c\in\{c_{\nmax},c_{\pmax}\}$ and hence
  \begin{equation}
    \label{termination time claim}
    \LRVar_{\vecx} (\run') \leq h + |g_{\run_{1}\ldots\run_{h}}| \leq h + u \leq h + l
  \end{equation}
  for all runs $\run'\in\Pre_{\run_{1}\ldots\run_{h} \, g_{\run_{1}\ldots\run_{h}}}$.

  We have $\P(B_{i} = \infty) = 0$ for all $i\in\NN$ which allows us to ignore all runs $\run\in\Runs$ with $B_{i}(\run) = \infty$.
  A corresponding proof can be found at the end.
  Let $f_{B_{i}}: \Runs \to \Path$ denote the function such that for all $i\in\NN$ and all runs $\run =\run_1 \run_2 \ldots$ we have $f_{B_i}(\run) = \run_{1}\ldots\run_{B_{i}(\run)}$.
  Moreover, for all $f\in\Path$, let $\Pre_{f \, \overline{g_f}}$ denote the set $\Pre_{f}\setminus\Pre_{f \, g_f}$.
  We elaborate on the value of $\E (\LRVar_{\vecx})$.
  Let $\Pre_{f_{B_{i}} \, g_{f_{B_{i}}}} = \{ \run \in \Runs \mid \run \in \Pre_{f_{B_{i}}(\run) \, g_{f_{B_{i}}(\run)}} \}$.
  Given a (measurable) set $A \subseteq \Runs$, let $\ind_{A}\colon \Runs \to \{0,1\}$ denote its indicator function, i.e., the random variable that maps $\run\in\Runs$ to $1$ iff $\run \in A$.
  Then for any $\run \in \Runs$, we have $\ind_{\Pre_{f_{B_{i}} \, g_{f_{B_{i}}}}}(\run) = 1$ if $\run$ has the prefix $f_{B_{i}}(\run)\, g_{f_{B_{i}}(\run)}$ and otherwise, we have $\ind_{\Pre_{f_{B_{i}} \, g_{f_{B_{i}}}}}(\run) = 0$.
  So $\ind_{\Pre_{f_{B_{i}} \, g_{f_{B_{i}}}}}$ maps a run $\run = \run_1 \run_2 \ldots$ to 1 iff $f_{B_i}(\run)$ reaches the region specified by the constants $\eps$, $l$, and $r$, and afterwards, one has the path $g_{\run_1 \ldots \run_{B_i(\run)}}$ which leads to termination.
  Similarly, $\ind_{\Pre_{f_{B_{i}} \, g_{f_{B_{i}}}}} \cdot \prod_{i'=1}^{i-1} \ind_{\Pre_{f_{B_{i'}} \, \overline{g_{f_{B_{i'}}}}}}$ maps a run $\run$ to 1 iff it starts with $f_{B_{i}}(\run)\, g_{f_{B_{i}}(\run)}$, but not with $f_{B_{i'}}(\run)\, g_{f_{B_{i'}}(\run)}$ for any $i' < i$.
  Thus, after the prefix $\run_1 \ldots\run_{B_i(\run)}$, one has the path that leads to termination, but the earlier prefixes $\run_1 \ldots\run_{B_{i'}(\run)}$ were not followed by the corresponding paths $g_{\run_1 \ldots\run_{B_{i'}(\run)}}$.
  So if $\ind_{\Pre_{f_{B_{i}} \, g_{f_{B_{i}}}}} \cdot \prod_{i'=1}^{i-1} \ind_{\Pre_{f_{B_{i'}} \, \overline{g_{f_{B_{i'}}}}}}$ maps a run $\run$ to 1, then the number of steps until termination is at most $B_i(\run) + |g_{\run_1 \ldots\run_{B_{i}(\run)}}| \leq B_i(\run) + l$, see \cref{termination time claim}.
  \begin{align*}
    \E (\LRVar_{\vecx})
     & = \E (\min \{k\in\NN_{>0} \mid {(\Val_{\vecx}(\run_{1}\ldots\run_{k}))}_{c} \leq 0 \text{ for some } c\in[m]\}) \\
     & \leq \E (\min \{k\in\NN_{>0} \mid {(\Val_{\vecx}(\run_{1}\ldots\run_{k}))}_{c} \leq 0 \text{ for some } c\in\{c_{\nmax},c_{\pmax}\}\}) \\
     & \leq \E \left( \sum_{i=1}^{\infty} \left( (B_{i} + l) \cdot \ind_{\Pre_{f_{B_{i}} \, g_{f_{B_{i}}}}} \cdot\prod_{i'=1}^{i-1} \ind_{\Pre_{f_{B_{i'}} \, \overline{g_{f_{B_{i'}}}}}} \right) \right) \\
     & \leq \E \left( \sum_{i=1}^{\infty} \left( \left(l \cdot \ind_{\Pre_{f_{B_{i}} \, g_{f_{B_{i}}}}} + B_{i} - B_{i-1}\right) \cdot \prod_{i'=1}^{i-1} \ind_{\Pre_{f_{B_{i'}} \, \overline{g_{f_{B_{i'}}}}}} \right)\right) \\
     & = \sum_{i=1}^{\infty }\E \left( \left(l \cdot \ind_{\Pre_{f_{B_{i}} g_{f_{B_{i}}}}} \!\! + B_{i} - B_{i-1}\right) \cdot \prod_{i'=1}^{i-1}\! \ind_{\Pre_{f_{B_{i'}} \overline{g_{f_{B_{i'}}}}}} \right) \tag{Mon.\ Conv.}
  \end{align*}
  The last line holds because of monotone convergence, i.e., the expected value of the limit is the limit of the expected value (for monotonic sequences of random variables as above, since all addends are non-negative and hence the sum is monotonically increasing).
  To see why the penultimate line holds, consider some $f_{B_{i}}$ for $i\in\NN_{>0}$ with $f_{B_{i'} \, g_{B_{i'}}}$ not being a prefix of $f_{B_{i}}$ for all $i'\in\{1,\ldots,i-1\}$.
  The contribution of runs in $\Pre_{f_{B_{i}} \, g_{f_{B_{i}}}}$ to the expected value in the third line is then given by $B_{i} + l$ and their contribution in the penultimate line is given by $l + \sum_{i'=1}^{i} (B_{i'} - B_{i'-1}) = B_{i} + l$.
  So in the third line, such a run only contributes one addend (viz.\ $B_{i} + l$) to the expected value, whereas in the penultimate line, such a run contributes $i$ addends (i.e., $B_{i'} - B_{i'-1}$ for all $1 \leq i' \leq i$) plus the addend $l$ (once).
  We have
  \begin{equation}
    \E\left((B_{i} - B_{i-1}) \prod_{i'=1}^{i-1} \ind_{\Pre_{f_{B_{i'}} \overline{g_{f_{B_{i'}}}}}} \right) \leq {\left( 1 - {(\min\{p,1-p\})}^{l}\right)}^{i-1} (l+C)
    \label{eq:positive_eigenvalues_dominating_constraints1}
  \end{equation}
  for all $i\in\NN_{>0}$ and some constant $C>0$.
  A corresponding proof is given at the end.
  By combining both inequations, one obtains
  \begin{align*}
           & \E (\LRVar_{\vecx}) \\
    \leq{} & \sum_{i=1}^{\infty }\E \left( \left(l \cdot \ind_{\Pre_{f_{B_{i}} \, g_{f_{B_{i}}}}} \!\! + B_{i} - B_{i-1}\right) \prod_{i'=1}^{i-1}\! \ind_{\Pre_{f_{B_{i'}} \, \overline{g_{f_{B_{i'}}}}}} \right) \\
    ={}    & \sum_{i=1}^{\infty} \left( \E \left( l \cdot \ind_{\Pre_{f_{B_{i}}\, g_{f_{B_{i}}}}} \cdot \prod_{i'=1}^{i-1} \ind_{\Pre_{f_{B_{i'}} \, \overline{g_{f_{B_{i'}}}}}} \right) + \E \left( (B_{i} - B_{i-1}) \cdot \prod_{i'=1}^{i-1} \ind_{\Pre_{f_{B_{i'}}\, \overline{g_{f_{B_{i'}}}}}} \right) \right) \\
    ={}    & \sum_{i=1}^{\infty} \left( l \cdot \E \left(\ind_{\Pre_{f_{B_{i}}\, g_{f_{B_{i}}}}} \cdot \prod_{i'=1}^{i-1} \ind_{\Pre_{f_{B_{i'}} \, \overline{g_{f_{B_{i'}}}}}} \right) + \E \left( (B_{i} - B_{i-1}) \cdot \prod_{i'=1}^{i-1} \ind_{\Pre_{f_{B_{i'}}\, \overline{g_{f_{B_{i'}}}}}} \right) \right) \\
    \leq{} & \sum_{i=1}^{\infty} \left( l \cdot \E \left(\prod_{i'=1}^{i-1} \ind_{\Pre_{f_{B_{i'}} \, \overline{g_{f_{B_{i'}}}}}} \right) + \E \left( (B_{i} - B_{i-1}) \cdot \prod_{i'=1}^{i-1} \ind_{\Pre_{f_{B_{i'}}\, \overline{g_{f_{B_{i'}}}}}} \right) \right) \\
    ={}    & \sum_{i=1}^{\infty} \left( l \cdot \E \left(\ind_{\bigcap_{i'=1}^{i-1} \Pre_{f_{B_{i'}} \, \overline{g_{f_{B_{i'}}}}}} \right) + \E \left( (B_{i} - B_{i-1}) \cdot \prod_{i'=1}^{i-1} \ind_{\Pre_{f_{B_{i'}}\, \overline{g_{f_{B_{i'}}}}}} \right) \right) \\
    ={}    & \sum_{i=1}^{\infty} \left( l \cdot \P \left(\bigcap_{i'=1}^{i-1} \Pre_{f_{B_{i'}} \, \overline{g_{f_{B_{i'}}}}} \right) + {( 1- {(\min\{p,1-p\})}^{l})}^{i-1} \cdot \left( l + C \right) \right) \tag{by \cref{eq:positive_eigenvalues_dominating_constraints1}} \\
    \leq{} & (2l + C) \sum_{i=1}^{\infty} {\left( 1- {\left(\min\{p,1-p\}\right)}^{l}\right)}^{i-1} < \infty .
  \end{align*}
  For the last line, $\Pre_{f_{B_{i'}} \, \overline{g_{f_{B_{i'}}}}}$ is the set of all runs except those runs $\run$ where after the path $f_{B_{i'}(\run)}$ one has the path $g_{f_{B_{i'}(\run)}}$.
  Let $u' \leq u \leq l$ be the length of $g_{f_{B_{i'}(\run)}}$, i.e., it has the form $\run_1' \ldots \run_{u'}'$.
  The probability for choosing $\run_1'$ is at least $\min\{p,1-p\}$ and thus, the probability for choosing $g_{f_{B_{i'}(\run)}}$ is at least $(\min\{p,1-p\})^{u'} \geq (\min\{p,1-p\})^{u} \geq (\min\{p,1-p\})^{l}$.
  Hence, the probability for \emph{not} choosing $g_{f_{B_{i'}(\run)}}$ (i.e., the probability for a run in $\Pre_{f_{B_{i'}} \, \overline{g_{f_{B_{i'}}}}}$) is at most $1 - (\min\{p,1-p\})^{u'} \leq 1 - (\min\{p,1-p\})^{l}$.
  Thus, the probability for a run in $\bigcap_{i'=1}^{i-1} \Pre_{f_{B_{i'}} \, \overline{g_{f_{B_{i'}}}}}$ is at most $\left(1 - (\min\{p,1-p\})^{l}\right)^{i-1}$.

  Hence, we have $\vecx \not \in \NT$ which contradicts our assumption.

  \medskip

  \input{./dual_positive_eigenvalues_for_eventually_dominating_constraints_remaining_parts_app.tex}
\end{proof}

%% file: braverman_key_claims.tex
\begin{enumerate}
  \item There is some constant $C_{1} > 0$ such that
        \[
          \left| \sum_{j=0}^{N-1} \sum_{k=0}^{N-1} z_{j_{0}+j,k_{0}+k} \right| \leq N C_{1}
        \]
        for all $N\in\NN_{>0}$ and $j_{0},k_{0} \in \NN$.
        Note that in contrast to the corresponding key claim in~\cite{braverman06}, the upper bound $N C_{1}$ scales linearly with $N$.
        So although we add $N^2$ complex numbers on the left-hand side of the equation, the upper bound is linear in $N$.
  \item There is some constant $C_{2} > 0$ such that \[\sum_{j=0}^{l-1}\sum_{k=0}^{l-1}
          |z_{j_{0}+j,k_{0}+k}| \ge C_{2}\] for all $j_{0},k_{0} \in \NN$. Thus, this provides a lower bound for a sum of $N^2$ addends when choosing $N = l$.
\end{enumerate}

%% file: braverman_generalisation_claims_app.tex
We now prove the two key claims used in the proof:

\begin{enumerate}
  \item There is some constant $C_{1} > 0$ such that
        \[
          \left| \sum_{j=0}^{N-1} \sum_{k=0}^{N-1} z_{j_{0}+j,k_{0}+k} \right| \leq N C_{1}
        \]
        for all $N\in\NN_{>0}$ and $j_{0},k_{0} \in \NN$.
  \item There is some constant $C_{2} > 0$ such that \[\sum_{j=0}^{l-1}\sum_{k=0}^{l-1}
          |z_{j_{0}+j,k_{0}+k}| \ge C_{2}\] for all $j_{0},k_{0} \in \NN$.
\end{enumerate}

Instead of proving the two key claims directly, we prove two slightly more general claims.
To that end, let $\mathfrak{m}_{1},\ldots,\mathfrak{m}_{l} \in \RR_{>0}$.
Let the complex units $\zeta_{1,1},\ldots,\zeta_{1,l},\zeta_{2,1},\ldots,\zeta_{2,l}$ fulfill the prerequisites of \Cref{lem:braverman_generalisation}.
We assume $l>0$ as we are only interested in the lemma's second case and disregard the addends of $z_{j,k}$ where $\gamma_i = 0$.
We will then prove the following two claims for all $j,k \in \NN$:
\begin{enumerate}
  \item There is some constant $C_{1} > 0$ such that
        \[
          \left| \sum_{j=0}^{N-1}\sum_{k=0}^{N-1} \sum_{i=1}^{l} \gamma_{i} \zeta_{1,i}^{j} \zeta_{2,i}^{k} \right| \leq NC_{1}
        \]
        for all $N\in\NN_{>0}$ and $\gamma_{1},\ldots,\gamma_{l} \in \CC$ with $|\gamma_{i}| = \mathfrak{m}_{i}$ for all $i\in[l]$.

  \item There is some constant $C_{2}$ such that
        \[\sum_{j=0}^{l-1}\sum_{k=0}^{l-1} \left| \sum_{i=1}^{l} \gamma_{i} \zeta_{1,i}^{j} \zeta_{2,i}^{k} \right| \geq C_{2}\] for all $\gamma_{1},\ldots,\gamma_{l} \in \CC$ with $|\gamma_{i}| = \mathfrak{m}_{i}$ for all $i\in[n]$.
\end{enumerate}
Here, it can be seen that the more general claims imply the original claims from above by ``pushing'' the product $\zeta_{1,i}^{j_{0}} \zeta_{2,i}^{k_{0}}$ into the coefficients $\gamma_{i}$, i.e., by considering coefficients $\gamma'_{i} = \gamma_{i}\zeta_{1,i}^{j_{0}}\zeta_{2,i}^{k_{0}}$ instead.
As $\zeta_{1,i},\zeta_{2,i}$ are complex units, the modulus (i.e., absolute value) $\mathfrak{m}_{i}$ of the coefficient $\gamma_{i}$ is retained, i.e., $|\gamma'_i| = |\gamma_i| = \mathfrak{m}_{i}$.

\bigskip

\textbf{Proof of the First Claim}: We elaborate on the term $\left| \sum_{j=0}^{N-1}\sum_{k=0}^{N-1}\sum_{i=1}^{l} \gamma_{i} \zeta_{1,i}^{j} \zeta_{2,i}^{k}
  \right|$:
\begin{align*}
         & \left| \sum_{j=0}^{N-1}\sum_{k=0}^{N-1}\sum_{i=1}^{l} \gamma_{i} \zeta_{1,i}^{j} \zeta_{2,i}^{k} \right| \\
  \leq{} & \sum_{i=1}^{l} \left| \sum_{j=0}^{N-1}\sum_{k=0}^{N-1} \gamma_{i} \zeta_{1,i}^{j} \zeta_{2,i}^{k} \right| \tag{Triang.\ Ineq.} \\
  ={}    & \sum_{\substack{i\in[l] \\\zeta_{1,i} = 1}} \left| \sum_{j=0}^{N-1}\sum_{k=0}^{N-1} \gamma_{i} \zeta_{2,i}^{k} \right| + \sum_{\substack{i\in[l]\\\zeta_{2,i} = 1}} \left| \sum_{j=0}^{N-1}\sum_{k=0}^{N-1} \gamma_{i} \zeta_{1,i}^{j} \right| + \sum_{\substack{i\in[l] \\\zeta_{1,i}\neq1\\\zeta_{2,i}\neq1}} \left| \sum_{j=0}^{N-1}\sum_{k=0}^{N-1} \gamma_{i}\zeta_{1,i}^{j}\zeta_{2,i}^{k} \right| \\
  ={}    & N \left(\sum_{\substack{i\in[l] \\\zeta_{1,i} = 1}} \left| \sum_{k=0}^{N-1} \gamma_{i} \zeta_{2,i}^{k} \right| + \sum_{\substack{i\in[l] \\\zeta_{2,i} = 1}} \left| \sum_{j=0}^{N-1} \gamma_{i} \zeta_{1,i}^{j} \right|\right) + \sum_{\substack{i\in[l] \\\zeta_{1,i}\neq1\\\zeta_{2,i}\neq1}} \left| \sum_{j=0}^{N-1}\sum_{k=0}^{N-1} \gamma_{i}\zeta_{1,i}^{j}\zeta_{2,i}^{k} \right| \\
  ={}    & N \left(\sum_{\substack{i\in[l] \\\zeta_{1,i} = 1}} \left| \gamma_{i} \frac{\zeta_{2,i}^{N} - 1}{\zeta_{2,i} - 1} \right| + \sum_{\substack{i\in[l]\\\zeta_{2,i} = 1}} \left| \gamma_{i} \frac{\zeta_{1,i}^{N} - 1}{\zeta_{1,i} - 1} \right|\right) + \sum_{\substack{i\in[l] \\\zeta_{1,i}\neq1\\\zeta_{2,i}\neq1}} \left| \gamma_{i} \frac{\zeta_{1,i}^{N} - 1}{\zeta_{1,i} - 1} \frac{\zeta_{2,i}^{N} - 1}{\zeta_{2,i} - 1} \right| \tag{Geometr.\ Series} \\
  \leq{} & N \left( \sum_{\substack{i\in[l] \\\zeta_{1,i}=1}} \frac{2\mathfrak{m}_{i}}{|\zeta_{2,i} - 1|} + \sum_{\substack{i\in[l]\\\zeta_{2,i} = 1}} \frac{2\mathfrak{m}_{i}}{|\zeta_{1,i}-1|} \right) + \sum_{\substack{i\in[l] \\\zeta_{1,i}\neq1\\\zeta_{2,i}\neq1}} \frac{4\mathfrak{m}_{i}}{|\zeta_{1,i} - 1| |\zeta_{2,i} - 1|} \tag{Triang. Ineq.}
\end{align*}
For the penultimate step, recall that for all $i\in[l]$ we have $\zeta_{1,i} \neq 1$ or $\zeta_{2,i} \neq 1$ by the prerequisites of the lemma.
For the last step, we have $|\gamma_i \cdot (\zeta_{2,i}^{N} - 1)| = |\gamma_i| \cdot |\zeta_{2,i}^{N} + (- 1)| \leq \mathfrak{m}_{i} \cdot (|\zeta_{2,i}^{N}| + |- 1|) = 2\mathfrak{m}_{i}$, etc.

By choosing $C_{1,1},C_{1,2},C_{1,3} \in \RR_{>0}$ large enough, one obtains
\[\left| \sum_{j=0}^{N-1}\sum_{k=0}^{N-1} \sum_{i=1}^{l} \gamma_{i}\zeta_{1,i}^{j}\zeta_{2,i}^{k} \right| \leq N (C_{1,1} + C_{1,2}) + C_{1,3}.\]
We choose $C_{1} \geq C_{1,1} + C_{1,2} + C_{1,3}$.
Then, we have
\[\left| \sum_{j=0}^{N-1}\sum_{k=0}^{N-1} \sum_{i=1}^{l}
  \gamma_{i}\zeta_{1,i}^{j}\zeta_{2,i}^{k} \right| \leq N C_{1}\] and the claim follows.

\bigskip

\textbf{Proof of the Second Claim}: Let $I_{1,1},\ldots,I_{1,d_{1}}$ be indices such that $\zeta_{1,I_{1,1}}, \ldots,\zeta_{1,I_{1,d_{1}}}$ are pairwise different complex numbers with $\{\zeta_{1,I_{1,1}},\ldots,\zeta_{1,I_{1,d_{1}}}\} = \{\zeta_{1,1},\ldots,\zeta_{1,l}\}$ and analogously let $I_{2,1},\ldots,I_{2,d_{2}}$ be indices such that $\zeta_{2,I_{2,1}},\ldots, \zeta_{2,I_{2,d_{2}}}$ are pairwise different with $\{\zeta_{2,I_{2,1}},\ldots,\zeta_{2,I_{2,d_{2}}}\} = \{\zeta_{2,1},\ldots,\zeta_{2,l}\}$.
According to the prerequisites of \Cref{lem:braverman_generalisation}, for every $i \in [l]$, there are unique $a\in[d_{1}]$ and $b\in[d_{2}]$ such that $(\zeta_{1,i},\zeta_{2,i}) = (\zeta_{1,I_{1,a}},\zeta_{2,I_{2,b}})$.

Consider the two matrices
\begin{align*}
  \mat{M_{1}} & = {\left(\Vandermonde \left(\zeta_{1,I_{1,1}},\ldots,\zeta_{1,I_{1,d_{1}}}\right)\right)}^{\T} \in \CC^{d_{1} \times d_{1}} \\
  \mat{M_{2}} & = {\left(\Vandermonde\left(\zeta_{2,I_{2,1}},\ldots, \zeta_{2,I_{2,d_{2}}}\right)\right)}^{\T} \in \CC^{d_{2} \times d_{2}},
\end{align*}
where $\Vandermonde$ denotes the Vandermonde matrix, and their Kronecker product $\mat{M_{1}} \otimes \mat{M_{2}} \in \CC^{d_{1}d_{2} \times d_{1}d_{2}}$ given by
\[
  \mat{M_{1}} \otimes \mat{M_{2}} = \begin{pmatrix}
    \zeta_{1,I_{1,1}}^{0} \cdot \mat{M_{2}}       & \cdots & \zeta_{1,I_{1,d_{1}}}^{0} \cdot \mat{M_{2}}       \\
    \vdots                                        & \ddots & \vdots                                            \\
    \zeta_{1,I_{1,1}}^{d_{1}-1} \cdot \mat{M_{2}} & \cdots & \zeta_{1,I_{1,d_{1}}}^{d_{1}-1} \cdot \mat{M_{2}} \\
  \end{pmatrix}
\]
with
\[
  \mat{M_{2}} =
  \begin{pmatrix}
    \zeta_{2,I_{2,1}}^{0}       & \cdots & \zeta_{2,I_{2,d_{2}}}^{0}       \\
    \vdots                      & \ddots & \vdots                          \\
    \zeta_{2,I_{2,1}}^{d_{2}-1} & \cdots & \zeta_{2,I_{2,d_{2}}}^{d_{2}-1}
  \end{pmatrix} .
\]
As a Kronecker product, the rank of $\mat{M_{1}} \otimes \mat{M_{2}}$ is given by the product of the ranks of $\mat{M_{1}}$ and $\mat{M_{2}}$.
Since both, $\mat{M_{1}}$ and $\mat{M_{2}}$, have full rank as transposed Vandermonde matrices of distinct values (see~\cite[Section 0.9.11]{hornmatrixanalysis}), $\mat{M_{1}} \otimes \mat{M_{2}}$ has full rank as well~\cite[Cor.\ 4.2.11]{horntopicsinmatrixanalysis}, i.e., its rank is $d_1 \cdot d_2$.

Let $\gamma_{1},\ldots,\gamma_{l} \in \CC$ be coefficients as required by the claim, i.e., $|\gamma_{i}| = \mathfrak{m}_{i}$ for all $i\in[l]$.
From this, let $\vec{\beta} =
  \begin{pmatrix}
    \beta_{1,1} & \cdots & \beta_{1,d_{2}} & \cdots & \beta_{d_{1},1} & \cdots & \beta_{d_{1},d_{2}}
  \end{pmatrix}
  ^{\T}$ be defined such that $\beta_{a,b} = \gamma_{i}$ whenever $\zeta_{1,I_{1,a}} = \zeta_{1,i}$ and $\zeta_{2,I_{2,b}} = \zeta_{2,i}$ for some (unique) $i\in[l]$ for $(a,b)\in[d_{1}]\times[d_{2}]$. Otherwise, $\beta_{a,b} = 0$.

Let $u \in \{1,\ldots,d_1 \cdot d_2\}$.
Then the $u$-th entry of $(\mat{M_{1}}\otimes\mat{M_{2}})\vec{\beta}$ has the form
\begin{align*}
   & \zeta_{1,I_{1,1}}^j \cdot \zeta_{2,I_{2,1}}^k \cdot \beta_{1,1} + \ldots + \zeta_{1,I_{1,1}}^j \cdot \zeta_{2,I_{2,d_2}}^k \cdot \beta_{1,d_2} + \ldots \\ &{}+
  \zeta_{1,I_{1,d_1}}^j \cdot \zeta_{2,I_{2,1}}^k \cdot \beta_{d_1,1} + \ldots +
  \zeta_{1,I_{1,d_1}}^j \cdot \zeta_{2,I_{2,d_2}}^k \cdot \beta_{d_1,d_2}
\end{align*}
For every $i\in [l]$, there are unique $a,b$ with $\beta_{a,b} = \gamma_i$ and where the corresponding addend has the form $\zeta_{1,I_{1,a}}^j \cdot \zeta_{2,I_{2,b}}^k \cdot \beta_{a,b} =\zeta_{1,I_{1,i}}^j \cdot \zeta_{2,I_{2,i}}^k \cdot \gamma_i$.
Thus, $\sum_{i=1}^{l} \gamma_{i}\zeta_{1,i}^{j}\zeta_{2,i}^{k} = {((\mat{M_{1}}\otimes\mat{M_{2}})\vec{\beta})}_{u}$ for all $(j,k)\in\{0,\ldots,d_{1}-1\}\times\{0,\ldots, d_{2} - 1\}$ whenever the $u$-th row of $\mat{M_{1}}\otimes\mat{M_{2}}$ contains terms $\zeta_{1,I_{1,i}}^{j}\zeta_{2,I_{2,i'}}^{k}$ for some $i \in \{1,\ldots,d_{1}\}$ and $i' \in \{1,\ldots,d_{2}\}$.
As $l>0$ and $\mathfrak{m}_{1},\ldots,\mathfrak{m}_{l} \in \RR_{>0}$ we have\footnote{Here we need that the tuples $(\zeta_{1,i},\zeta_{2,i})$ for $i\in\{1,\ldots,l\}$ are pairwise different.
Otherwise, $\beta_{a,b}$ would have to be the sum of all $\gamma_i$ where $\zeta_{1,I_{1,a}} = \zeta_{1,i}$ and $\zeta_{2,I_{2,b}} = \zeta_{2,i}$.
But this means that $\beta_{a,b}$ could be $0$ and if that happens for all $(a,b)\in[d_{1}]\times[d_{2}]$, then we would have $\vec{\beta} = \veczero$.}
$\vec{\beta} \neq \veczero$ and hence by the full rank of $\mat{M_{1}}\otimes\mat{M_{2}}$, we additionally have $\left(\mat{M_{1}}\otimes\mat{M_{2}}\right)\vec{\beta} \neq \vec{0}$.
This in turn implies $\sum_{i=1}^{l} \gamma_{i}\zeta_{1,i}^{j}\zeta_{2,i}^{k} \neq 0$ for some $(j,k)\in \{0,\ldots,d_{1}-1\}\times\{0,\ldots,d_{2} - 1\}$.

Moreover, there is a lower bound $C_{2,1} > 0$ such that for all $\vec{\beta}$ defined as above (i.e., for any $\gamma_{1},\ldots,\gamma_{l}$ with $|\gamma_{i}| = \mathfrak{m}_{i}$) we have $\left\lVert (\mat{M_{1}}\otimes\mat{M_{2}}) \vec{\beta}\right\rVert^{2} > C_{2,1}$.
Here, as usual we define the norm of a vector
$(x_{1} \ldots x_{k})^{\T}$ as $\left\lVert(x_{1} \ldots x_{k})^{\T}\right\rVert = \sqrt{x_1^2 + \ldots + x_k^2}$.
To see that, assume w.l.o.g.\ $I_{1,1} = I_{2,1} = 1$.
Then, fix $\beta_{1,1} = |\gamma_{1}| = \mathfrak{m}_{1}$.
As the concrete value of $\gamma_{1}$ is unknown to us (we only have $|\gamma_{1}| = \mathfrak{m}_{1}$) we set $\beta_{1,1}$ to its modulus instead.
This, however, does not invalidate the subsequent line of argument, as multiplications with complex units (such as $\frac{\gamma_{1}}{|\gamma_{1}|}$) do not change the norm of the vector $\left\lVert \left(\mat{M_{1}}\otimes\mat{M_{2}}\right)\vec{\beta}\right\rVert$.
Recall $(\mat{M_{1}}\otimes\mat{M_{2}})\vec{\beta} \neq \vec{0}$ for all $\vec{\beta} =
  \begin{pmatrix}
    \beta_{1,1} & \cdots
  \end{pmatrix}
$. If $d_{1}d_{2} = 1$ set $C_{2,1} = \left\lVert (\mat{M_{1}\otimes\mat{M_{2}}})\vec{\beta} \right\rVert = | \beta_{1,1}| = \mathfrak{m}_{1} > 0$. Otherwise, define $\vec{\beta'}$ as $\vec{\beta}$ with the first entry removed and $\mat{M'}$ as $\mat{M_{1}}\otimes\mat{M_{2}}$ with the first column removed. Furthermore, set $\vec{\delta} = (\mat{M_{1}}\otimes\mat{M_{2}})
  \begin{pmatrix}
    \beta_{1,1} & 0 & \cdots & 0
  \end{pmatrix}
  ^{\T}$. Then, $(\mat{M_{1}}\otimes\mat{M_{2}})\vec{\beta} = \mat{M'}\vec{\beta'} + \vec{\delta} \neq \veczero$. As we are dealing with the finite dimensional Hilbert space $\CC^{d_{1}d_{2} - 1}$ there is a~\cite[Thm.\ 3.6.4]{hilbertspaces} (unique) best approximation $\vec{\beta^{*}}$ minimizing $\operatorname{dist}(\vec{\beta'}) = \left\lVert \mat{M'}\vec{\beta'} + \vec{\delta} \right\rVert^{2}$ for all $\vec{\beta'} \in \CC^{d_{1}d_{2} - 1}$. We define $C_{2,1} = \operatorname{dist}(\vec{\beta^{*}}) > 0$.

Assume that $|((\mat{M_1} \otimes \mat{M_2}) \vec{\beta})_u| < \sqrt{\frac{C_{2,1}}{d_{1}d_{2}}}$ for all $u \in \{1, \ldots, d_1\cdot d_2 \}$.
This would imply $(((\mat{M_1} \otimes \mat{M_2}) \vec{\beta})_u)^2 < \frac{C_{2,1}}{d_{1}d_{2}}$ for all $u$ and hence, $\sum_{u \in\{1, \ldots, d_1\cdot d_2 \}} (((\mat{M_1} \otimes \mat{M_2}) \vec{\beta})_u)^2 < C_{2,1}$, i.e., $\left\lVert (\mat{M_{1}}\otimes\mat{M_{2}}) \vec{\beta}\right\rVert^{2} < C_{2,1}$ in contradiction to our observation above.
Hence, for arbitrary $\gamma_{1},\ldots,\gamma_{l}$ with $|\gamma_{1}| = \mathfrak{m}_{1} > 0$ we have $\abs{\sum_{i=1}^{l} \gamma_{i} \zeta_{1,i}^{j} \zeta_{2,i}^{k}} \geq \sqrt{\frac{C_{2,1}}{d_{1}d_{2}}} > 0$ for some $(j,k) \in \{0,\dots,d_{1}-1\}\otimes\{0,\dots,d_{2}-1\} \subseteq {\{0,\dots,l-1\}}^{2}$ implying the claim of the second item by choosing $C_{2} > 0$ correspondingly.

%% file: dual_positive_eigenvalues_for_eventually_dominating_constraints_remaining_parts_app.tex
We now prove the missing parts in the proof of \Cref{lem:dual_positive_eigenvalues_for_eventually_dominating_constraints}, i.e., we show $\P(B_{i} = \infty) = 0$ for all $i\in\NN$ and we prove the inequation \cref{eq:positive_eigenvalues_dominating_constraints1}.

Assume the existence of constants $l\in\NN_{>0}$, $r\in\NN$, and $\eps \in \AA_{>0}$ such that $\eps \cdot l \geq r + 1$.
The intuition is that we consider a prefix of a run $\run$ where we have first performed $B_{i}(\run) + l$ steps (i.e., $\URVar(f) \cdot |f|$ was in a safe region for at least the $i$-th time and then we performed $l$ more steps in an additional path $f'$).
These $l$ additional steps would suffice for the sub-path $g_f$, since its length is at most $l$.
We now want to approximate how many steps one needs after the path $f f'$ in order to reach a safe region again.

\begin{figure}
  \begin{center}
    \input{illustration_full_middle_with_ABC.tex}
  \end{center}
  \caption{\label{fig:safe_with_ABC} Variant of \Cref{fig:safe} Indicating the Considered Safe Regions (A), (B), and (C)}
\end{figure}
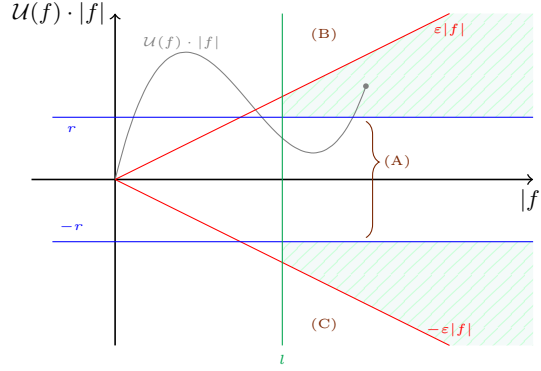

Before proving the parts that were missing in the proof of \Cref{lem:dual_positive_eigenvalues_for_eventually_dominating_constraints}, we define three ($\F$-measurable) random variables $T_{\mathrm{in}}, T_{\mathrm{out},1}, T_{\mathrm{out},2}\colon \Runs \to \NN \cup \{\infty\}$.
We need the random variable $T_{\mathrm{in}}$ for the case where $\URVar(f f') \cdot |f f'|$ is in the (unsafe) region $(A)$ of \Cref{fig:safe_with_ABC} (i.e., where $|\URVar(f f') \cdot |f f'|| < r$).
If one continues with a path $f''$ where $|\URVar(f'') \cdot |f''|| \geq 2r$, then during the execution of the path $f''$ one must reach a safe region.
The reason is that $\URVar(f f' f'') \cdot |f f' f''|$ cannot ``jump'' over the safe region, since when $|f|$ increases by $1$, then $\URVar(f) \cdot |f|$ changes by at most 1.
Note that we made sure that the ``height'' of the safe region is at least 1 by requiring $\eps \cdot l \geq r+ 1$.
We will show that in expectation, one only needs a path $f''$ of constantly many steps (bounded by some $C_{\mathrm{in}}$) until we have $|\URVar(f'') \cdot |f''|| \geq 2r$ and one thus reaches the safe region again.

The random variable $T_{\mathrm{out},1}$ is needed if $\URVar(ff') \cdot |ff'|$ is in the (unsafe) region $(B)$ of \Cref{fig:safe_with_ABC} (i.e., where $\URVar(ff') \cdot |ff'| > \eps \cdot |ff'|$).
(Analogously, $T_{\mathrm{out},2}$ is needed if $\URVar(ff') \cdot |ff'|$ is in the (unsafe) region $(C)$ of \Cref{fig:safe_with_ABC}.) Since we performed $B_i(\run) + l$ steps, $\URVar(ff') \cdot |ff'|$ is at most $l$ above the border $\eps \cdot |f|$ of the upper safe region.
If one continues with a path $f''$ where $\URVar(f'') \cdot |f''| \leq \eps \cdot |f''| -l$, then ``$-l$'' ensures that one ``undoes'' the previous increase by at most $l$ and then one is below the (red) border of the upper safe region in \Cref{fig:safe_with_ABC}.
We will show that in expectation, one only needs a path $f''$ of constantly many steps (bounded by some $C_{\mathrm{out}}$) until one has $\URVar(f'') \cdot |f''| \leq \eps \cdot |f''|-l$ and one thus reaches the safe region again.
To prove this, we will use the idea of a ranking supermartingale (see, e.g., \cite[Def.\ 5.4]{fioriti15}) and the initial value of the ranking supermartingale will yield the constant upper bound $C_{\mathrm{out}}$.
In the following, recall again that $\URVar(f) \cdot |f| = |f|_{\symMat{A}} - p \cdot |f|$.
\begin{align*}
  T_{\mathrm{in}}    & = \min \{ j \in \NN \mid |\URVar (\run_{1}\ldots\run_{j}) \cdot j| \geq 2r \} \\
                     & = \min \{j \in \NN \mid \left||\run_{1}\run_{2}\ldots\run_{j}|_{\symMat{A}} - jp\right| \geq 2r \} \\
  T_{\mathrm{out},1} & = \min \{ j \in \NN \mid \URVar (\run_{1}\ldots\run_{j}) \cdot j \leq \eps j - l\} \\
                     & = \min \{ j \in \NN \mid |\run_{1}\ldots\run_{j}|_{\symMat{A}} - jp + l - \eps j \leq 0\} \\
  T_{\mathrm{out},2} & = \min \{j \in \NN \mid \URVar (\run_{1}\ldots\run_{j}) \cdot j \geq l - \eps j\} \\
                     & = \min \{ j \in \NN \mid - |\run_{1}\ldots\run_{j}|_{\symMat{A}} + jp + l - \eps j \leq 0\}
\end{align*}
where we again use the convention $\min \emptyset = \infty$.

\bigskip

\textbf{On $\mathbf{\E (T_{\mathrm{in}})}$}: Recall $p\in(0,1)$ and set $\delta = \left\lceil\tfrac{4r}{1-p}\right\rceil$.
We have
\begin{equation}
  \E (T_{\mathrm{in}}) \leq \E \left(\min \{ j\in\NN \mid j>\delta, \, \run_{j - \delta + 1} \run_{j - \delta + 2} \ldots \run_{j} = \symMat{A}^{\delta} \}\right).
  \label{TinAlternative}
\end{equation}
Here, $\min \{ j\in\NN \mid j>\delta, \, \run_{j - \delta + 1} \run_{j - \delta + 2}
  \ldots \run_{j} = \symMat{A}^{\delta} \}$ maps every run to the smallest number $j$ such that the $\delta$ steps before the $j$-th step were all $\symMat{A}$s. To show that the inequation \cref{TinAlternative} holds, it suffices to prove that for every run $\run = \run_{1}\ldots\run_{h} \run_{h+1} \ldots \run_{h+\delta}$ with $h\in\NN$ and $\run_{h+1} = \ldots = \run_{h+\delta} = \symMat{A}$ (i.e., where the last $\delta$ steps were all $\symMat{A}$s), there is some $j\in\{1,\dots,h+\delta\}$ with $||\run_{1}\ldots\run_{j}|_{\symMat{A}} - jp| \geq 2r$ (i.e., $T_{\mathrm{in}}$ yields at most the index of the last of these $\delta$ steps). To prove this sufficient claim, we show that in fact there is a $j\in\{h,h+\delta\}$ with $||\run_{1}\ldots\run_{j}|_{\symMat{A}} - jp| \geq 2r$. The reason is that if the inequation does not hold for $j = h$, then we have $||\run_{1}\ldots\run_{h}|_{\symMat{A}} - hp| < 2r$. This implies $- 2 r < |\run_{1}\ldots\run_{h}|_{\symMat{A}} - hp$. Then one concludes $|\run_{1}\ldots\run_{h}\run_{h+1}\ldots \run_{h+\delta}|_{\symMat{A}} - p (h+\delta) = |\run_{1}\ldots\run_{h}|_{\symMat{A}} - h p + (1-p) \delta > -2r + 4r = 2r$.

Next, consider the following over-approximation.
In the second line, we divide the run into sub-paths of length $\delta$ and $j$ indicates when the first sub-path only consists of $\symMat{A}$s.
\begin{align*}
  \E (T_{\mathrm{in}})
   & \leq \E \left(\min \{ j\in\NN \mid j>\delta, \, \run_{j - \delta + 1} \run_{j - \delta + 2} \ldots \run_{j} = \symMat{A}^{\delta} \}\right) \tag{by \cref{TinAlternative}} \\
   & \leq \E \left(\delta\min \{ j\in\NN_{>0} \mid \run_{(j-1)\delta+1} \run_{(j-1)\delta + 2} \ldots \run_{j\delta} = \symMat{A}^{\delta} \}\right) \\
   & = \E \left(\delta \sum_{j=1}^{\infty} j\ind_{M_{i}}\right),
\end{align*}
where we define the corresponding $\F$-measurable (pairwise disjoint) sets
\begin{equation*}
  M_{j} = \{\run_{1}\ldots \mid \run_{(j-1)\delta+1}\ldots\run_{j\delta} = \symMat{A}^{\delta},\, \run_{(j'-1)\delta+1}\ldots\run_{j'\delta} \neq \symMat{A}^{\delta} \text{ for all } j'\in[j-1]\}
\end{equation*}
of runs for all $j\in\NN_{>0}$.
Here, $M_j$ are all runs where the $j$-th sub-path of length $\delta$ only consists of $\symMat{A}$s, whereas all previous sub-paths of length $\delta$ also contained $\symMat{B}$s.
To obtain an upper bound for the probability of runs in $M_j$, note that $p^{\delta}$ is the probability for $\delta$ many subsequent $\symMat{A}$s.
Thus, $1-p^{\delta}$ is the probability \emph{not} to have $\delta$ many subsequent $\symMat{A}$s.
Hence, an upper bound for the probability of runs in $M_j$ is the probability that $j-1$ times one does not have $\delta$ many subsequent $\symMat{A}$s, i.e., $\P (M_{j}) < {(1-p^{\delta})}^{j-1}$.
Moreover, $\P (\bigcup_{j=1}^{\infty} M_{j}) = 1$, because the probability of \emph{never} having a sub-path of length $\delta$ consisting only of $\symMat{A}$s is $0$.
Hence,
\begin{align*}
  \E (T_{\mathrm{in}}) \leq \E \left(\delta \sum_{j=1}^{\infty} j\ind_{M_{j}}\right) & = \delta \sum_{j=1}^{\infty} j \E (\ind_{M_{j}}) = \delta \sum_{j=1}^{\infty} j \P (M_{j}) \\
                                                                                     & < \delta \sum_{j=1}^{\infty} j {\left(1 - p^{\delta}\right)}^{j-1} < C_{\mathrm{in}}
\end{align*}
for some large enough constant $C_{\mathrm{in}} > 0$ as $|(1 - p^{\delta})| < 1$ (i.e., the geometric series converges), where we used monotone convergence to obtain the first equation.

\bigskip

\textbf{On $\mathbf{\E (T_{\mathrm{out}})}$}: We define (discrete) random variables $R_{h,j}\colon \Runs \rightarrow \AA_{\geq 0}$ for all $h\in\{1,2\},j\in\NN_{>0}$ which express how far we are away from the (red) border of \Cref{fig:safe_with_ABC} in the $j$-th step.
\[
  R_{h,j} (\run_{1}\ldots) =
  \begin{cases}
    v + 1+ \eps, & \text{if } v = {(-1)}^{h+1}|\run_{1}\ldots\run_{j}|_{\symMat{A}} + {(-1)}^{h}jp + l - \eps j > 0 \\
    0,           & \text{otherwise}
  \end{cases}
\]
We will use these random variables as ranking supermartingales and therefore, we have to ensure that all random variables $R_{h,j}$ are non-negative.\footnote{We use $v + 1+ \eps$ instead of just $v$ to ensure the ``decrease'', i.e., $\E (R_{h,j+1} - R_{h,j} \mid \F_{j}) \leq R_{h,j} - 1$.
  To provide some intuition: Assume that at point $j$ we are ``very close'' to the barrier, i.e., $v$ is almost $0$.
  If in the next step we surpass the barrier, then by definition $R_{h,j+1} = 0$ by the second case.
  But this means that the value of $v$ has not decreased by $\eps$ in expectation (as we performed the move away from the barrier to $v' = v' + (1-p) - \eps$ undisturbed).
  So we have to add $1 + \eps$ to the value of $R_{h,j}$ to decrease ``enough'' in this final approach of the barrier (by letting $R_{h,j}$ fall from close to $1 + \eps$ to $0$ in the surpassing step).
}
Consider the two stochastic processes ${(R_{h,j})}_{j=1}^{\infty}$ for $h\in\{1,2\}$ w.r.t.\ the filtration $\F_{1} \subseteq \F_{2} \subseteq \ldots$ where $\F_{j} = \sigma\left( \left\{\Pre_{f} \mid f \in \Path,\, |f| \leq j\right\} \right)$ for all $j \in \NN_{>0}$.
Clearly, the two stochastic processes ${(R_{h,j})}_{j=1}^{\infty}$ for $h\in\{1,2\}$ are adapted to the filtration $\F_{1} \subseteq \F_{2} \subseteq \ldots$ (see \cref{def:filtration_and_adaptedness}).
By definition, $T_{\mathrm{out},h} (\run) = \min \{j\in\NN_{>0} \mid R_{h,j}(\run) = 0\}$ for $h\in\{1,2\}$.
Hence, $T_{\mathrm{out},h}(\run) = \min \{j \in \NN_{>0} \mid R_{h,j}(\run)=0\}$ is a stopping time for ${(R_{h,j})}_{j=1}^{\infty}$ with $h\in\{1,2\}$.

Hence, ${(R_{h,j})}_{j=1}^{\infty}$ with $h\in\{1,2\}$ are ranking-su\-per\-mar\-tin\-gales, see \cite[Def.\ 5.4]{fioriti15}.
This means on the one hand that $\E (|R_{h,j}|)$ is finite.
The reason is that for every run we have $v = {(-1)}^{h+1}|\run_{1}\ldots\run_{j}|_{\symMat{A}} +
  {(-1)}^{h}jp + l - \eps j \leq j + l - \eps j$ and thus $R_{h,j} (\run_{1}\ldots) \leq j + l - \eps j + 1 + \eps$ which implies $\E (|R_{h,j}|) \leq j + l - \eps j + 1 + \eps < \infty$ for all $(h,j)\in\{1,2\}\times\NN_{>0}$. On the other hand, this means that in expectation the distance to the (red) border of \Cref{fig:safe_with_ABC} decreases by at least the drift $\eps$, i.e., for all $j\in\NN_{>0}$ and $h\in\{1,2\}$ we have
\begin{align*}
  \E \left(R_{h,j+1} - R_{h,j}\mid \F_{j}\right)(\run) & \leq \E ( {(-1)}^{h+1}|\run_{j+1}|_{\symMat{A}} \mid \F_{j})(\run) + {(-1)}^{h}p - \eps \\
                                                       & = {(-1)}^{h+1} \E(|\run_{j+1}|_{\symMat{A}} \mid \F_{j})(\run) + {(-1)}^{h}p - \eps \\
                                                       & \leq -\eps
\end{align*}
for (almost) all $\run \in \Runs$ with $T_{\mathrm{out},h}(\run) > j$, i.e., runs that did not yet reach the safe region.
This is due to \cref{lem:atoms_and_conditional_expectations}, as for every $j\in\NN_{>0}$ and every run $\run = \run_{1}\run_{2}\ldots$, the set $\Pre_{\run_{1}\run_{2}\ldots\run_{j}}$ is an atom of the $\sigma$-field $\F_{j}$ and hence $\E(|\run_{j+1}|_{\symMat{A}} \mid \F_{j})(\run) \cdot \P(\Pre_{\run_{1}\run_{2}\ldots\run_{j}})= \E (|\run_{j+1}|_{\symMat{A}} \cdot \ind_{\Pre_{\run_{1}\run_{2}\ldots\run_{j}}}) = p \cdot \P(\Pre_{\run_{1}\run_{2}\ldots\run_{j}}) > 0$.
Therefore, $\E(|\run_{j+1}|_{\symMat{A}} \mid \F_{j})(\run) = p$.
Here, we used that for $h = 1$, $R_{h,j} (\run_{1}\ldots) > 0$ implies $R_{h,j+1} (\run_{1}\ldots) = R_{h,j} (\run_{1}\ldots) + 1 - p - \eps$ for $\run_{j+1} = \matA$ ($\matA$ is selected with probability $p$) and $R_{h,j+1} (\run_{1}\ldots) \leq R_{h,j} (\run_{1}\ldots) - p - \eps$ if $\run_{j+1} = \matB$ (selected with probability $1-p$) for all runs $\run = \run_{1}\ldots$.
In the second case, the inequation is necessary as for the value of $R_{h,j+1}$ the second case in its definition might apply (i.e., $v \leq 0$).
Likewise, if $h = 2$, then $R_{h,j+1} (\run_{1}\ldots) \leq R_{h,j} (\run_{1}\ldots) - 1 + p - \eps$ for $\run_{j+1} = \matA$ (probability $p$) and $R_{h,j+1} (\run_{1}\ldots) = R_{h,j} (\run_{1}\ldots) + p - \eps$ if $\run_{j+1} = \matB$ (probability $1-p$).
Therefore, the expected value of the initial value of the ranking supermartingale divided by the drift is an upper bound for the expected value of the stopping time of the stochastic process ${(R_{h,j})}_{j=1}^{\infty}$.
In other words, \cite[Lemma 5.5]{fioriti15} implies $\E (T_{\mathrm{out},h}) \leq \frac{\E (R_{h,1})}{\eps} < \frac{2+l}{\eps} = C_{\mathrm{out}}$ for both $h\in\{1,2\}$.
Here, note that $R_{h,1} (\run_{1}\ldots) \leq 1 + l - \eps + 1 + \eps = 2 + l$.

\medskip

Recall the $\F$-measurable random variables $B_{i}\colon \Runs \to \NN\cup\{\infty\}$ defined as \input{./dual_positive_eigenvalues_for_eventually_dominating_constraints_def_Xi.tex} for all $i\in\NN$ with the convention $\min \emptyset = \infty$.

For all $i\in\NN$, we will now show that $P(B_{i} = \infty) = 0$ by induction on $i$.
In the induction base, we have $i = 0$.
By definition of $B_0$, we have $\P(B_{0} = \infty) = \P (B_{0} > 0) = 0$.

In the induction step, let $i\in\NN_{>0}$.
For all $i' \in \{1, \ldots, i\}$, by the induction hypothesis we have $\P (B_{i'-1} = \infty) = 0$.
For every such $i'$, we define the set of paths $F_{i'-1} =\{
  f_{B_{i'-1}}(\run) \mid \run \in \Runs, \, B_{i'-1}(\run) < \infty \}$ where as before, $f_{B_{i}}(\run) = \run_{1}\ldots \run_{B_{i}(\run)}$ for all $i \in \NN$ and all $\run = \run_{1}\run_{2}\ldots$ with $B_{i}(\run) < \infty$. The induction hypothesis implies
\begin{equation}
  \label{probFi-1}
  \P(\Pre_{F_{i'-1}}) = 1.
\end{equation}
Now consider an arbitrary subset $F \subseteq F_{i'-1}$.
Note that
\begin{equation}
  \label{disjoint}
  \mbox{for all $f,f' \in F$ with $f \neq f'$ we have $\Pre_{f} \cap \Pre_{f'} =
      \emptyset$,}
\end{equation}
i.e., $F$ does not contain two paths where one is a proper prefix of the other.
For paths in $F$ (which therefore are prefixes of a run $\run$ that ends in $B_{i'-1}(\run)$), we now want to approximate how many steps one needs in expectation from this point onwards to reach $B_{i'}(\run)$.
We have
\begin{align}
  \E \left((B_{i'} - B_{i'-1}) \cdot \ind_{\Pre_{F}}\right)
   & = \sum_{f \in F} \E \left( (B_{i'} - B_{i'-1}) \cdot \ind_{\Pre_{f}} \right) \nonumber \\
   & = \sum_{f \in F}\sum_{f' \in {\{\symMat{A},\symMat{B}\}}^{l}} \E \left((B_{i'} - B_{i'-1}) \cdot \ind_{\Pre_{ff'}} \right)
  \label{stepsFromi-1Toi}
\end{align}

Recall the constants $C_{\mathrm{in}},C_{\mathrm{out}}$ from our previous discussion on $\E(T_{\mathrm{in}})$ and $\E(T_{\mathrm{out},h})$ for $h\in\{1,2\}$ and set $C = \max\{C_{\mathrm{in}},C_{\mathrm{out}}\}$.
In the following we first show
\begin{equation}
  \E \left( \left(B_{i'} - B_{i'-1}\right) \ind_{\Pre_{ff'}}\right) \leq \P(\Pre_{ff'}) \cdot (l + C) \tag{\dag} \label{eq:dual_positive_eigenvalues_for_eventually_dominating_constraints_remaining_parts_dag}
\end{equation}
for all addends of the above sum, i.e., all $f \in F$ and $f'\in{\{\symMat{A},\symMat{B}\}}^{l}$.
Recall that $B_{i'}$ is at least $l$ greater than $B_{i'-1}$ (we defined it in this way such that $B_{i'}$ is greater than the number of steps needed to reach $B_{i'-1}$ and execute $g_f$), corresponding to the addend $l$ on the right-hand side of the equation above.
We distinguish three mutually exclusive cases depending upon the value of $\URVar(ff')$.
\begin{enumerate}
  \item The execution $ff'$ leads to a safe region, i.e., we have $|\URVar(ff') \cdot |ff'|| \geq r$ and $\URVar (ff') \in [-\eps,\eps]$.
        Then, $\E ((B_{i'} - B_{i'-1}) \ind_{\Pre_{ff'}}) = \E ((|ff'| - |f|) \ind_{\Pre_{ff'}}) \leq l \cdot \P(\Pre_{ff'})$.
  \item The execution $ff'$ leads to a position between the two safe regions, i.e., to area $(A)$ in \Cref{fig:safe_with_ABC}.
        In other words, we have $|\URVar (ff') \cdot |ff'|| < r \le \eps l - 1 $.
        Then, a safe region is entered within the execution $ff'f''$ for some $f''\in\Path$ when $|\URVar (f'') \cdot |f''|| \geq 2r$.
        The reason is that for some prefix $f_{s}$ of $f''$ we have $r \leq |\URVar (ff'f_{s}) \cdot |ff'f_{s}|| \leq r+1 \leq \eps l$, since when $|f|$ increases by $1$, then $\URVar(f) \cdot |f|$ changes by at most 1 (see the discussion before defining $T_{\mathrm{in}}$).
        So $f_s$ is the minimal prefix of $f''$ such that the path $ff'f_{s}$ enters the safe region.

        Combining this with our previous discussion on $\E (T_{\mathrm{in}})$, one obtains
        \begin{align*}
                 & \E \left((B_{i'} - B_{i'-1}) \cdot \ind_{\Pre_{ff'}}\right) \\
          \leq{} & \E \big( \big(|ff'| + \min\{j \in \NN \mid |\URVar (\run_{|ff'| + 1}\ldots\run_{|ff'| + j})j| \geq 2r \} - B_{i'-1}\big) \cdot \ind_{\Pre_{ff'}} \big) \\
          ={}    & \E \big( \big(|f'| + \min\{j \in \NN \mid |\URVar (\run_{|ff'|+1}\ldots\run_{|ff'| + j})j| \geq 2r \}\big) \cdot \ind_{\Pre_{ff'}} \big) \tag{as $B_{i'-1} = |f|$} \\
          ={}    & \P (\Pre_{ff'}) \cdot \E \left( |f'| + \min\{j \in \NN \mid |\URVar (\run_{1}\ldots\run_{j})j| \geq 2r \} \right) \tag{as $\run_{j}$ are i.i.d.} \\
          ={}    & \P (\Pre_{ff'}) \cdot (l + \E \left(T_{\mathrm{in}}\right)) \leq \P(\Pre_{ff'}) \cdot (l + C_{\mathrm{in}}),
        \end{align*}
        where $C_{\mathrm{in}} \leq C$.
        For the penultimate line, note that since $\run_{j}$ are i.i.d.\ and we are interested in the expected value, it does not matter whether we regard paths of length $j$ that start at position $|f f'| +1$ or paths of length $j$ that start at position $1$.

  \item The execution $ff'$ leads to a position outside of the safe regions, i.e., to area $(B)$ or $(C)$ in \Cref{fig:safe_with_ABC}.
        In other words, we have $|\URVar (ff') \cdot |ff'|| > \eps |ff'| \geq r+1$.
        We have the following, where the second line holds because $ \left||f|_{\symMat{A}} - p|f| \right| = |\URVar(f) \cdot |f|| \leq \eps\cdot |f|$ since $B_{i'-1}
          = |f|$ means that after the path $f$ one is in a safe region.
        \begin{align*}
          \left| \URVar(ff')\cdot|ff'| \right| & = \left||f|_{\symMat{A}} - p|f| + |f'|_{\symMat{A}} - p|f'|\right| \\
                                               & \leq \left| |f|_{\symMat{A}} - p|f| \right| + \left| |f'|_{\symMat{A}} - p|f'| \right| \\
                                               & \leq \eps|f| + \left| |f'|_{\symMat{A}} - p|f'| \right| \tag{$B_{i'-1} = |f|$} \\
                                               & < \eps|f| + |f'| = \eps|f| + l.
        \end{align*}
        We first regard the case where $ff'$ leads to a position in the unsafe area $(B)$, i.e., where $\URVar(ff') \cdot |ff'| > \eps|ff'|$.
        Then, the safe region is entered within the execution $ff'f''$ for some $f''\in\Path$ when $\URVar(ff'f'') \cdot |ff'f''| \leq \eps|ff'f''|$.
        This is for instance the case when $\URVar(f'') \cdot |f''| \leq \eps|f''| - l$.
        Intuitively, the reason is that because of $|f'| = l$, after the path $ff'$ we are at most $l$ above the upper border of the safe region.
        So if one extends $ff'$ by such an $f''$ we reach (or traverse) the safe region.
        Formally, $\URVar(f'') \cdot |f''| \leq \eps|f''| - l$ implies
        \begin{align*}
          \URVar(ff'f'') \cdot |ff'f''|
           & = |ff'f''|_{\symMat{A}} - p \cdot (|f| + |f'| + |f''|) \\
           & \leq |f|_{\symMat{A}} - p\cdot |f| + |f'|_{\symMat{A}} + \URVar(f'') \cdot |f''| \\
           & \leq \underbrace{|f|_{\symMat{A}} - p\cdot |f|}_{\leq \eps |f|, \text{ as $B_{i'-1} = |f|$}}+ \underbrace{|f'|_{\symMat{A}}}_{\leq l} + \eps|f''| - l \\
           & \leq \eps |f| + \eps|f''| \\
           & \leq \eps|ff'f''|
        \end{align*}
        Thus, then there is a prefix $f_{s}$ of $f''$ such that $ff'f_{s}$ reaches the safe region.
        If $f_{s}$ is the minimal prefix of $f''$ with $\eps \cdot |ff'f_{s}| \geq \URVar(ff'f_{s})\cdot |ff'f_{s}|$, then we must have $\URVar(ff'f_{s})\cdot |ff'f_{s}| > r$ as $\eps|ff'f_{s}| \geq \eps |f'| = \eps l \geq r+1$.

        The case where $ff'$ leads to a position in the unsafe area $(C)$ is analogous, i.e., where $\URVar(ff')|ff'| < - \eps|ff'|$.
        Here, the lower safe region is entered within the execution $ff'f''$ for $f''\in\Path$ when $\URVar(f'') \cdot |f''| \geq l - \eps |f''|$.

        Combining the first case $\URVar(ff') \cdot |ff'| \geq \eps|ff'|$ with our previous discussion on $\E (T_{\mathrm{out},1})$, one obtains
        \begin{align*}
               & \E \left( (B_{i'} - B_{i'-1}) \cdot \ind_{\Pre_{ff'}} \right) \\
          \leq & \E \big( \big(|ff'| + \min\{ j\in\NN \mid \URVar(\run_{|ff'|+1}\ldots\run_{|ff'|+j}) \cdot j \leq \eps j - l\} - B_{i'-1} \big) \cdot \ind_{\Pre_{ff'}} \big) \\
          =    & \E \big( \big(|f'| + \min\{
            j\in\NN \mid \URVar(\run_{|ff'|+1}\ldots\run_{|ff'|+j}) \cdot j \leq \eps j -
          l\} \big) \cdot \ind_{\Pre_{ff'}} \big) \tag{as $B_{i'-1} = |f|$} \\
          =    & \P(\Pre_{ff'}) \cdot \big(l + \E \big( \min\{ j\in\NN \mid \URVar(\run_{|ff'|+1}\ldots\run_{|ff'|+j}) \cdot j + l - \eps j \leq 0\} \big)\big) \\
          =    & \P(\Pre_{ff'}) \cdot \big(l + \E \big( \min\{
            j\in\NN \mid \URVar(\run_{1}\ldots\run_{j}) \cdot j + l - \eps j \leq
          0\} \big)\big) \tag{as $\run_{j}$ are i.i.d.} \\
          =    & \P(\Pre_{ff'}) \cdot \left(l + \E \left(T_{\mathrm{out},1}\right)\right) \leq \P(\Pre_{ff'}) \cdot \left(l + C_{\mathrm{out}}\right)
        \end{align*}
        where $C_{\mathrm{out}} \leq C$.
        If $\URVar(ff') \cdot |ff'| < \eps|ff'|$, one similarly derives $\E ((B_{i'} - B_{i'-1}) \cdot \ind_{\Pre_{ff'}}) \leq \P(\Pre_{ff'}) \cdot (l + C_{\mathrm{out}})$.
\end{enumerate}

Now we obtain
\begin{align*}
  \E (B_{i})
   & = \sum_{i'=1}^{i} \E \left(B_{i'} - B_{i'-1}\right) \\
   & = \sum_{i'=1}^{i} \P(\Pre_{F_{i'-1}}) \cdot \E \left(B_{i'} - B_{i'-1}\right) \tag{by \cref{probFi-1}} \\
   & = \sum_{i'=1}^{i} \E \left( (B_{i'} - B_{i'-1})\ind_{\Pre_{F_{i'-1}}} \right) \\
   & = \sum_{i'=1}^{i}\sum_{f \in F_{i'-1}}\sum_{{f'\in\{\symMat{A},\symMat{B}\}}^{l}} \E \left(\left(B_{i'}-B_{i'-1}\right) \ind_{\Pre_{ff'}}\right) \tag{by \cref{stepsFromi-1Toi}} \\
   & \leq \sum_{i'=1}^{i}\sum_{f \in F_{i'-1}}\sum_{{f'\in\{\symMat{A},\symMat{B}\}}^{l}} \P(\Pre_{ff'}) \cdot (l + C) \tag{by \cref{eq:dual_positive_eigenvalues_for_eventually_dominating_constraints_remaining_parts_dag}} \\
   & = \sum_{i'=1}^{i} (l + C) \tag{as $\sum_{f \in F_{i'-1}}\sum_{{f'\in\{\symMat{A},\symMat{B}\}}^{l}} \P(\Pre_{ff'}) = 1$} \\
   & = i \cdot (l + C) \\
   & < \infty
\end{align*}
implying $\P(B_{i} = \infty) = 0$, which proves the induction step.
As $P(B_{i} = \infty) = 0$ for all $i\in\NN$, we will ignore runs $\run\in\Runs$ with $B_{i} (\run) = \infty$ from now on.
This motivates the denotation of the path $\run_{1}\ldots\run_{B_{i}(\run)}$ by $f_{B_{i}}(\run)$.
\medskip

Finally, we prove \cref{eq:positive_eigenvalues_dominating_constraints1}, i.e., we prove the over-approximation
\begin{equation*}
  \E \left((B_{i} - B_{i-1}) \prod_{i'=1}^{i-1} \ind_{\Pre_{f_{B_{i'}} \, \overline{g_{f_{B_{i'}}}}}} \right) \leq {(1 - {(\min\{p,1-p\})}^{l})}^{i-1} (l+C)
\end{equation*}
for all $i\in\NN_{>0}$, where $g_f \in \Path$ are possibly empty finite paths such that $|g_f| \leq l$ for all $f\in\Path$.
We denoted the set $\Pre_{f}\setminus\Pre_{f \, g_f}$ as $\Pre_{f \, \overline{g_f}}$ for all $f\in\Path$.

Let $i \in\NN_{>0}$.
Then let $F\subseteq\Path$ be the set of all prefixes of runs $\run$ which end in $B_{i}(\run)$ and where for all $i' < i$ we have that $f_{B_{i'}}(\run) g_{f_{B_{i'}}}(\run)$ is not a prefix of $\run$.
Thus,
\[ F = \left\{ f_{B_{i}}(\run) \Bigm\vert \run \in \Runs,
  \left(\prod_{i'=1}^{i-1} \ind_{\Pre_{f_{B_{i'}}\overline{g_{f_{B_{i'}}}}}}\right) (\run) =
  1\right\}.\]

To obtain an upper bound for the probability of runs in $\Pre_F$, note that the probability for paths from $g_{f_{B_{i'}}}$ is at least $(\min\{p,1-p\})^l$ and hence, the probability for paths from $\overline{g_{f_{B_{i'}}}}$ is at most $1 - {(\min\{p,1-p\})}^{l}$.
Thus, the probability for runs where $\left(\prod_{i'=1}^{i-1} \ind_{\Pre_{f_{B_{i'}}\overline{g_{f_{B_{i'}}}}}}\right) (\run) = 1$ is at most ${\left(1 - {(\min\{p,1-p\})}^{l}\right)}^{i-1}$ which implies $\P(\Pre_F) \leq {\left(1 - {(\min\{p,1-p\})}^{l}\right)}^{i-1}$.
We conclude
\begin{align*}
  \E \left((B_{i} - B_{i-1}) \prod_{i'=1}^{i-1} \ind_{\Pre_{f_{B_{i'}} \, \overline{g_{f_{B_{i'}}}}}} \right)
   & = \sum_{f \in F} \E \left(\left(B_{i} - B_{i-1}\right) \ind_{\Pre_{f}}\right) \tag{by \cref{disjoint}} \\
   & = \sum_{f \in F}\sum_{f'\in{\{\symMat{A},\symMat{B}\}}^{l}} \E \left(\left(B_{i} - B_{i-1}\right) \ind_{\Pre_{ff'}}\right) \tag{by \cref{stepsFromi-1Toi}} \\
   & \leq \sum_{f \in F}\sum_{f'\in{\{\symMat{A},\symMat{B}\}}^{l}} \P (\Pre_{ff'}) (l + C) \tag{by \cref{eq:dual_positive_eigenvalues_for_eventually_dominating_constraints_remaining_parts_dag}} \\
   & = \sum_{f \in F}\P (\Pre_{f}) (l + C) \\
   & = (l + C) \cdot \P(\Pre_F) \\
   & \leq {(1 - {(\min\{p,1-p\})}^{l})}^{i-1} (l+C) .
\end{align*}

%% file: illustration_full_middle_with_ABC.tex
{
\begin{tikzpicture}[scale=1.1,every node/.style={font=\tiny}]
  \draw [draw=none, fill=SeaGreen!5] (2,0.75) -- (5,0.75) -- (5,2) -- (4,2) -- (2,1) -- (2,0.75);
  \draw [draw=none, line space=5pt,pattern=my north east lines, pattern color=SeaGreen!30] (2,0.75) -- (5,0.75) -- (5,2) -- (4,2) -- (2,1) -- (2,0.75);

  \draw [draw=none, fill=SeaGreen!5] (2,-0.75) -- (5,-0.75) -- (5,-2) -- (4,-2) -- (2,-1) -- (2,-0.75);
  \draw [draw=none,line space=5pt,pattern=my north east lines, pattern color=SeaGreen!30] (2,-0.75) -- (5,-0.75) -- (5,-2) -- (4,-2) -- (2,-1) -- (2,-0.75);

  \draw [->,line width=0.6pt] (-1,0)--(5,0) node[below]{\footnotesize $|f|$};
  \draw [->,line width=0.6pt] (0,-2)--(0,2) node[left]{\footnotesize $\URVar(f) \cdot |f|$};

  \node[Brown] at (2.5, 1.75) (A) {(B)};
  \node[Brown] at (2.5, -1.75) (A) {(C)};

  \draw[domain=0:3, smooth, variable=\x, gray] plot ({\x}, {0.6875*\x*\x*\x - 3.313*\x*\x + 4.125*\x});
  \node[gray] at (0.8, 1.65) (A) {$\URVar(f)\cdot|f|$};
  \node[gray] at (3,1.125)[circle,fill,inner sep=0.75pt]{};

  \draw [red] (0,0) -- (4,2) node[below]{$\eps |f|$};
  \draw [red] (0,0) -- (4,-2) node[above]{$-\eps |f|$};

  \draw [blue] (-0.75, 0.75) node[below,xshift=0.25cm]{$ r$} -- (5, 0.75);
  \draw [blue] (-0.75,-0.75) node[above,xshift=0.25cm]{$-r$} -- (5,-0.75);

  \draw [Green] (2,-2) node[below]{$l$} -- (2,2);

  \draw [Brown,decorate,decoration={brace,mirror,amplitude=5pt,aspect=0.65}] (3,-0.70) -- (3,0.70) node [pos=0.65,right,xshift=3pt] {(A)};
\end{tikzpicture}
}

%% file: towards_ent_witnesses_app.tex
\proofsForSection{sec:towards_witnesses_for_non-termination}
\label{app:proofs_for_towards_witnesses_for_non-termination}

\soundnessofwitnesses*{}
\begin{proof}
  Let $\vecx \in W$.
  Thus, there is some $d\in\{\nmax,\pmax\}$ such that for all $c\in[m]$ we have
  \begin{equation}
    \label{realGreaterComplex}
    \sum_{i\in\mathfrak{R}_{d,c,\vecx}} \gamma_{c,i}(\vecx) > \sum_{i\in\mathfrak{C}_{d,c,\vecx}} |\gamma_{c,i}(\vecx)|.
  \end{equation}
  W.l.o.g., we assume $d=\pmax$ as the other case is symmetric.

  For all $c\in[m]$ and $f\in\Path$, we define $v_{c}(f) = \sum_{i\in\constrainttermgroup_{\pmax,c,\vecx}} \zeta_{i,\matA}^{|f|_{\symMat{A}}} \zeta_{i,\matB}^{|f|_{\symMat{B}}} \gamma_{c,i}(\vecx)$.
  We have
  \begin{align*}
    v_{c}(f) & = \sum_{i\in\mathfrak{R}_{\pmax,c,\vecx}} \gamma_{c,i}(\vecx) +
    \sum_{i\in\mathfrak{C}_{\pmax,c,\vecx}} \zeta_{i,\matA}^{|f|_{\symMat{A}}}
    \zeta_{i,\matB}^{|f|_{\symMat{B}}} \gamma_{c,i}(\vecx) \\
             & \geq \sum_{i\in\mathfrak{R}_{\pmax,c,\vecx}} \gamma_{c,i}(\vecx) -
    \left| \sum_{i\in\mathfrak{C}_{\pmax,c,\vecx}} \zeta_{i,\matA}^{|f|_{\symMat{A}}}
    \zeta_{i,\matB}^{|f|_{\symMat{B}}} \gamma_{c,i}(\vecx) \right| \\
             & \geq \sum_{i\in\mathfrak{R}_{\pmax,c,\vecx}} \gamma_{c,i}(\vecx) -
    \sum_{i\in\mathfrak{C}_{\pmax,c,\vecx}} \left|\zeta_{i,\matA}^{|f|_{\symMat{A}}}
    \zeta_{i,\matB}^{|f|_{\symMat{B}}}\right| \cdot \left|\gamma_{c,i}(\vecx) \right| \\
             & = \sum_{i\in\mathfrak{R}_{\pmax,c,\vecx}} \gamma_{c,i}(\vecx) -
    \sum_{i\in\mathfrak{C}_{\pmax,c,\vecx}} \left|\gamma_{c,i}(\vecx) \right|
  \end{align*}
  where we used the triangle inequation and that $|\zeta_{i,\matA}| = |\zeta_{i,\matB}| = 1$ for all $i\in[n]$.
  For all $c\in[m]$, we define $\rho_{c}$ as the value of this inequation's right-hand side, where $\rho_c > 0$ by \cref{realGreaterComplex}.

  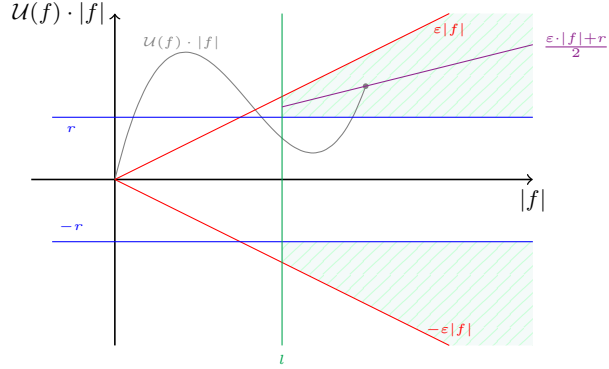
\begin{figure}
    \begin{center}
      \input{illustration_full_middle.tex}
    \end{center}
    \caption{\label{fig:safe_with_middle} Illustration of the General Idea Using Safe Regions with Middle Line}
  \end{figure}

  Then, by \Cref{lem:domination_of_eventually_dominating_constraint_term_groups}, for all $c\in[m]$ there are constants $\eps_{c} \in \AA_{>0}$, $r_{c} \in \NN$, and $l_{c} \in \NN_{>0}$ such that $\sign({(\Val_{\vecx} (f))}_{c}) = \sign(v_{c}(f)) = 1$ for all $f \in \Path$ with $|f| \geq l_c$, $\URVar(f) \cdot |f| \geq r_c$, and $\URVar\in[0,\eps_{c}]$.
  We select $r\in\NN$ with $r \geq \max\{r_{1},\ldots, r_{m}\}$, $\eps \in (0,1-p]$ with $\eps \leq \min \{\eps_{1},\ldots,\eps_{m}\}$, and $l \geq \max\{l_{1}, \ldots,l_{m}\}$ large enough such that $\eps l \geq r+2$.
  Then, $\Val_{\vecx} (f) > \veczero$ whenever $|f| \geq l$, $\URVar(f) \cdot |f| \geq r$, and $\URVar(f) \in [0,\eps]$.
  The intuition is that we regard the line $\frac{\eps \cdot |f| + r}{2}$ in the middle of the upper safe region (the purple line in \Cref{fig:safe_with_middle}) and the \emph{middle region} which is between 1 above and 1 below this middle line.
  We need $\eps \leq 1-p$ to ensure that $\URVar(f) \cdot |f|$ can reach the whole (upper) safe region (recall that $\URVar(f) \leq 1-p$).
  Moreover, we need $\eps l \geq r+2$ (i.e., $\eps \cdot l - r \geq 2$) to ensure that the whole middle region is in the safe region.

  Let $\hat{f} \in \Path$ satisfy these requirements such that additionally $\abs{\URVar(\hat{f}) \cdot |\hat{f}| - \frac{\eps |\hat{f}| + r}{2}} \leq 1$, i.e., $\URVar(\hat{f}) \cdot |\hat{f}|$ reaches the middle region.
  This will be required for \cref{eq:soundness_of_witnesses1}, see the end of the proof.
  To see that such an $\hat{f}$ always exists consider the surjective mapping $u_{n}\colon
    {\{\matA,\matB\}}^{n} \rightarrow \{r_{n,i} \mid i \in \{0,\ldots,n\}\}$ with $r_{n,i}
    = i - p n$ and $u_{n}(f) = \URVar(f) \cdot |f| = |f|_{\symMat{A}} - p |f|$ for $n \in \NN$. As we only consider paths $\hat{f}$ with $|\hat{f}| \geq l$ let $n \geq l$. We have $\eps n + (1-p) n \geq r + (1-p) n$ (as $\eps l \geq r + 2$) and thus $2 (1-p) n \geq r + \eps n$ (as $\eps \leq 1 - p$) implying $r_{n,n} = (1-p) n \geq \frac{\eps n + r}{2} \geq 0 \geq r_{n,0}$. Moreover, $r_{n,i+1} - r_{n,i} = 1$ for all $i \in \{0,\ldots,n-1\}$ and therefore there exists an $i \in \{0,\ldots,n\}$ with $|r_{n,i} - \frac{\eps n + r}{2}| \leq 1$. Let $\hat{f}$ with $|\hat{f}| = n \geq l$ be such that $|\hat{f}|_{\symMat{A}} = i$. Then, $|\URVar(\hat{f}) \cdot |\hat{f}| - \frac{\eps n + r}{2}| = |u_{n}(\hat{f}) - \frac{\eps|\hat{f}| + r}{2}|= |r_{n,i} - \frac{\eps n + r}{2}| \leq 1$.

  Consider the ($\F$-measurable) random variable $T\colon \Runs \to \NN\cup\{\infty\}$ defined as \input{./ent_witnesses_definition_T}
  So $T$ maps any run $\run$ to the smallest $i$ such that if $\hat{f}$ takes place before $\run$, then $\run$ reaches an unsafe region after performing $i$ steps from $\hat{f}$ (according to \Cref{fig:safe_with_middle}).
  The expectation of the random variable $T$ is infinite, i.e.,
  \begin{equation}
    \E (T) = \infty. \label{eq:soundness_of_witnesses1}
  \end{equation}
  We will show \cref{eq:soundness_of_witnesses1} at the end of the proof.

  Let $\vecy = \matA^{|\hat{f}|_{\symMat{A}}} \matB^{|\hat{f}|_{\symMat{B}}} \vecx$, i.e., $\vecy$ results from $\vecx$ by executing the program according to the path $\hat{f}$.
  We now show $\vecy \in \NT$, which implies $\vecx \in \ENT$.

  If $T(\run) = k$ for some arbitrary run $\run\in\Runs$ and $k\in\NN\cup\{\infty\}$, then $\URVar(\hat{f}\run_{1}\ldots\run_{i}) \cdot i \geq r$ as well as $\URVar(\hat{f}\run_{1}\ldots\run_{i}) \in [0,\eps]$ for all $i\in\{1,\ldots,k-1\}$.
  Note that since $\URVar(\hat{f}\run_{1}\ldots\run_{i}) \geq 0$, this implies $\URVar(\hat{f}\run_{1}\ldots\run_{i}) \cdot (|\hat{f}| + i) \geq \URVar(\hat{f}\run_{1}\ldots\run_{i}) \cdot i \geq r$.
  It follows that $\Val_{\vecy} (\run_{1}\ldots\run_{i}) = \Val_{\vecx} (\hat{f}\run_{1}\ldots\run_{i}) > \veczero$.
  Therefore, $\LRVar_{\vecy} (\run) > i$ and hence, $\LRVar_{\vecy} (\run) \geq k$.
  As this holds for arbitrary runs, one concludes $\LRVar_{\vecy} \geq T$.

  Combined with \cref{eq:soundness_of_witnesses1}, one obtains
  \[
    \E \left(\LRVar_{\vecy}\right) \geq \E (T) = \infty,
  \]
  completing the proof of \Cref{lem:soundness_of_witnesses}.

  \medskip

  \input{./ent_witnesses_remaining_parts_app.tex}
\end{proof}

\boosting*{}
\begin{proof}
  Let $\vecx \in \semiring^{n}$ be a non-terminating input, i.e., $\vecx \in \NT$.
  Then, according to \Cref{lem:dual_positive_eigenvalues_for_eventually_dominating_constraints}, there is a $d\in\{\nmax,\pmax\}$ such that for all $c\in[m]$ we have $\sum_{i\in\mathfrak{R}_{d,c,\vecx}} \gamma_{c,i}(\vecx) > 0$.
  In particular, $\constrainttermgroup_{d,c,\vecx} \neq \emptyset$.
  For all $c\in[m]$ we consider the set $\{I_{c,1},\ldots,I_{c,l_{c}}\} = \{i \in \mathfrak{C}_{d,c,\vecx} \mid \gamma_{c,i}(\vecx) \neq 0\}$ of indices corresponding to at least one non-real or non-positive eigenvalue for the eventually dominating terms corresponding to the constraint $c$.
  Here, we filter out indices $i \in \mathfrak{C}_{d,c,\vecx}$ with $\gamma_{c,i}(\vecx) = 0$ as they do not contribute to the sum $\sum_{i\in\mathfrak{C}_{d,c,\vecx}} |\gamma_{c,i} (\vecx)|$.
  In addition, this will allow us to avoid a division by $0$ later on.

  For $c\in[m]$, we define
  \[
    w_{c}(\vecx) = \sum_{i\in\mathfrak{R}_{d,c,\vecx}} \gamma_{c,i} (\vecx) .
  \]

  We will then inductively construct inputs $\vecx_{\!c,j} \in \semiring^{n}$ for all $(c,j) \in [m]\times ([l_c] \cup \{0\})$ such that
  \begin{equation}
    \text{for all $(c',j') \leq_{\mathrm{lex}} (c,j):$ If $j' \neq 0$, then } \frac{j'}{l_{c'}} w_{c'}\left(\vecx_{\!c,j}\right) > \sum_{\tilde{j}\in [j']} \left|
    \gamma_{c',I_{c',\tilde{j}}} \left(\vecx_{\!c,j}\right) \right|
    \label{boosting claim}
  \end{equation}
  Here, we use the lexicographic ordering, where $\leq_{\mathrm{lex}}$ is defined like $\leq_{\mathrm{lex},\pmax}$ in \Cref{def:dominating_constraint_term_groups}, i.e., $(c',j') <_{\mathrm{lex}} (c,j)$ iff $c' < c$ or both $c' = c$ and $j' \leq j$.
  In \cref{boosting claim}, note that if $l_{c'} = 0$, then $j' = 0$, i.e., we never divide by 0.

  So for all $c' \in [m]$ and $j' = l_{c'} \neq 0$, we have $w_{c'}\left(\vecx_{\!m,l_m}\right) > \sum_{\tilde{j} \in [l_{c'}]} \left| \gamma_{c',I_{c',\tilde{j}}} \left(\vecx_{\!m,l_m}\right) \right|$, i.e., $\sum_{i\in\mathfrak{R}_{d,c',\vecx_{\!m,l_m}}} \gamma_{c',i} (\vecx_{\!m,l_m}) > \sum_{i\in\mathfrak{C}_{d,c',\vecx_{\!m,l_m}}}\left| \gamma_{c',i} \left(\vecx_{\!m,l_m}\right) \right|$.
  Moreover, $w_{c'}\left(\vecx_{\!m,l_m}\right) > \sum_{\tilde{j} \in [l_{c'}]} \left| \gamma_{c',I_{c',\tilde{j}}} \left(\vecx_{\!m,l_m}\right) \right|$ also holds if $l_{c'} = 0$.
  Then the right-hand side is 0.
  For the left-hand side we have $w_{c'}\left(\vecx\right) = \sum_{i\in\mathfrak{R}_{d,c',\vecx}} \gamma_{c',i} (\vecx) > 0$.
  We will later see that for $i\in\mathfrak{R}_{d,c',\vecx}$, the $\gamma_{c',i}(\vecx_{\!m,l_m})$ only result from $\gamma_{c',i}(\vecx)$ by multiplication and addition with positive numbers.
  Thus, this implies $w_{c'}\left(\vecx_{\!m,l_m}\right) = \sum_{i\in\mathfrak{R}_{d,c',\vecx_{\!m,l_{m}}}} \gamma_{c',i} (\vecx_{\!m,l_{m}}) > 0$.
  As we will explain below, for all $(c,j) \in [m] \times ([l_{c} \cup \{0\}])$ and $c' \in [m]$ we have $\constrainttermgroup_{d,c',\vecx_{\!c,j}} = \constrainttermgroup_{d,c',\vecx}$ and hence $\mathfrak{R}_{d,c',\vecx_{\!c,j}} = \mathfrak{R}_{d,c',\vecx}$ as well as $\mathfrak{C}_{d,c',\vecx_{\!c,j}} = \mathfrak{C}_{d,c',\vecx}$.

  For the construction of $\vecx_{\!c,j}$, we use induction w.r.t.\ the lexicographic ordering, i.e., when constructing $\vecx_{\!c,j}$, we assume that we already have the input $\vecx_{\!c',j'}$ for its ``predecessor'' $(c',j') <_{\mathrm{lex}} (c,j)$.

  We first regard the case $j = 0$.
  Here, we set $\vecx_{1,0} = \vecx\in\semiring^{n}$ and $\vecx_{\!c,0} = \vecx_{\!c-1,l_{c-1}}\in\semiring^{n}$ for $c\in\{2,\ldots,m\}$.

  Now we regard the case $j \geq 1$.
  W.l.o.g., we assume $\zeta_{I_{c,j},\matA} \neq 1$ as otherwise one simply has to consider $\matB$ instead of $\matA$ in the following.
  We first show the existence of natural numbers $e_{1},\ldots,e_{k} \in \NN$ such that $|\zeta_{I_{c,j},\matA}^{e_{1}} + \cdots + \zeta_{I_{c,j},\matA}^{e_{k}} + 1| < \frac{w_{c}(\vecx_{\!c,j-1})}{l_{c} \cdot |\gamma_{c,I_{c,j}}(\vecx_{\!c,j-1})|}$.
  As $l_{c}, |\gamma_{c,I_{c,j}}(\vecx_{\!c,j-1})| > 0$, the fraction on the right-hand side is well defined and by \Cref{lem:dual_positive_eigenvalues_for_eventually_dominating_constraints} as well as our previous discussion we have $\frac{w_{c}(\vecx_{\!c,j-1})}{l_{c} \cdot |\gamma_{c,I_{c,j}}(\vecx_{\!c,j-1})|} > 0$.
  We distinguish two cases depending on whether the complex unit $\zeta_{I_{c,j},\matA}$ is a root of unity.
  \begin{enumerate}
    \item $\zeta_{I_{c,j},\matA} = e^{2\pi \im r}$ for some rational $r = \frac{p}{q}$ with $p,q \in \ZZ$, i.e., $\zeta_{I_{c,j},\matA}$ is a $q$-th root of unity.
          Then, $\zeta_{I_{c,j},\matA}^{q} = 1$ and $\left|1 + \sum_{\tilde{k}=1}^{q-1} \zeta_{I_{c,j},\matA}^{\tilde{k}}\right| = \left|\frac{\zeta_{I_{c,j},\matA}^{q} - 1}{\zeta_{I_{c,j},\matA} - 1}\right| = 0$.
    \item $\zeta_{I_{c,j},\matA} = e^{2 \pi \im r }$ for some irrational real $r\in\RR\setminus\QQ$.
          It is well known that the orbit $\{\zeta_{I_{c,j},\matA}^{e_{1}} \mid e_{1}\in\NN\}$ of such a complex unit is dense in the unit circle $\{z \in \CC \mid |z|=1\}$.\footnote{This is called irrational rotation.
            For more details we refer to~\cite[Thm.\ 3.2.3]{ergodictheory} for a proof.}
          So $\zeta_{I_{c,j},\matA}^{e_{1}}$ can get arbitrarily close to $-1$ and thus, $\zeta_{I_{c,j},\matA}^{e_{1}} + 1 \in \CC$ can get arbitrarily close to $0$.
          Hence, there is some $e_{1} \in \NN$ such that $|\zeta_{I_{c,j},\matA}^{e_{1}} + 1| < \frac{w_{c} (\vecx_{\!c,j-1})}{l_{c}
            \left| \gamma_{c,I_{c,j}} (\vecx_{\!c,j-1})\right|}$.
  \end{enumerate}

  Following \Cref{lem:shared_modulus_of_eigenvalues_in_constraint_term_groups}, let $\hat{a}_{c}$ be the shared modulus of all eigenvalues $a_{i}$ of $\matA$ for $i\in\constrainttermgroup_{d,c,\vecx}$, i.e., $|a_i| = \hat{a}_{c}$ for all $i \in \constrainttermgroup_{d,c,\vecx}$.
  As the rationals are dense in the reals, one chooses positive rationals $r_{1} = \tfrac{p_{1}}{q},\ldots,r_{k} = \tfrac{p_{k}}{q}$ with $q\in\NN_{>0}$ and $p_{\tilde{k}} \in \NN_{>0}$ for all $\tilde{k}\in[k]$ that are ``close-enough'' approximations of $\hat{a}_{c}^{-e_{1}},\ldots,\hat{a}_{c}^{-e_{k}}$ such that
  \begin{equation*}
    \left|\frac{p_{1}}{q} \hat{a}_{c}^{e_{1}} \zeta_{I_{c,j},\matA}^{e_{1}} + \cdots + \frac{p_{k}}{q} \hat{a}_{c}^{e_{k}} \zeta_{I_{c,j},\matA}^{e_{k}} + 1\right| < \frac{w_{c} (\vecx_{\!c,j-1})}{l_{c} \cdot \left| \gamma_{c,I_{c,j}} (\vecx_{\!c,j-1}) \right|} .
  \end{equation*}
  By multiplying on both sides with the (positive) denominator $q$ one obtains
  \begin{equation}
    \left|p_{1} \hat{a}_{c}^{e_{1}} \zeta_{I_{c,j},\matA}^{e_{1}} + \cdots + p_{k} \hat{a}_{c}^{e_{k}} \zeta_{I_{c,j},\matA}^{e_{k}} + q \right| < q \cdot \frac{w_{c}(\vecx_{\!c,j-1})}{l_{c} \cdot \left| \gamma_{c,I_{j}} (\vecx_{\!c,j-1}) \right|}.
    \label{eq:boosting1}
  \end{equation}
  We set $\vecx_{\!c,j} =(p_{1} \matA^{e_{1}} + \cdots + p_{k} \matA^{e_{k}})\cdot \vecx_{\!c,j-1} + q \cdot \vecx_{\!c,j-1}$.
  Then, $\vecx_{\!c,j} \in \semiring^{n}$ as $\matA\in\semiring^{n \times n}$ and $p_{1},\ldots,p_{k},q\in\NN_{>0} \subseteq \semiring$.

  For the term on the left-hand side of \cref{boosting claim}
  for all $c' \leq c$ we have:
  \begin{align*}
    w_{c'} (\vecx_{\!c,j})
     & = \!\!\!\!\!\!\sum_{i \in \mathfrak{R}_{d,c',\vecx_{\!c,j}}}\!\!\!\!\!\!\!\!\!\! \gamma_{c',i} (\vecx_{\!c,j}) \\
     & = \!\!\!\!\!\!\sum_{i \in \mathfrak{R}_{d,c',\vecx_{\!c,j}}}\!\!\!\!\!\!\!\!\!\! \gamma_{c',i} ((p_{1} \matA^{e_{1}} + \cdots + p_{k} \matA^{e_{k}})\cdot \vecx_{\!c,j-1} + q \cdot \vecx_{\!c,j-1}) \\
     & = \!\!\!\!\!\!\sum_{i \in \mathfrak{R}_{d,c',\vecx_{\!c,j}}}\!\!\!\!\!\!\!\!\!\! \left( p_{1} \cdot a_i^{e_{1}} \cdot \gamma_{c',i} (\vecx_{\!c,j-1}) + \cdots + p_{k} \cdot a_i^{e_{k}} \cdot \gamma_{c',i} (\vecx_{\!c,j-1}) + q \cdot \gamma_{c',i} (\vecx_{\!c,j-1})\right) \tag{by \Cref{GammaAB}, as $\gamma_{c',i}$ is linear} \\
     & = (p_{1} \hat{a}_{c'}^{e_{1}} + \cdots +
    p_{k} \hat{a}_{c'}^{e_{k}} + q) \cdot
    \!\!\!\!\sum_{i \in \mathfrak{R}_{d,c',\vecx_{\!c,j}}}\!\!\!\!\!\!\!\!
    \gamma_{c',i}(\vecx_{\!c,j-1}) \tag{as $a_i = |a_i| = \hat{a}_{c'}$ for positive real
    eigenvalues $a_i$} \\
     & = (p_{1} \hat{a}_{c'}^{e_{1}} + \cdots +
    p_{k} \hat{a}_{c'}^{e_{k}} + q) \cdot \!\!\!\!\!\!\sum_{i \in \mathfrak{R}_{d,c',\vecx_{\!c,j-1}}}\!\!\!\!\!\!\!\!\!\!\!\!
    \gamma_{c',i}(\vecx_{\!c,j-1}) \tag{as $\mathfrak{R}_{d,c',\vecx_{\!c,j}} =
    \mathfrak{R}_{d,c',\vecx_{\!c,j-1}}$} \\
     & = (p_{1} \hat{a}_{c'}^{e_{1}} + \cdots +
    p_{k} \hat{a}_{c'}^{e_{k}} + q) \cdot
    w_{c'}(\vecx_{\!c,j-1}) \tag{$\dagger\dagger$} \label{eq:boosting_proof_daggerdagger} \\
     & \geq q \cdot w_{c'}(\vecx_{\!c,j-1}) \tag{$\dagger\dagger\dagger$} \label{eq:boosting_proof_daggerdaggerdagger}
  \end{align*}
  Here, to see that $\mathfrak{R}_{d,c',\vecx_{\!c,j}} = \mathfrak{R}_{d,c',\vecx_{\!c,j-1}}$, let $(\scaleinside,\scaleoutside) \in \mathcal{I}$ be such that $\constrainttermgroup_{(\scaleinside,\scaleoutside)} = \constrainttermgroup_{d,c',\vecx_{\!c,j-1}}$.
  Hence, for $(\scaleinside',\scaleoutside')$ with $(\scaleinside',\scaleoutside') >_{\mathrm{lex},d} (\scaleinside,\scaleoutside)$ we have $\constrainttermgroup_{(\scaleinside',\scaleoutside'),c',\vecx_{\!c,j-1}} = 0$.
  By \Cref{def:constraint_term_groups} and the definition of $\vecx_{\!c,j}$, we conclude $\constrainttermgroup_{(\scaleinside',\scaleoutside'),c',\vecx_{\!c,j}} = \emptyset$ as
  \begin{align*}
        & \sum_{i \in \constrainttermgroup_{(\scaleinside',\scaleoutside')}} \zeta_{i,\matA}^{|f|_{\symMat{A}}} \zeta_{i,\matB}^{|f|_{\symMat{B}}} \gamma_{c',i} (\vecx_{\!c',j}) \\
    ={} & \sum_{i=1}^{k} p_{i} \sum_{i'\in\constrainttermgroup_{(\scaleinside',\scaleoutside')}} \zeta_{i',\matA}^{|f|_{\symMat{A}} + e_{i}} \zeta_{i',\matB}^{|f|_{\symMat{B}}} \gamma_{c',i'} (\vecx_{\!c,j-1}) + q \sum_{i\in\constrainttermgroup_{(\scaleinside',\scaleoutside')}} \zeta_{i,\matA}^{|f|_{\symMat{A}} + e_{i}} \zeta_{i,\matB}^{|f|_{\symMat{B}}} \gamma_{c',i} (\vecx_{\!c,j-1})
  \end{align*}
  is $0$ for all $f\in\Path$ since by \Cref{def:constraint_term_groups}
  already all its (outer) addends are $0$.
  Moreover, $\constrainttermgroup_{(\scaleinside,\scaleoutside),c',\vecx_{\!c,j-1}} \neq \emptyset$ implies $\constrainttermgroup_{(\scaleinside,\scaleoutside),c',\vecx_{\!c,j}}
    \neq \emptyset$ as for $f\in\Path$ we have
  \begin{align*}
        & \sum_{i\in\constrainttermgroup_{(\scaleinside,\scaleoutside)}} \zeta_{1,\matA}^{|f|_{\symMat{A}}} \zeta_{2,\matB}^{|f|_{\symMat{B}}} \gamma_{c',i} (\vecx_{\!c,j}) \\
    ={} & \sum_{i\in\mathfrak{R}_{(\scaleinside,\scaleoutside)}} (p_{1} + \ldots + p_{k} + q) \gamma_{c',i} (\vecx_{\!c,j-1}) \\
        & {}+ \sum_{i\in\mathfrak{C}_{(\scaleinside,\scaleoutside)}} (p_{1} + \ldots + p_{k} + q) \zeta_{i,\matA}^{|f|_{\symMat{A}} + e_{i}} \zeta_{i,\matB}^{|f|_{\symMat{B}}} \gamma_{c',i} (\vecx_{\!c,j-1})
  \end{align*}
  where the left sum on the right-hand side (over $\mathfrak{R}_{(\scaleinside,\scaleoutside)}$) is positive and does not depend on $f$ whereas the right sum is either $0$ for all $f$, or will take different values depending on $f$ (see~\Cref{lem:braverman_generalisation}).
  This implies that the right-hand side cannot be $0$ for all $f$.
  Hence, $\constrainttermgroup_{d,c,\vecx_{\!c,j}} = \constrainttermgroup_{(\scaleinside,\scaleoutside)} = \constrainttermgroup_{d,c,\vecx_{\!c,j-1}}$ by \Cref{def:dominating_constraint_term_groups} and thus $\mathfrak{R}_{d,c,\vecx_{\!c,j}} = \mathfrak{R}_{d,c,\vecx_{\!c,j-1}}$.

  For the term on the right-hand side of \cref{boosting claim}, for all $(c',j') \leq_{\mathrm{lex}} (c,j)$ we have
  \begin{align*}
           & \sum_{\substack{\tilde{j}\in[j']}} \left| \gamma_{c',I_{c',\tilde{j}}} (\vecx_{\!c,j})
    \right| \\
    ={}    & \sum_{\substack{\tilde{j}\in[j']}} \Big| \gamma_{c',I_{c',\tilde{j}}} ((p_{1} \matA^{e_{1}} + \cdots + p_{k} \matA^{e_{k}})\cdot \vecx_{\!c,j-1} + q \cdot \vecx_{\!c,j-1}) \Big| \\
    ={}    & \sum_{\substack{\tilde{j}\in[j']}} \left| (p_{1} a_{I_{c',\tilde{j}}}^{e_{1}} + \cdots + p_{k} a_{I_{c',\tilde{j}}}^{e_{k}} + q) \cdot \gamma_{c',I_{c',\tilde{j}}} (\vecx_{\!c,j-1}) \right| \tag{as above} \\
    ={}    & \sum_{\substack{\tilde{j}\in[j']}} \Big| (p_{1} \hat{a}_{c'}^{e_{1}} \zeta_{I_{c',\tilde{j}},\matA}^{e_{1}}+ \cdots + p_{k} \hat{a}_{c'}^{e_{k}} \zeta_{I_{c',\tilde{j}},\matA}^{e_{k}}+ q) \cdot \gamma_{c',I_{c',\tilde{j}}} (\vecx_{\!c,j-1}) \Big| \tag{$a_{I_{c',\tilde{j}}} = |a_{I_{c',\tilde{j}}}| \zeta_{{I_{c',\tilde{j}}},\matA} = \hat{a}_{c'} \zeta_{I_{c',\tilde{j}},\matA}$ for $a_{I_{c',\tilde{j}}} \not \in \AA_{>0}$} \\
    \leq{} &
    (p_{1} \hat{a}_{c'}^{e_{1}} + \cdots + p_{k}
    \hat{a}_{c'}^{e_{k}} + q)
    \cdot \sum_{\substack{\tilde{j}\in[j']}} \left|
    \gamma_{c',I_{c',\tilde{j}}} (\vecx_{\!c,j-1})
    \right| \tag{$\ddag$} \label{eq:boosting_proof_ddag}
  \end{align*}
  where we used the triangle inequation in the last step.

  Now we first prove that $\vecx_{\!c,j}$ satisfies the claim \cref{boosting claim} for $(c',j')$ that are strictly smaller than $(c,j)$, i.e., we show that for all $(c',j') <_{\mathrm{lex}} (c,j)$ with $j' \neq 0$ we have
  \begin{equation}
    \label{boostingClaimIS1}
    \frac{j'}{l_{c'}} w_{c'}\left(\vecx_{\!c,j}\right) > \sum_{\tilde{j}\in [j']} \left|
    \gamma_{c',I_{c',\tilde{j}}} \left(\vecx_{\!c,j}\right) \right|
  \end{equation}
  To see this, note that by the induction hypothesis we have
  \[
    \frac{j'}{l_{c'}} w_{c'}\left(\vecx_{\!c,j-1}\right) > \sum_{\tilde{j}\in [j']} \left| \gamma_{c',I_{c',\tilde{j}}} \left(\vecx_{\!c,j-1}\right) \right|
  \]
  implying
  \begin{align*}
        & \frac{j'}{l_{c'}} (p_{1} \hat{a}_{c'}^{e_{1}} + \cdots + p_{k} \hat{a}_{c'}^{e_{k}} + q) w_{c'}\left(\vecx_{\!c,j-1}\right) \\
    >{} & (p_{1} \hat{a}_{c'}^{e_{1}} + \cdots + p_{k} \hat{a}_{c'}^{e_{k}} + q) \sum\limits_{\tilde{j}\in [j']} \left| \gamma_{c',I_{c',\tilde{j}}} \left(\vecx_{\!c,j-1}\right) \right|
  \end{align*}
  as $p_{1} \hat{a}_{c'}^{e_{1}} + \cdots + p_{k} \hat{a}_{c'}^{e_{k}} + q > 0$.
  By \cref{eq:boosting_proof_daggerdagger} this is equivalent to
  \begin{equation*}
    \frac{j'}{l_{c'}} w_{c'}\left(\vecx_{\!c,j}\right)
    > (p_{1} \hat{a}_{c'}^{e_{1}} + \cdots + p_{k} \hat{a}_{c'}^{e_{k}} + q) \sum_{\tilde{j}\in [j']} \left| \gamma_{c',I_{c',\tilde{j}}} \left(\vecx_{\!c,j-1}\right) \right|
  \end{equation*}
  which in turn by \cref{eq:boosting_proof_ddag} implies
  \[
    \frac{j'}{l_{c'}} w_{c'}\left(\vecx_{\!c,j}\right) > \sum_{\tilde{j}\in [j']} \left| \gamma_{c',I_{c',\tilde{j}}} \left(\vecx_{\!c,j}\right) \right| .
  \]

  Finally, we have to show that $\vecx_{\!c,j}$ also satisfies \cref{boosting claim} for $(c',j') = (c,j)$, i.e.,
  \begin{equation}
    \label{boostingClaimIS2}
    \frac{j}{l_{c}} w_{c}\left(\vecx_{\!c,j}\right) > \sum_{\tilde{j}\in [j]} \left|
    \gamma_{c,I_{c,\tilde{j}}} \left(\vecx_{\!c,j}\right) \right|
  \end{equation}
  To this end, note that
  \begin{align*}
           & \left|\gamma_{c,I_{c,j}} \left(\vecx_{\!c,j}\right)\right| \\
    ={}    & \left| (p_{1} \hat{a}_c^{e_{1}}\zeta_{I_{c,j},\matA}^{e_{1}} + \cdots + p_{k} \hat{a}_c^{e_{k}} \zeta_{I_{c,j},\matA}^{e_{k}} + q) \cdot \gamma_{c,I_{c,j}}(\vecx_{\!c,j-1})\right| \tag{as above} \\
    <{}    & q \cdot \frac{w_{c}\left(\vecx_{\!c,j-1}\right)}{l_{c} \left| \gamma_{c,I_{j}} (\vecx_{\!c,j-1}) \right|} \left|\gamma_{c,I_{c,j}} \left(\vecx_{\!c,j-1}\right)\right| \tag{by \cref{eq:boosting1}} \\
    ={}    & q \frac{w_{c}\left(\vecx_{\!c,j-1}\right)}{l_{c}} \\
    \leq{} & \frac{1}{l_{c}} w_{c} \left( \vecx_{\!c,j}\right) \tag{by \labelcref{eq:boosting_proof_daggerdaggerdagger} for $c' = c$}
  \end{align*}

  This proves \cref{boostingClaimIS2}, since
  \begin{align*}
         & \sum_{\tilde{j}\in[j]} \left| \gamma_{c,I_{c,\tilde{j}}} \left(\vecx_{\!c,j}\right) \right| \\
    ={}  & \left| \gamma_{c,I_{c,j}} \left(\vecx_{\!c,j}\right) \right| + \sum_{\tilde{j}\in[j-1]} \left| \gamma_{c,I_{c,\tilde{j}}} \left(\vecx_{\!c,j}\right) \right| \\
    <{}  & \frac{1}{l_{c}} w_{c} \left( \vecx_{\!c,j}\right) + \sum_{\tilde{j}\in[j-1]} \left| \gamma_{c,I_{c,\tilde{j}}} \left(\vecx_{\!c,j}\right) \right| \\
    \leq & \frac{1}{l_{c}} w_{c} \left( \vecx_{\!c,j}\right) + \frac{j-1}{l_{c}} w_{c} \left(\vecx_{\!c,j}\right) \tag{by \cref{boostingClaimIS1} for $j' = j-1$ and $c' = c$} \\
    ={}  & \frac{j}{l_{c}} w_{c} \left(\vecx_{\!c,j}\right)
  \end{align*}
\end{proof}

%% file: ent_witnesses_remaining_parts_app.tex
In order to complete the proof of \Cref{lem:soundness_of_witnesses} we have to show that \cref{eq:soundness_of_witnesses1} holds, which states $\E (T) = \infty$.
To that end, recall the definition \input{./ent_witnesses_definition_T.tex}of the $\F$-measurable random variable $T$.
It maps any run $\run$ to the smallest $i$ such that $\hat{f}$ extended by the first $i$ steps of $\run$ reaches an unsafe region.
We want to show $\E (T) = \infty$, i.e., that in expectation, an infinite number of steps would be needed to leave the safe region again after the path $\hat{f}$.
Here, $\hat{f}\in\Path$ with $|\hat{f}|\geq l$ and $\eps > 0$, $r\in\NN$, as well as $l \in \NN_{>0}$ are constants such that $0 < \eps \leq 1-p$, $\eps l \geq r+2$, $\URVar(\hat{f}) \cdot |\hat{f}| \geq r$, and $\URVar(\hat{f})\in[0,\eps]$.
Additionally, $|\URVar(\hat{f})\cdot |\hat{f}| - \frac{\eps |\hat{f}| + r}{2}| \leq 1$, i.e., $\hat{f}$ reaches the middle region which is between $1$ above and $1$ below the middle line $\frac{\eps \cdot |f| + r}{2}$.

We now extend $\hat{f}$ by $f'_1, f'_2, \ldots\in \{\symMat{A},\symMat{B}\}$ such that $\hat{f}f'_{1}\ldots f'_{i}$ always stays in the middle region.
To this end, one inductively defines $f'_{i} \in \{\symMat{A},\symMat{B}\}$ for all $i\in\NN_{>0}$ as
\[
  f'_{i} =
  \begin{cases}
    \symMat{A} & \text{if }\URVar (\hat{f}f'_{1}\ldots f'_{i-1}) \cdot
    |\hat{f}f'_{1}\ldots f'_{i-1}| \leq \frac{\eps \cdot (|\hat{f}| + i) + r}{2} \\
    \symMat{B} & \text{otherwise.}
  \end{cases}
\]
We write $\hat{f}_{i} = \hat{f}f'_{1}\ldots f'_{i}$ for all $i\in\NN$ from now on to simplify the notation.
We claim that for all $i\in\NN$ we have $|\URVar(\hat{f}_{i}) \cdot |\hat{f}_{i}| - \frac{\eps |\hat{f}_{i}| + r}{2}| \leq 1$.
Clearly, this holds for $i=0$ by the requirements on the finite path $\hat{f}$.

In the induction step, assume that this claim holds for a fixed $i\in\NN$.
First, we consider the case $\URVar(\hat{f}_{i}) \cdot |\hat{f}_{i}| > \frac{\eps (|\hat{f}_{i}| + 1) + r}{2}$.
Then we have $f'_{i+1} = \symMat{B}$ and thus, $|\hat{f}_{i+1}|_{\symMat{A}} = |\hat{f}_{i}|_{\symMat{A}}$.
This implies
\begin{align*}
  \URVar(\hat{f}_{i+1}) \cdot |\hat{f}_{i+1}| & = |\hat{f}_{i+1}|_{\symMat{A}} - p (i+1) \\
                                              & = |\hat{f}_{i}|_{\symMat{A}} - p \cdot i - p \\
                                              & = \URVar(\hat{f}_{i})\cdot |\hat{f}_{i}| - p > \frac{\eps |\hat{f}_{i+1}| + r}{2} - 1
\end{align*}
and $\frac{\eps |\hat{f}_{i+1}| + r}{2} + 1 > \frac{\eps|\hat{f}_{i}| + r}{2} + 1 \geq \URVar(\hat{f}_{i}) \cdot |\hat{f}_{i}| = \URVar(\hat{f}_{i+1}) \cdot |\hat{f}_{i+1}| + p > \URVar(\hat{f}_{i+1}) \cdot |\hat{f}_{i+1}|$.

Otherwise, we have the case $\URVar(\hat{f}_{i}) \cdot |\hat{f}_{i}| \leq \frac{\eps (|\hat{f}_{i}| + 1) + r}{2}$.
Then we have $f'_{i+1} = \symMat{A}$ and thus, $|\hat{f}_{i+1}|_{\symMat{A}} = |\hat{f}_{i}|_{\symMat{A}} +1$.
This implies
\begin{align*}
  \URVar(\hat{f}_{i+1}) \cdot |\hat{f}_{i+1}|
   & = |\hat{f}_{i+1}|_{\symMat{A}} - p (i+1) \\
   & = |\hat{f}_{i}|_{\symMat{A}} - p \cdot i + 1 -p \\
   & = \URVar(\hat{f}_{i}) \cdot |\hat{f}_{i}| + 1 -p \\
   & \leq \frac{\eps |\hat{f}_{i+1}| + r}{2} + 1-p < \frac{\eps |\hat{f}_{i+1}| + r}{2} + 1
\end{align*}
and $\URVar(\hat{f}_{i+1}) \cdot |\hat{f}_{i+1}| = \URVar(\hat{f}_{i}) \cdot |\hat{f}_{i}| + 1 -p \geq \frac{\eps |\hat{f}_i| +r}{2} - p > \frac{\eps |\hat{f}_{i+1}| +r}{2} - 1$.
The claim follows in both cases.

Note that $\URVar(\hat{f}_{i}) \cdot |\hat{f}_{i}|$ is always in the (upper) safe region, i.e., we have $r \leq \URVar(\hat{f}_{i}) \cdot |\hat{f}_{i}| \leq \eps |\hat{f}_{i}|$.
To see this, note that $\eps l \geq r+2$ implies
\begin{align*}
  \frac{\eps |\hat{f}_i| + r}{2} & \geq \frac{\eps \cdot l + r}{2} \tag{as $|\hat{f}_i| \geq |\hat{f}| \geq l$} \\
                                 & \geq \frac{2r + 2}{2} \tag{as $\eps l \geq r+2$} \\
                                 & \geq r + 1
\end{align*}
and thus, $\URVar(\hat{f}_{i}) \cdot |\hat{f}_{i}| \geq \frac{\eps |\hat{f}_i| + r}{2} -1 \geq r$.
Moreover, we have $\eps \cdot |\hat{f}_i| - r\geq \eps \cdot l - r \geq 2$ and thus $\frac{\eps \cdot |\hat{f}_i|}{2} - \frac{r}{2} \geq 1$, i.e., $\eps \cdot |\hat{f}_i| \geq \frac{\eps \cdot |\hat{f}_i| +r}{2} + 1 \geq \URVar(\hat{f}_{i}) \cdot |\hat{f}_{i}|$.

In order to prove \cref{eq:soundness_of_witnesses1} we redefine the random variable $T\colon \Runs \to \NN\cup\{\infty\}$ as a function $T\colon \Path \to \Runs \to \NN\cup\{\infty\}$ mapping finite paths $f\in\Path$ to random variables $\Runs \to \NN\cup\{\infty\}$ in the following way:
\begin{align*}
  T(f) & = \min \{i\in\NN_{>0} \mid \URVar(f\run_{1}\ldots\run_{i}) \cdot (|f| + i) < r \text{ or } \URVar(f\run_{1}\ldots\run_{i}) \not \in [0,\eps]\} \\
       & = \min \{i \in \NN_{>0} \mid X_{i}(f) < 0 \text{ or } X_{i}(f) > \eps (|f|+i) - r\}
\end{align*}
where
\[
  X_{t}(f) = \URVar(f\run_{1}\ldots\run_{t}) \cdot |f\run_{1}\ldots\run_{t}| - r
\]
for all $t\in\NN_{>0}$.
So $X_t: \Path \to \Runs \to \AA$, where $X_t(f)$ maps any run $\run$ to the result of $\URVar(f') \cdot |f'| - r$, where $f'$ extends $f$ by the first $t$ steps of the run $\run$.
The equation \cref{eq:soundness_of_witnesses1} then corresponds to the statement $\E \left(T\left(\hat{f}\right)\right) = \infty$.

In the following, we show that for every $f\in\Path$, the stochastic process ${(X_{t}(f))}_{t=0}^{\infty}$ is a martingale w.r.t.\ its natural filtration $\F_{0}
  \subseteq \F_{1} \subseteq \cdots$, i.e.,
\begin{equation*}
  \F_{t} = \sigma \left( X_{0}(f), X_{1}(f), \ldots, X_{t}(f) \right) = \sigma\left( \left\{\Pre_{f'} \mid f'\in\Path,\, |f'| \leq t\right\} \right)
\end{equation*}
for all $t\in\NN$, according to \Cref{def:martingale}.
Note that $\F_{t}$ is independent from the chosen $f \in \Path$.
To this end, one has to show that the process ${(X_{t}(f))}_{t=0}^{\infty}$ is adapted to the filtration, it is integrable, and the conditional expectation of $X_{t+1}(f)$ given the $\sigma$-field $\F_{t}$ equals $X_{t}(f)$.
To prove this, we fix an arbitrary $f\in\Path$.
Clearly, ${(X_{t}(f))}_{t=0}^{\infty}$ is adapted to ${(\F_{t})}_{t=0}^{\infty}$, i.e., every $X_{t}$ is $\F_{t}$-measurable.
Moreover, the process is integrable, i.e., for all $t \in \NN$ we have $\E(|X_t(f)|) < \infty$.
The reason is that for all $\run \in \Runs$, we have $|X_{t}(f)(\run)| = |\URVar(f\run_{1}\ldots\run_{t}) \cdot |f\run_{1}\ldots\run_{t}| - r| \leq |\URVar(f\run_{1}\ldots\run_{t}) \cdot |f\run_{1}\ldots\run_{t}|| + r <|f\run_{1}\ldots\run_{t}| + r = |f| + t + r$.
Finally, to show that the conditional expectation of $X_{t+1} (f)$ given $\F_{t}$ coincides with $X_t(f)$, let $t\in\NN$ be arbitrary.\footnote{The conditional expectation $\E (X_{t+1} (f) \mid \F_{t})$ is the (with probability $1$) unique $\F_{t}$-measurable random variable with $\E(\E(X_{t+1} \mid \F_{t}) \ind_{A}) = \E (X_{t+1}(f) \ind_{A})$ for all $A \in \F_{t}$.}
Then,
\begin{align*}
  \E (X_{t+1} (f) \mid \F_{t})(\run) - X_{t} (f)
   & = \E (X_{t+1}(f) - X_{t}(f) \mid \F_{t})(\run) \tag{$X_t(f)$ is $\F_t$-measurable} \\
   & = \E \left(|\run_{t+1}|_{\symMat{A}} \mid \F_{t} \right)(\run) - p \\
   & = 0,
\end{align*}
for (almost) all $\run = \run_{1} \run_{2} \ldots \in \Runs$, following \cref{lem:atoms_and_conditional_expectations}, as $\Pre_{\run_{1}\run_{2}\ldots\run_{t}}$ is an atom of the $\sigma$-field $\F_{t}$ and hence, similar to the proof of \cref{lem:dual_positive_eigenvalues_for_eventually_dominating_constraints}, $\E(|\run_{t+1}|_{\symMat{A}} \mid \F_{t})(\run) \cdot \P(\Pre_{\run_{1}\run_{2}\ldots\run_{t}})= \E (|\run_{t+1}|_{\symMat{A}} \cdot \ind_{\Pre_{\run_{1}\run_{2}\ldots\run_{t}}}) = p \cdot \P(\Pre_{\run_{1}\run_{2}\ldots\run_{t}}) > 0$.
Therefore, $\E(|\run_{t+1}|_{\symMat{A}} \mid \F_{t})(\run) = p$.
This concludes the proof that ${(X_{t}(f))}_{t=0}^{\infty}$ is a martingale.

Moreover, as $\{T = t\} \in \F_{t}$ for all $t\in\NN$, the random variable $T(f)$ is a stopping time for ${(X_{t}(f))}_{t=0}^{\infty}$ for arbitrary $f\in\Path$.
Finally, note that ${(X_{t}(f))}_{t=0}^{\infty}$ is difference-bounded, i.e., we have $|X_{t+1}(f)(\run) - X_{t}(f)(\run)| < 1$ for all $t \in \NN$ and all $\run \in \Runs$.\footnote{This strong form of difference-boundedness will be needed later on for Hoeffding's inequation (\Cref{thm:azumahoeffdinginequality}).
For the optional stopping theorem (that we will also use later on), the weaker difference-boundedness property $\E(|X_{t+1}(f) - X_{t}(f)| \mid \F_{t}) \leq c$ for some constant $c$ would be sufficient.
Recall that for every martingale, we have $\E(X_{t+1}(f) - X_{t}(f) \mid \F_{t}) = 0$, but this does not necessarily imply the boundedness of $\E(|X_{t+1}(f) - X_{t}(f)| \mid \F_{t})$.}

Our goal is to prove $\E (T(\hat{f})) = \infty$, i.e., if we extend the path $\hat{f}$, then in expectation, we need infinitely many steps to leave the safe region again.
Our proof proceeds by contradiction.
Assume $\E (T(\hat{f})) < \infty$.
Note that for all $i\in\NN$ and all $\run \in\Pre_{\hat{f}_{i}}$, we have $T(\hat{f})(\run)\geq T(\hat{f}_{i})(\run)$.
The reason is that for $T(\hat{f})(\run)$, the first $i$ steps after $\hat{f}$ are also taken into account.
Moreover, the probability that a run extends $\hat{f}$ to $\hat{f}_{i}$ is not $0$, since $1 > p > 0$ (i.e., the probability for runs from $\Pre_{\hat{f}_{i}}$ is not 0).
Thus, $\E (T(\hat{f})) < \infty$ implies $\E (T(\hat{f}_{i})) \leq \E (T(\hat{f})) < \infty$.\footnote{Note that we only have $T(\hat{f}) (\run) \geq T(\hat{f}_{i}) (\run)$ for $\run \in \Pre_{\hat{f}_{i}}$ if $\run$ does not leave the safe region between $\hat{f}$ and $\hat{f}_{i}$.
  Here, this is the case as at these positions, $\run$ is always in the middle region which is contained in the safe region.
}

This directly implies $\P (T (\hat{f}_{i}) = \infty) = 0$ for all $i\in\NN$ which allows us to ignore runs $\run\in\Runs$ with $T (\hat{f}_{i})(\run) = \infty$ in the following.
As the expected stopping time $\E (T(\hat{f}_i))$ is finite, we can apply the optional stopping theorem as stated in \Cref{thm:optionalstopping}
to all processes ${(X_{t}(\hat{f}_{i}))}_{t=0}^{\infty}$ with stopping time $T(\hat{f}_{i})$.
(The reason is that the optional stopping theorem requires that the process is a martingale, that the expected stopping time is finite, and that the process is difference-bounded.) Therefore, we obtain
\begin{equation}
  \label{OST}
  \E\left(X_{0} (\hat{f}_{i})\right) = \E \left( X_{T(\hat{f}_{i})} (\hat{f}_{i}) \right)
\end{equation}
for all $i\in\NN$.
Here, $X_0(\hat{f}_{i})$ maps any run $\run$ to the number $\URVar(\hat{f}_{i}) \cdot |\hat{f}_{i}| - r$, i.e., it ignores its input $\run$ and hence, we have $\E\left(X_{0} (\hat{f}_{i})\right) = X_{0} (\hat{f}_{i}) = \URVar(\hat{f}_{i}) \cdot |\hat{f}_{i}| - r$.
So while the martingale property of $X_t$ implies $\E\left(X_{0} (\hat{f}_{i})\right) = \E\left(X_{t} (\hat{f}_{i})\right)$ for all $t \in \NN$, the optional stopping theorem yields the stronger property $\E\left(X_{0} (\hat{f}_{i})\right) = \E \left( X_{T(\hat{f}_{i})} (\hat{f}_{i}) \right)$, because in the right-hand side, the index of $X$ is a random variable instead of a constant, i.e., $X_{T(\hat{f}_{i})(\run)} (\hat{f}_{i})(\run)$ maps any run $\run$ to the value of $X_t(\hat{f}_{i})(\run)$ at the time $t = T(\hat{f}_{i})(\run)$.
To see the difference, consider the process ${(X_{t}(f_{\eps}))}_{t=0}^{\infty}$ where $f_{\eps}$ with $|f_{\eps}| = 0$ is the empty path and the stopping time $T' = \min \{i \in \NN_{>0} \mid X_{i}(f_{\eps}) > -r\}$.
Then, for all $t \in \NN$ we have $\E (X_{t} (f_{\eps})) = \E(X_{0}(f_{\eps})) = -r$ by the martingale property of ${(X_{t}(f_{\eps}))}_{t=0}^{\infty}$ but $\E(X_{T(f_{\eps})}(f_{\eps})) > -r$.
Note that the last inequation is well defined as $\P (T_{f} (\eps) < \infty) = 1$ which can be shown by a simple random walk argument.
The apparent contradiction to the optional stopping theorem is resolved by observing $\E (T_{f} (f_{\eps})) = \infty$.

We now show that \cref{OST} leads to a contradiction.
To this end, we show that $\E \left( X_{T(\hat{f}_{i})} (\hat{f}_{i}) \right)$ is bounded by a constant for $i \to \infty$.
In contrast, $\E\left(X_{0} (\hat{f}_{i})\right) = \URVar(\hat{f}_{i}) \cdot |\hat{f}_{i}| - r$ diverges for $i \to \infty$.
Thus, this shows that the optional stopping theorem cannot be applied here and hence, our assumption $\E (T(\hat{f})) < \infty$ must be wrong.

By only considering the positive contributions to the expectation $\E \left(X_{T(\hat{f}_{i})} (\hat{f}_{i})\right)$ (i.e., the runs where we leave the safe region via its upper border), we obtain:
\begin{equation*}
  \E\left(X_{0} (\hat{f}_{i})\right) = \E \left( X_{T(\hat{f}_{i})} (\hat{f}_{i}) \right)
  \leq \E \left( X_{T(\hat{f}_{i})}(\hat{f}_{i}) \cdot \ind_{\left\{ X_{T(\hat{f}_{i})} (\hat{f}_{i}) > 0 \right\}} \right)
\end{equation*}

We elaborate on the value of this expectation for $i\in\NN$:
\begin{align*}
  \E (X_{0} (\hat{f}_{i}))
   & \leq \E \left( X_{T(\hat{f}_{i})}(\hat{f}_{i}) \cdot \ind_{\left\{ X_{T(\hat{f}_{i})} (\hat{f}_{i}) > 0 \right\}} \right) \\
   & = \E \left( X_{T(\hat{f}_{i})}(\hat{f}_{i}) \cdot \ind_{ \left\{ X_{T(\hat{f}_{i})} (\hat{f}_{i}) \geq \eps (|\hat{f}_{i}| + T(\hat{f}_{i})) - r \right\} }\right ) \tag{Def. of $T (\hat{f}_{i})$} \\
   & \leq \E \left( \left(\eps (|\hat{f}_{i}| + T(\hat{f}_{i})) + 1\right) \cdot \ind_{ \left\{ X_{T(\hat{f}_{i})} (\hat{f}_{i}) \geq \eps (|\hat{f}_{i}| + T(\hat{f}_{i})) - r \right\} }\right )
\end{align*}
as for every $i \in \NN$ we have $X_{T(\hat{f}_{i}) - 1}(\hat{f}_{i}) \leq \eps (|\hat{f}_{i}| + T(\hat{f}_{i}) - 1) - r$ by definition of $T(\hat{f}_{i})$ and $X_{T(\hat{f}_{i})}(\hat{f_{i}}) \leq X_{T(\hat{f}_{i}) - 1}(\hat{f}_{i}) + 1$ by definition of $X_{t}$.
Hence, $X_{T(\hat{f}_{i})}(\hat{f}_{i}) \leq \eps(|\hat{f}_{i}| + T(\hat{f}_{i})) + 1$, i.e., $\eps(|\hat{f}_{i}| + T(\hat{f}_{i})) + 1$ is an upper bound on the value of the process when the run has left the safe region via its upper border for the first time.

We continue in order to express the expected value by probabilities.
Here, note that $T (\hat{f}_{i}) > 0$ by definition of $T$ (it looks for the smallest index $i\in\NN_{>0}$).
\begin{align*}
  \E (X_{0} (\hat{f}_{i}))
   & \leq \E \left( \left(\eps (|\hat{f}_{i}| + T(\hat{f}_{i})) + 1\right) \cdot \ind_{ \left\{ X_{T(\hat{f}_{i})} (\hat{f}_{i}) \geq \eps (|\hat{f}_{i}| + T(\hat{f}_{i})) - r \right\} }\right ) \\
   & = \E \Big( \sum_{t=1}^{\infty} \left(\eps (|\hat{f}_{i}| + t) +
    1\right) \cdot \ind_{ \left\{ X_{t} (\hat{f}_{i}) \geq \eps (|\hat{f}_{i}| + t) -
      r \right\} } \cdot \ind_{ \{T(\hat{f}_{i}) = t\} }\Big
  ) \tag{as $T (\hat{f}_{i}) > 0$} \\
   & = \sum_{t=1}^{\infty} \E \Big( \left(\eps (|\hat{f}_{i}| + t) +
    1\right) \cdot \ind_{ \left\{ X_{t} (\hat{f}_{i}) \geq \eps (|\hat{f}_{i}| + t) -
      r \right\}
    } \cdot \ind_{ \{T(\hat{f}_{i}) = t\}
  }\Big ) \tag{by Mon. Conv.} \\
   & = \sum_{t=1}^{\infty} \left(\eps (|\hat{f}_{i}| + t) + 1\right) \cdot \E \left(\ind_{ \left\{ X_{t} (\hat{f}_{i}) \geq \eps (|\hat{f}_{i}| + t) - r \right\} }\cdot \ind_{ \{T(\hat{f}_{i}) = t\} }\right ) \\
   & \leq \sum_{t=1}^{\infty} \left(\eps (|\hat{f}_{i}| + t) + 1\right) \cdot \E \left(\ind_{ \left\{ X_{t} (\hat{f}_{i}) \geq \eps (|\hat{f}_{i}| + t) - r \right\} } \right ) \\
   & = \sum_{t=1}^{\infty} \left(\eps (|\hat{f}_{i}| + t) + 1\right) \cdot \P \left( \left\{ X_{t} (\hat{f}_{i}) \geq \eps (|\hat{f}_{i}| + t) - r \right\} \right ) \\
   & = \sum_{t=1}^{\infty} \left(\eps (|\hat{f}_{i}| + t) + 1\right) \cdot \P \left( \left\{ X_{t} (\hat{f}_{i}) - X_{0} (\hat{f}_{i}) \geq \eps (|\hat{f}_{i}| + t) - r - X_{0}(\hat{f}_{i})\right\} \right )
\end{align*}
Recall that by construction of $\hat{f}_{i}$ we have $X_{0} (\hat{f}_{i}) = \URVar(\hat{f}_{i}) \cdot |\hat{f}_{i}| - r \leq \tfrac{\eps|\hat{f}_{i}| + r}{2} + 1 - r$ for all $i\in\NN$, and thus, $\eps (|\hat{f}_{i}| + t) - r - X_{0}(\hat{f}_{i}) \geq \eps (|\hat{f}_{i}| + t) - r - \tfrac{\eps|\hat{f}_{i}| + r}{2} - 1 + r = \eps \left(\frac{|\hat{f}_{i}|}{2} + t\right) - 1 - \frac{r}{2}$.

We continue our previous computation by under-appro\-xi\-mating the lower bound $\eps (|\hat{f}_{i}| + t) - r - X_{0}(\hat{f}_{i})$ using that $X_{0}(\hat{f}_{i}) \leq \frac{\eps|\hat{f}_{i}| + r}{2} + 1$.
\begin{equation}
  \begin{split}
    &\E (X_{0} (\hat{f}_{i})) \\
    \leq{}& \sum_{t=1}^{\infty} \left(\eps (|\hat{f}_{i}| + t) + 1\right) \cdot\P \left( \left\{ X_{t} (\hat{f}_{i}) - X_{0} (\hat{f}_{i}) \geq \eps \left(\frac{|\hat{f}_{i}|}{2} + t\right) - 1 - \frac{r}{2} \right\} \right) .
  \end{split}
  \tag{\dag}\label{eq:ent_witnesses_remaining_parts1}
\end{equation}

Next, one over-approximates the probability on the right-hand side of \cref{eq:ent_witnesses_remaining_parts1} by applying Hoeffding's inequation (see \Cref{thm:azumahoeffdinginequality}) to the difference-bounded martingale ${(X_{t}(\hat{f}_i))}_{t=0}^{\infty}$.
Here, $\exp(n)$ stands for $e^n$.
Hoeffding's inequation is the step in Line 2, which is applicable for all difference-bounded martingales.
Essentially, it allows us to conclude that for $|\hat{f}_i| \to \infty$ (i.e., for growing height of the safe region), the probability to leave the safe region via its upper border decreases exponentially.
\begin{align}
   & \P \left( \left\{ X_{t} (\hat{f}_{i}) - X_{0} (\hat{f}_{i}) \geq \eps \left(\frac{|\hat{f}_{i}|}{2} + t\right) - 1 - \frac{r}{2} \right\} \right) \nonumber \\
   & \leq \exp \left( \frac{- {\left( \eps \left(\frac{|\hat{f}_{i}|}{2} + t\right) - 1 - \frac{r}{2} \right)}^{2}}{2t} \right) \nonumber \\
   & = \exp \left( \frac{ - \eps^{2} {\left(\frac{|\hat{f}_{i}|}{2} + t\right)}^{2} + 2 \eps \left(\frac{|\hat{f}_{i}|}{2} + t\right) (1 + \frac{r}{2}) - {(1 + \frac{r}{2})}^{2} }{2t} \right) \nonumber \\
   & \leq \exp \left( \frac{ -\eps^{2} {\left(\frac{|\hat{f}_{i}|}{2} + t\right)}^{2} + 2 \eps \left(\frac{|\hat{f}_{i}|}{2} + t\right) (1 + \frac{r}{2}) }{2t} \right) \nonumber \\
   & = \exp \left( \frac{-\eps^{2} {\left(\frac{|\hat{f}_{i}|^{2}}{4} + |\hat{f}_{i}|t + t^{2}\right)} + 2 \eps \left(\frac{|\hat{f}_{i}|}{2} + t\right) (1 + \frac{r}{2}) }{2t} \right) \tag{\ddag} \label{eq:ent_witnesses_remaining_parts2}
\end{align}
We choose $I\in\NN$ large enough such that
\[
  \eps^2 \left(\frac{|\hat{f}_{i}|^{2}}{4} + |\hat{f}_{i}|t + t^{2}\right) \geq 4 \eps \cdot \left(\frac{|\hat{f}_{i}|}{2} + t\right) \cdot \left(1+\frac{r}{2}\right)
\]
for all $i\in\NN_{\geq I}$ and all $t\in\NN_{>0}$.
We have
\[
  \eps^2 \left(\frac{|\hat{f}_{i}|^{2}}{4} + |\hat{f}_{i}|t + t^{2}\right) \geq 4 \eps \cdot \left(\frac{|\hat{f}_{i}|}{2} + t\right) \cdot \left(1+\frac{r}{2}\right)
\]
which is equivalent to
\[
  -\frac{1}{2} \eps^2 \left(\frac{|\hat{f}_{i}|^{2}}{4} + |\hat{f}_{i}|t + t^{2}\right) \leq -2 \eps \cdot \left(\frac{|\hat{f}_{i}|}{2} + t\right) \cdot \left(1+\frac{r}{2}\right)
\]
which in turn holds iff
\begin{align*}
         & -\eps^2 \left(\frac{|\hat{f}_{i}|^{2}}{4} + |\hat{f}_{i}|t + t^{2}\right) + 2 \eps \cdot \left(\frac{|\hat{f}_{i}|}{2} + t\right) \cdot \left(1+\frac{r}{2}\right) \\
  \leq{} & -\frac{1}{2} \eps^2 \left(\frac{|\hat{f}_{i}|^{2}}{4} + |\hat{f}_{i}|t + t^{2}\right).
\end{align*}
We restrict ourselves to $i \geq I$ and continue our previous computation from inequation \cref{eq:ent_witnesses_remaining_parts2}.
\begin{align*}
         & \P \left( \left\{ X_{t} (\hat{f}_{i}) - X_{0} (\hat{f}_{i}) \geq \eps \left(\frac{|\hat{f}_{i}|}{2} + t\right) - 1 - \frac{r}{2} \right\} \right) \\
  \leq{} & \exp \left( \frac{- \frac{1}{2}\eps^{2} {\left(\frac{|\hat{f}_{i}|^{2}}{4} + |\hat{f}_{i}|t + t^{2}\right)} }{2t} \right) \\
  ={}    & \exp \left( \frac{- \eps^{2} {\left(\frac{|\hat{f}_{i}|^{2}}{4} + |\hat{f}_{i}|t + t^{2}\right)} }{4t} \right) \\
  \leq{} & \exp \left( \frac{-\eps^{2} {\left(|\hat{f}_{i}|t + t^{2}\right)}
  }{4t} \right) \tag{as $\frac{|\hat{f}_{i}|^{2}}{4} > 0$} \\
  ={}    & \exp \left( - \frac{\eps^{2}}{4} (|\hat{f}_{i}| + t) \right )
\end{align*}
By inserting this last inequation into \cref{eq:ent_witnesses_remaining_parts1}, one obtains
\[
  \E (X_{0} (\hat{f}_{i})) \leq \sum_{t=1}^{\infty}
  \frac{\eps (|\hat{f}_{i}| + t)
    + 1}{e^{\frac{\eps^{2}(|\hat{f}_{i}|+t)}{4}}} .
\]
Hence,
\begin{equation*}
  \lim_{i\to\infty} \E \left( X_{0} (\hat{f}_{i})\right) \leq \lim_{i\to\infty} \sum_{t=1}^{\infty} \frac{\eps (|\hat{f}_{i}| + t) + 1}{e^{\frac{\eps^{2}(|\hat{f}_{i}|+t)}{4}}} = \lim_{i\to\infty} \sum_{t=1}^{\infty} \frac{\eps (|\hat{f}| + i + t) + 1}{e^{\frac{\eps^{2}(|\hat{f}|+{i} + t)}{4}}} < \infty .
\end{equation*}
The reason that the last limit is finite is that the nominator $\eps (|\hat{f}| + i + t) + 1$ grows linearly for $i \to \infty$, whereas the denominator $e^{\frac{\eps^{2}(|\hat{f}|+{i}+t)}{4}}$ grows exponentially.
Moreover, $t$ is also linear in the nominator and exponential in the denominator (this is needed because otherwise, the sum from $t = 1$ to $\infty$ would not converge).

On the other hand, recall that we have $X_{0} (\hat{f}_{i}) = \URVar(\hat{f}_{i}) \cdot |\hat{f}_{i}| - r \geq \tfrac{\eps|\hat{f}_{i}| + r}{2} - 1 - r = \tfrac{\eps|\hat{f}_{i}| - r}{2} - 1$.
Hence,
\[
  \lim_{i\to\infty} \E\left(X_{0} (\hat{f}_{i})\right) \geq \lim_{i\to\infty}
  \left(\frac{\eps |\hat{f}_{i}| - r}{2} - 1\right) = \infty .
\]
Thus, we have a contradiction.
Hence, our assumption that $\E (T(\hat{f})) < \infty$ must have been wrong which in turn implies $\E (T(\hat{f})) = \infty$.

%% file: deciding_past_app.tex
\proofsForSection{sect:Deciding PAST}
\label{app:Deciding PAST}

{\renewcommand{\footnote}[1]{}
  \semialgebraicwitnesssetforalgebraicloops*{}
}
\begin{proof}
  We first show $W_{d} = \biguplus_{\mathfrak{c} \in \mathcal{I}^{m}} W_{d,\mathfrak{c}}$ for all $d\in\{\nmax,\pmax\}$.
  The direction ``$\supseteq$'' is trivial so we only show ``$\subseteq$''.
  To that end, let $\vecx \in W_{d}$.
  Then, by \Cref{def:witnesses_for_ent}, for every $c \in [m]$ there is a $(\scaleinside_{c},\scaleoutside_{c})$ with $\constrainttermgroup_{d,c,\vecx}
    = \constrainttermgroup_{(\scaleinside_{c},\scaleoutside_{c}),c,\vecx}
    \neq \emptyset$. The last inequation holds as $\vecx \in W_{d}$. Hence, $\vecx \in W_{d,\mathfrak{c}}$ for $\mathfrak{c} = ((\scaleinside_{1},\scaleoutside_{1}),\ldots,(\scaleinside_{m},\scaleoutside_{m}))$.

  Let $d\in\{\nmax,\pmax\}$, $\mathfrak{c} = ((\scaleinside_{1},\scaleoutside_{1}),\ldots,(\scaleinside_{m},\scaleoutside_{m})) \in \mathcal{I}^{m}$, and let $R$ denote the right-hand side of Equation \cref{eq:witness_set_is_semialgebraic1}.

  For \cref{eq:witness_set_is_semialgebraic1}, to prove ``$\subseteq$'', by \Cref{def:witnesses_for_ent} for all $\vecx \in W_{d,\mathfrak{c}} \subseteq W_{d}$ and all $c\in[m]$ we have
  \begin{equation}
    \begin{split}
      \sum_{i\in\mathfrak{R}_{d,c,\vecx}} \gamma_{c,i} (\vecx) = \sum_{i\in\mathfrak{R}_{(\scaleinside_{c},\scaleoutside_{c}),c,\vecx}} \gamma_{c,i} (\vecx)
      &> \sum_{i\in\mathfrak{C}_{d,c,\vecx}} |\gamma_{c,i} (\vecx)| \\
      &= \sum_{i\in\mathfrak{C}_{(\scaleinside_{c},\scaleoutside_{c}),c,\vecx}} |\gamma_{c,i} (\vecx)|.
    \end{split}
    \label{lhsSemialgebraic}
  \end{equation}
  Thus, \cref{lhsSemialgebraic} implies $\vecx \in \CGTZero_{c,(\scaleinside_{c},\scaleoutside_{c})}$ for all $c\in[m]$.
  For all $c\in[m]$ and all $(\scaleinside,\scaleoutside)$ with $(\scaleinside,\scaleoutside) >_{\mathrm{lex},d} (\scaleinside_{c},\scaleoutside_{c})$ we have $\sum_{i\in\constrainttermgroup_{(\scaleinside,\scaleoutside)}} \zeta_{i,\matA}^{|f|_{\symMat{A}}} \zeta_{i,\matB}^{|f|_{\symMat{B}}} \gamma_{i}(\vecx) = 0$ for all $f \in \Path$ by \Cref{def:constraint_term_groups,def:dominating_constraint_term_groups}.
  Clearly, $\sum_{i\in\constrainttermgroup_{(\scaleinside,\scaleoutside)} \cap \mathfrak{R}} \zeta_{i,\matA}^{|f|_{\symMat{A}}} \zeta_{i,\matB}^{|f|_{\symMat{B}}} \gamma_{i}(\vecx)$ is always constant for all $f$.
  In contrast, by \Cref{lem:braverman_generalisation}\ref{it:braverman_generalisation_itemB}, $\sum_{i\in\constrainttermgroup_{(\scaleinside,\scaleoutside)} \cap \mathfrak{C}} \zeta_{i,\matA}^{|f|_{\symMat{A}}}
    \zeta_{i,\matB}^{|f|_{\symMat{B}}} \gamma_{i}(\vecx)$ is only constant if it is the constant $0$. Thus, both sums $\sum_{i\in\constrainttermgroup_{(\scaleinside,\scaleoutside)}
      \cap \mathfrak{R}} \zeta_{i,\matA}^{|f|_{\symMat{A}}}
    \zeta_{i,\matB}^{|f|_{\symMat{B}}} \gamma_{i}(\vecx)$ and $\sum_{i\in\constrainttermgroup_{(\scaleinside,\scaleoutside)} \cap \mathfrak{C}} \zeta_{i,\matA}^{|f|_{\symMat{A}}}
    \zeta_{i,\matB}^{|f|_{\symMat{B}}} \gamma_{i}(\vecx)$ must be constant 0 and therefore, $\vecx \in \CEqZero_{c,(\scaleinside,\scaleoutside)}$. Hence, $\vecx \in R$.

  To prove ``$\supseteq$'' for \cref{eq:witness_set_is_semialgebraic1}, let $\vecx \in R \subseteq \AA^{n}$.
  Then, for all $c\in[m]$ there must be a tuple $(\scaleinside_{c},\scaleoutside_{c}) \in \mathcal{I}$ such that $\vecx \in \CGTZero_{c,(\scaleinside_{c},\scaleoutside_{c})}$ and $\vecx \in \CEqZero_{c,(\scaleinside,\scaleoutside)}$ for every $(\scaleinside,\scaleoutside) \in \mathcal{I}$ with $(\scaleinside,\scaleoutside) >_{\mathrm{lex},d} (\scaleinside_{c},\scaleoutside_{c})$.
  Hence, for all $(\scaleinside,\scaleoutside) >_{\mathrm{lex},d}
    (\scaleinside_{c},\scaleoutside_{c})$ we have $\sum_{{i\in\constrainttermgroup_{(\scaleinside,\scaleoutside)}}}
    \zeta_{i,\matA}^{|f|_{\symMat{A}}}\zeta_{i,\matB}^{|f|_{\symMat{B}}}\gamma_{c,i}(\vecx) = 0$ for all $f \in \Path$, see \Cref{lem:braverman_generalisation}\ref{it:braverman_generalisation_itemA}. Thus, $\constrainttermgroup_{d,c,\vecx} = \constrainttermgroup_{(\scaleinside_{c},\scaleoutside_{c}),c,\vecx}$ and $\sum_{i \in \mathfrak{R}_{d,c,\vecx}} \gamma_{c,i} (\vecx) > \sum_{i\in\mathfrak{C}_{d,c,\vecx}} |\gamma_{c,i} (\vecx)|$ for all $c\in[m]$. Hence, $\vecx \in W_{d,\mathfrak{c}}$ for $\mathfrak{c} = ((\scaleinside_{1},\scaleoutside_{1}),\ldots,(\scaleinside_{m},\scaleoutside_{m}))$.

  Recall that $\gamma_{c,i} (\vecx) = {(\matC\matS)}_{c,i}{(\matS^{-1}\vecx)}_{i} \in \QQbar$ for all $(c,i,\vecx) \in [m]\times[n]\times\AA^{n}$ as $\matA,\matB,\matC$ only consist of algebraic entries, see \Cref{algebraicMatrices}.

  Hence, the sets $\CEqZero_{c,(\scaleinside,\scaleoutside)},\CGTZero_{c,(\scaleinside,\scaleoutside)}$ for all $(c,(\scaleinside,\scaleoutside))\in[m]\times\mathcal{I}$, as well as $W_{d,\mathfrak{c}}$ for $d\in\{\nmax,\pmax\},\, \mathfrak{c} \in \mathcal{I}^{m}$ and the set $W$ are semialgebraic, as finite unions and intersections of semialgebraic sets are semialgebraic again.
\end{proof}

\wdasfiniteunionofconvexsets*{}
\begin{proof}
  Let $d\in\{\nmax,\pmax\}$, $\mathfrak{c} \in \mathcal{I}^{m}$, and let $\vecz = t\vecx + (1-t) \vecy$ for arbitrary $\vecx,\vecy \in W_{d,\mathfrak{c}}$ and $t\in(0,1)$.
  Fix an arbitrary $c\in[m]$.
  By definition of $W_{d,\mathfrak{c}}$ we have $\constrainttermgroup_{d,c,\vecx} = \constrainttermgroup_{d,c,\vecy} = \constrainttermgroup_{(\scaleinside,\scaleoutside),c,\vecx} = \constrainttermgroup_{(\scaleinside,\scaleoutside),c,\vecy} \neq \emptyset$, $\mathfrak{R}_{d,c,\vecx} = \mathfrak{R}_{d,c,\vecy}$, and $\mathfrak{C}_{d,c,\vecx} = \mathfrak{C}_{d,c,\vecy}$ for some $(\scaleinside,\scaleoutside) \in \mathcal{I}$.
  Similar to the proof of \Cref{lem:boosting}, the linearity of $\gamma_{c,i}$ moreover implies $\mathfrak{R}_{d,c,\vecz} = \mathfrak{R}_{d,c,\vecx} = \mathfrak{R}_{d,c,\vecy}$ and $\mathfrak{C}_{d,c,\vecz} = \mathfrak{C}_{d,c,\vecx} = \mathfrak{C}_{d,c,\vecy}$.
  To see this, fix some $c\in[m]$ and first consider $(\scaleinside',\scaleoutside')$ with $(\scaleinside',\scaleoutside') >_{\mathrm{lex},d} (\scaleinside,\scaleoutside)$.
  By \Cref{def:constraint_term_groups}, we conclude $\constrainttermgroup_{(\scaleinside',\scaleoutside'),c,\vecz} = \emptyset$ as
  \begin{align*}
        & \sum_{i \in \constrainttermgroup_{(\scaleinside',\scaleoutside')}} \zeta_{i,\matA}^{|f|_{\symMat{A}}} \zeta_{i,\matB}^{|f|_{\symMat{B}}} \gamma_{c,i} (\vecz) \\
    ={} & t \sum_{i \in \constrainttermgroup_{(\scaleinside',\scaleoutside')}} \zeta_{i,\matA}^{|f|_{\symMat{A}}} \zeta_{i,\matB}^{|f|_{\symMat{B}}} \gamma_{c,i} (\vecx) + (1-t)\sum_{i \in \constrainttermgroup_{(\scaleinside',\scaleoutside')}} \zeta_{i,\matA}^{|f|_{\symMat{A}}} \zeta_{i,\matB}^{|f|_{\symMat{B}}} \gamma_{c,i} (\vecy)
  \end{align*}
  is $0$ for all $f\in\Path$ since by \Cref{def:constraint_term_groups} both addends are $0$.
  Furthermore, $\constrainttermgroup_{(\scaleinside,\scaleoutside),c,\vecx} = \constrainttermgroup_{(\scaleinside,\scaleoutside),c,\vecy} \neq \emptyset$ implies $\constrainttermgroup_{(\scaleinside,\scaleoutside),c,\vecz}
    \neq \emptyset$ as for $f\in\Path$ we have
  \begin{align*}
        & \sum_{i\in\constrainttermgroup_{(\scaleinside,\scaleoutside)}} \zeta_{1,\matA}^{|f|_{\symMat{A}}} \zeta_{2,\matB}^{|f|_{\symMat{B}}} \gamma_{c,i} (\vecz) \\
    ={} & t \sum_{i\in\mathfrak{R}_{(\scaleinside,\scaleoutside)}} \gamma_{c,i} (\vecx_{c,j-1}) + (1-t) \sum_{i\in\mathfrak{R}_{(\scaleinside,\scaleoutside)}} \gamma_{c,i} (\vecx_{c,j-1}) \\
        & {}+ \sum_{i\in\mathfrak{C}_{(\scaleinside,\scaleoutside)}} \zeta_{i,\matA}^{|f|_{\symMat{A}} + e_{i}} \zeta_{i,\matB}^{|f|_{\symMat{B}}} (t \gamma_{c,i} (\vecx) + (1-t) \gamma_{c,i} (\vecy))
  \end{align*}
  where the two leftmost sums ``$\sum_{i\in\mathfrak{R}_{(\scaleinside,\scaleoutside)}}
    \ldots$'' on the right-hand side are positive and do not depend on $f$ whereas the rightmost sum is either $0$ for all $f$, or will take different values depending on $f$ (see~\Cref{lem:braverman_generalisation}). This implies that the right-hand side cannot be $0$ for all $f$. Hence, $\constrainttermgroup_{d,c,\vecz} = \constrainttermgroup_{(\scaleinside,\scaleoutside)} = \constrainttermgroup_{d,c,\vecx} = \constrainttermgroup_{d,c,\vecy} \neq \emptyset$ by \Cref{def:dominating_constraint_term_groups}.

  Moreover,
  \begin{align*}
    \sum_{i\in\mathfrak{R}_{d,c,\vecz}} \gamma_{c,i} (\vecz)
     & = \sum_{i\in\mathfrak{R}_{d,c,\vecz}} (t \gamma_{c,i} (\vecx) + (1-t) \gamma_{c,i} (\vecy)) \\
     & = t \sum_{i\in\mathfrak{R}_{d,c,\vecx}} \gamma_{c,i} (\vecx) + (1-t) \sum_{i\in\mathfrak{R}_{d,c,\vecy}} \gamma_{c,i} (\vecy) \\
     & > t\sum_{i\in\mathfrak{C}_{d,c,\vecx}} \left|\gamma_{c,i} (\vecx)\right| + (1-t)\sum_{i\in\mathfrak{C}_{d,c,\vecy}} \left|\gamma_{c,i} (\vecy)\right| \tag{$\vecx,\vecy \in W_{d,\mathfrak{c}} \subseteq W_{d}$} \\
     & \geq \sum_{i\in\mathfrak{C}_{d,c,\vecz}} \left|\gamma_{c,i} (\vecz)\right| \tag{Triang.\ Ineq.}
  \end{align*}
  and thus $\vecz \in W_{d}$.
  Combining both statements yields $\vecz \in W_{d,\mathfrak{c}}$ and therefore concludes the proof.
\end{proof}

\computingwitnessesfornontermination*{}
\begin{proof}
  As the loop is non-terminating, by \cref{thm:deciding_past}, we can compute a corresponding witness $\vecx \in W\cap\semiring^{n} \subseteq \ENT \cap \semiring^{n}$ for eventual non-termination.
  Since $\vecx \in W$, we have $\vecx \in W_{d,\mathfrak{c}}$ for some $d\in\{\nmax,\pmax\}$ and $\mathfrak{c} = ((\scaleinside_{1},\scaleoutside_{1}),\ldots,(\scaleinside_{m},\scaleoutside_{m})) \in \mathcal{I}^{m}$.
  Note that $\vecx \in W_{d,\mathfrak{c}}$ implies $\constrainttermgroup_{(\scaleinside_{c},\scaleoutside_{c})} \cap \mathfrak{R} = \mathfrak{R}_{d,c,\vecx}$ and $\constrainttermgroup_{(\scaleinside_{c},\scaleoutside_{c})} \cap \mathfrak{C} = \mathfrak{C}_{d,c,\vecx}$.
  Therefore, for $v(f)$ from \cref{lem:domination_of_eventually_dominating_constraint_term_groups}, there is a constant $\rho$ such that
  \begin{align*}
    v(f) & = \sum_{i \in \constrainttermgroup_{d,c,\vecx}} \zeta_{i,\matA}^{|f|_{\symMat{A}}} \zeta_{i,\matB}^{|f|_{\symMat{B}}} \gamma_{c,i} (\vecx) \\
         & > \sum_{i \in \constrainttermgroup_{(\scaleinside_{c},\scaleoutside_{c})} \cap \mathfrak{R}} \gamma_{c,i} (\vecx) - \sum_{i \in \constrainttermgroup_{(\scaleinside_{c},\scaleoutside_{c})} \cap \mathfrak{C}} |\gamma_{c,i} (\vecx)| > \rho > 0
  \end{align*}
  holds for all $c\in [m]$ and all $f\in\Path$ by definition as $\vecx \in W_{d,\mathfrak{c}}$.
  From this, one computes the constants $\eps_{c}$, $r_{c}$, and $l_{c}$ from the proof of \Cref{lem:soundness_of_witnesses} for every $c \in [m]$ according to the (constructive) proof of \Cref{lem:domination_of_eventually_dominating_constraint_term_groups}.
  From these constants, this time following the proof of \Cref{lem:soundness_of_witnesses}, one obtains constants $\eps$, $r$, and $l$ such that $\Val_{\vecx} (f) > \veczero$ for all $f\in\Path$ with $|f| \geq l$, $\URVar(f) \cdot |f| \geq r$, and $\URVar(f) \in [0,\eps]$.
  Among these, one chooses a suitable $\hat{f}$, see \cref{lem:soundness_of_witnesses}, and computes $\vecy = \matA^{|\hat{f}|_{\symMat{A}}} \matB^{|\hat{f}|_{\symMat{B}}}\vecx$.
  Then, from the proof of \Cref{lem:soundness_of_witnesses} one concludes $\vecy \in \NT \cap \semiring^{n}$, i.e., $\vecy$ is a witness for non-termination.
\end{proof}